\documentclass{article}
\usepackage{amssymb}
\usepackage{amsmath}
\usepackage{amsfonts}

\newtheorem{theorem}{Theorem}

\newtheorem{definition}[theorem]{Definition}

\newtheorem{notation}[theorem]{Notation}

\newtheorem{proposition}[theorem]{Proposition}

\newenvironment{proof}[1][Proof]{\noindent\textbf{#1.} }{\ \rule{0.5em}{0.5em}}
\input{tcilatex}
\begin{document}

\title{\textbf{Supersymmetry Flows, Semi-Symmetric Space Sine-Gordon Models
And The Pohlmeyer Reduction}}
\author{David M. Schmidtt \\
\\
\textit{Instituto de F\'{\i}sica Te\'{o}rica - IFT/UNESP}\\
\textit{Rua Dr. Bento Teobaldo Ferraz, 271, Bloco II}\\
\textit{CEP 01140-070, S\~{a}o Paulo, SP, Brasil.}}
\maketitle

\begin{abstract}
We study the extended supersymmetric integrable hierarchy underlying the
Pohlmeyer reduction of superstring sigma models on semi-symmetric
superspaces $F/G.$ This integrable hierarchy is constructed by coupling two
copies of the homogeneous integrable hierarchy associated to the loop Lie
superalgebra extension $\widehat{\mathfrak{f}}$ of the Lie superalgebra $%
\mathfrak{f}$ of $F$ and this is done by means of the algebraic dressing
technique and a Riemann-Hilbert factorization problem. By using the
Drinfeld-Sokolov procedure we construct explicitly, a set of 2D spin $\pm
1/2 $ conserved supercharges generating supersymmetry flows in the phase
space of the reduced model. We introduce the bi-Hamiltonian structure of the
extended homogeneous hierarchy and show that the two brackets are of the
Kostant-Kirillov type on the co-adjoint orbits defined by the light-cone Lax
operators $L_{\pm }$. By using the second symplectic structure, we show that
these supersymmetries are Hamiltonian flows, we compute part of the
supercharge algebra and find the supersymmetric field variations they
induce. We also show that this second Poisson structure coincides with the
canonical Lorentz-invariant symplectic structure of the WZNW model involved
in the Lagrangian formulation of the extended integrable hierarchy, namely,
the semi-symmetric space sine-Gordon model (SSSSG), which is the Pohlmeyer
reduced action functional for the transverse degrees of freedom of
superstring sigma models on the cosets $F/G.$ We work out in some detail the
Pohlmeyer reduction of the $AdS_{2}\times S^{2}$ and the $AdS_{3}\times
S^{3} $ superstrings and show that the new conserved supercharges can be
related to the supercharges extracted from 2D superspace. In particular, for
the $AdS_{2}\times S^{2}$ example, they are formally the same.
\end{abstract}

\tableofcontents

\newpage

\section{Introduction.}

Recently \cite{tseytlin}, Grigoriev and Tseytlin motivated by a desired to
find a useful 2D Lorentz-invariant reformulation of the classical-integrable 
$AdS_{5}\times S^{5}$ Green-Schwarz (GS) superstring world-sheet theory in
terms of physical/transverse degrees of freedom only, constructed the
Pohlmeyer reduced version of the $AdS_{5}\times S^{5}$ coset sigma model
action. The corresponding reduced Lagrangian is of a non-Abelian Toda type:
a gauged WZNW model with an integrable potential coupled to a set of 2D
fermionic fields. Some of the main features of the reduced action is that it
is Lorentz-invariant, the small-fluctuation spectrum near the trivial vacuum
has the same number of bosonic and fermionic degrees of freedom and its
integrable structure is equivalent to that of the initial sigma model. The
structure of the reduced action suggest the presence of 2D supersymmetry, a
fact that was confirmed explicitly for the simplest case of the sigma model
on $AdS_{2}\times S^{2},$ which turned out to be equivalent to the $(2,2)$
supersymmetric extension of the sine-Gordon model \cite{tseytlin}.

In \cite{tseytlin II} were discussed the possible existence of hidden 2D
supersymmetry in the first non-trivial reduced model which corresponds to
the GS sigma model in the $AdS_{3}\times S^{3}$ background. The reduced
action seems to be a $(2,2)$ supersymmetric extension of the complex
sine-Gordon coupled in a non-trivial way with its hyperbolic counterpart,
i.e the $(2,2)$ complex sinh-Gordon model, but its explicit superspace
structure could not been identified as was done in the $AdS_{2}\times S^{2}$
case.

The Lie algebraic structure behind the Pohlmeyer reduction goes beyond the $%
AdS_{5}\times S^{5}$ case and is common to other sigma models and such a
reduction can be performed, in principle and without major complications, on
any GS superstring sigma model on a semi-symmetric superspace $F/G$, in
which the Lie algebra $\mathfrak{g}$ of $G$ is the zero locus of a $%
\mathbb{Z}
_{4}$ automorphism of the Lie superalgebra $\mathfrak{f}$ of $F.$ However,
despite of the simplicity for constructing the reduced models, there are
still a number of open problems yet to be solved at the classical level, see
for instance \cite{tseytlin},\cite{tseytlin II}. Among them and the one we
are most interested here is related to the conjectured existence of
world-sheet supersymmetry in the reduced models \cite{tseytlin}, which have
resisted to go beyond the simplest case $AdS_{2}\times S^{2}$ and remains as
a non-trivial open question. It would be surprising to find it because of
the initial sigma model is of GS type anyway.

In the present contribution, we start to study this question from the point
of view of integrable systems, providing some evidence supporting such a
conjecture. The strategy will be to identify the integrable structure behind
the Pohlmeyer reduction process and use it to identify the would-be 2D
world-sheet supersymmetry with the fermionic symmetry flows already present
in the underlying integrable hierarchy. The outcome is that the 2D
supersymmetry is associated to a special loop superalgebra $\widehat{%
\mathfrak{f}}^{\perp }\subset \widehat{\mathfrak{f}}$ constructed out of a
subalgebra $\mathfrak{f}^{\perp }\subset \mathfrak{f}$ by means of the
dressing flow transformations.

The outline of the paper is as follows. The chapter 2 is the main one and
includes the general results. In section 2.1, we introduce the extended
homogeneous hierarchy by using the dressing group and a Riemann-Hilbert
factorization problem. The hierarchy is defined in terms of three
gradations: the homogeneous gradation, associated to the loop extension $%
\widehat{\mathfrak{f}}$ of the superalgebra $\mathfrak{f,}$ the natural $%
\mathbb{Z}
_{4}$ gradation of $\mathfrak{f,}$ responsible for the matching of the
physical degrees of freedom in the reduced model and a $%
\mathbb{Z}
_{2}$ gradation, responsible for the consistency of the symmetry flows
induced by a special sub-superalgebra $\mathfrak{f}^{\perp }$ of $\mathfrak{%
f.}$ It is also shown how the usual gauge transformations can be interpreted
as the lowest symmetry flows of the extended integrable hierarchy. In
section 2.2, we define the relativistic sector of the hierarchy, which
provides the Lax operators governing the Pohlmeyer reduced models. In
section 2.3, we make a first tentative to introduce the 2D supersymmetry
flows of the hierarchy, where some obstructions related to the locality of
the gauge group are mentioned. In section 2.4, we make a review of the
Lagrangian formulation of the relativistic sector of the hierarchy, i.e we
introduce the semi-symmetric space sine-Gordon model (SSSSG), which is the
action functional for the physical degrees of freedom in the reduction of GS
sigma models. In section 2.5, we use the Drinfeld-Sokolov procedure to
construct explicitly a set of $dim$ $\mathfrak{f}_{1,3}^{\perp }$ 2D spin $%
\pm 1/2$ conserved supercharges associated to the fermionic symmetry flows
generated by the odd elements of $\mathfrak{f}_{1,3}^{\perp }$ of $\mathfrak{%
f}^{\perp }.$ In section 2.6, we show that the extended homogeneous
hierarchy is bi-Hamiltonian in which the two symplectic structures take the
form of Kostant-Kirillov brackets on the co-adjoint orbits defined by the
Lax operators $L_{\pm }$ introduced in section 2.2. In section 2.7, by using
the second bracket, we compute part of the supercharge algebra, deduce the
poisson form of the supersymmetry flow variations for the fields showing
that they are hamiltonian flows in the Pohlmeyer reduced phase space and
also mention on a subtlety related to the presence of the gauge group in the
supercharge algebra. It is also shown that the second symplectic structure
is equivalent to the canonical symplectic structure of the WZNW model
involved in the construction of the SSSSG models. In chapter 3, we make a
fast review of the Pohlmeyer reduction process in order to show how
everything fits in the construction presented in chapter 2. In chapter 4 we
work out explicit examples with the aim of exploring, in a first
approximation, the relation between some well-known superspace results with
the supersymmetry flow approach we have adopted. Finally, we make the
concluding remarks and pose what will be done in the near future. There are
two appendices including some technical details used in the computations. We
have included some previous known results in parts of the body of the paper
with the aim of render it as self-contained as possible.

\section{General analysis.}

This is the main chapter and includes all the results of the paper. The idea
is to introduce and study the integrable supersymmetric hierarchy underlying
the Pohlmeyer reduction of superstring sigma models. The most important
result is the explicit construction of the supercharges generating 2D
Hamiltonian fermionic symmetry flows on the phase space of the reduced
models, see (\ref{AKNS supercharges}) below.

\subsection{The extended homogeneous hierarchy.}

Here we show how to locate gauge symmetries in the context of the algebraic
dressing technique. We refine the results of \cite{susyflows-mKdV} in order
address later the situation we are most interested, namely the Pohlmeyer
reduction of superstring sigma models.

Start by considering a finite dimensional Lie superalgebra $\mathfrak{f}$
endowed with an order four linear automorphism $\Omega ,$ $\Omega :\mathfrak{%
f\rightarrow f,}$ $\Omega \left( \left[ X,Y\right] \right) =\left[ \Omega
\left( X\right) ,\Omega \left( Y\right) \right] ,$ $\Omega ^{4}=I.$ The
superalgebra $\mathfrak{f}$ then admits a $%
\mathbb{Z}
_{4}$ grade space decomposition%
\begin{equation}
\mathfrak{f=f}_{0}\mathfrak{\oplus f}_{1}\mathfrak{\oplus f}_{2}\mathfrak{%
\oplus f}_{3},  \label{Z4 grading}
\end{equation}%
which is consistent with the (anti)-commutation relations $\left[ \mathfrak{f%
}_{i},\mathfrak{f}_{j}\right] \subset \mathfrak{f}_{(i+j)\func{mod}4}.$ The
subspace $\mathfrak{f}_{j}$ is formed by the elements of $\mathfrak{f}$ with 
$%
\mathbb{Z}
_{4}$ grading $j,$ $\Omega (\mathfrak{f}_{j})=(i)^{j}$ $\mathfrak{f}_{j}.$
The even (or bosonic) subalgebra is $\mathfrak{f}_{B}=\mathfrak{f}_{0}%
\mathfrak{\oplus f}_{2}$ while the odd (or fermionic) part of $\mathfrak{f}$
is formed by $\mathfrak{f}_{F}=\mathfrak{f}_{1}\mathfrak{\oplus f}_{3}.$

We need to introduce a semisimple element $\Lambda \in \mathfrak{f}_{2}$
which induces the following superalgebra spliting%
\begin{equation*}
\mathfrak{f=f}^{\perp }\text{ }\mathfrak{\oplus }\text{ }\mathfrak{f}%
^{\parallel },\text{ \ \ \ \ \ }\mathfrak{f}^{\perp }\text{ }\mathfrak{\cap }%
\text{ }\mathfrak{f}^{\parallel }=\oslash ,
\end{equation*}%
where $\mathfrak{f}^{\perp }\equiv \ker (ad(\Lambda ))$ and $\mathfrak{f}%
^{\parallel }\equiv \func{Im}(ad(\Lambda )).$

We restric ourselves to the situation in which $\mathfrak{f}$ admits an
extra $%
\mathbb{Z}
_{2}$ gradation $\sigma :$ $\mathfrak{f\rightarrow f,}$ $\sigma \left( \left[
X,Y\right] \right) =\left[ \sigma \left( X\right) ,\sigma \left( Y\right) %
\right] ,$ $\sigma ^{2}=I$ with $\sigma (\mathfrak{f}^{\perp })=\mathfrak{f}%
^{\perp }$ and $\sigma (\mathfrak{f}^{\parallel })=-\mathfrak{f}^{\parallel
},$ implying that $\mathfrak{f}$ is also a symmetric space\footnote{%
In particular, this is satisfied by the superalgebras entering the Pohlmeyer
reduction of $AdS_{n}\times S^{n},$ $n=2,3,5$ and $AdS_{4}\times 
\mathbb{C}
P^{3}$ superstring sigma models. The only exception is $n=2,$ which has $%
\mathfrak{f}_{0}^{\perp }=\varnothing .$} 
\begin{equation}
\left[ \mathfrak{f}^{\perp },\mathfrak{f}^{\perp }\right] \subset \mathfrak{f%
}^{\perp },\text{ \ \ \ \ \ }\left[ \mathfrak{f}^{\perp },\mathfrak{f}%
^{\parallel }\right] \subset \mathfrak{f}^{\parallel }\text{ , \ \ \ \ \ }%
\left[ \mathfrak{f}^{\parallel },\mathfrak{f}^{\parallel }\right] \subset 
\mathfrak{f}^{\perp }.  \label{kernel-image finite}
\end{equation}

The algebraic structure underlying the integrable hierarchy we are
interested in, is defined by the following graded loop Lie superalgebra%
\begin{equation}
\widehat{\mathfrak{f}}=\dbigoplus\limits_{n\in 
\mathbb{Z}
=-\infty }^{+\infty }\left( \dbigoplus\limits_{j\in 
\mathbb{Z}
=0}^{3}z^{2n+\frac{j}{2}}\otimes \mathfrak{f}_{j}\right) ,
\label{loop algebra}
\end{equation}%
which can be rewriten as a half-integer decomposition%
\begin{equation}
\widehat{\mathfrak{f}}=\dbigoplus\limits_{r\in 
\mathbb{Z}
/2=-\infty }^{+\infty }\widehat{\mathfrak{f}}_{r},\text{ \ \ \ \ \ }\left[ Q,%
\widehat{\mathfrak{f}}_{r}\right] =r\widehat{\mathfrak{f}}_{r}
\label{half-integer expansion}
\end{equation}%
in terms of the homogeneous gradation $Q=z\frac{d}{dz}.$ The complex
variable $z$ will enter later in the Lax operators as the spectral parameter
and it is worth to note that under $Q$, the integer and half-integer
elements of $\widehat{\mathfrak{f}}$ are, respectively, bosonic and
fermionic in character.

The splitting (\ref{kernel-image finite}) is now lifted to the affine
algebra $\widehat{\mathfrak{f}}$ which we write in the form%
\begin{equation*}
\widehat{\mathfrak{f}}=\mathcal{K\oplus M},
\end{equation*}%
where $\mathcal{K=K}_{B}\oplus \mathcal{K}_{F}=\ker (ad(\Lambda ^{(\pm
1)})), $ $\mathcal{M=M}_{B}\oplus \mathcal{M}_{F}=\func{Im}(ad(\Lambda
^{(\pm 1)}))$ and $\left[ Q,\Lambda ^{(\pm 1)}\right] =\pm \Lambda ^{(\pm
1)}.$ In what follows the superscript $r$ of an element $X^{(r)}\in \widehat{%
\mathfrak{f}}_{r}$ stands for the homogeneous grading $\left[ Q,X^{(r)}%
\right] =rX^{(r)},$ $r\in 
\mathbb{Z}
/2$ and projections along $\mathcal{K}$ and $\mathcal{M}$ will be denoted by 
$(\ast \mathcal{)}^{\perp }$ and $(\ast )^{\parallel },$ respectively$.$ As
above, we have%
\begin{equation}
\left[ \mathcal{K}\text{,}\mathcal{K}\right] \subset \mathcal{K}\text{, \ \
\ \ \ }\left[ \mathcal{K}\text{,}\mathcal{M}\right] \subset \mathcal{M}\text{%
, \ \ \ \ \ }\left[ \mathcal{M},\mathcal{M}\right] \subset \mathcal{K}\text{.%
}  \label{kernel-image infinite}
\end{equation}

This homogeneous half-integer gradation is enough for all our purposes%
\footnote{%
This is why we have chosen the name homogeneous integrable hierarchy in this
paper.}, namely, to introduce the symmetry flows, to deduce the Lax
operators and to extract the conserved charges. The first and second
relations in (\ref{kernel-image infinite}) are at the heart of the
implementation of symmetry flows by means of the algebraic dressing
technique in which the symmetries are associated to the subalgebra $\mathcal{%
K}$ while the dynamical physical fields are associated to $\mathcal{M}$
inducing a mapping $\delta _{\mathcal{K}}:\mathcal{M\rightarrow M}$ from
physical fields to physical fields. This is also the algebraic setting
behind the Drinfeld-Sokolov procedure we shall use later.

We now proceed to make this construction more precise. Decompose (\ref%
{half-integer expansion}) as $\widehat{\mathfrak{f}}=\widehat{\mathfrak{f}}%
_{-}+\widehat{\mathfrak{f}}_{+}$, where $\widehat{\mathfrak{f}}_{\pm }$ are
the positive and negative subalgebras induced by the gradation $Q$ and
introduce the following loop supergroup matrices: the so-called dressing
matrices%
\begin{eqnarray}
\Theta &=&\exp \left( \chi ^{\left( -1/2\right) }+\chi ^{\left( -1\right)
}+\chi ^{\left( -3/2\right) }+...\right) ,\text{ \ \ \ \ \ }\Theta ^{\prime }%
\text{ }=\text{ }B^{-1}\Theta ,  \label{Dressing matrices} \\
\Pi &=&B\Pi ^{\prime },\text{ \ \ \ \ \ }\Pi ^{\prime }\text{ }=\text{ }\exp
-\left( \chi ^{\left( +1/2\right) }+\chi ^{\left( +1\right) }+\chi ^{\left(
+3/2\right) }+...\right) ,  \notag
\end{eqnarray}%
where $B\in G=\exp \widehat{\mathfrak{f}}_{0}$ is the Toda field, $\chi
^{\left( r\right) }=\psi ^{\left( r\right) }+\theta ^{(r)}\in \widehat{%
\mathfrak{f}}_{r},$ $\theta ^{(r)}\in \widehat{\mathfrak{f}}_{r}^{\perp }$
and $\psi ^{\left( r\right) }\in \widehat{\mathfrak{f}}_{r}^{\parallel }$
are related to the matter fields. Note that $B$ appears in two different
positions and this will be very useful later.

Recall \cite{babelon-bernad} that the dressing transformation of $x\in 
\widehat{F}$ by $g\in \widehat{F}$ is defined by $^{g}x=\left(
xgx^{-1}\right) _{\pm }xg_{\pm }^{-1},$ where $g=g_{-}^{-1}g_{+}.$ For an
element $g=\exp A$ with $A=A_{+}+A_{-}$ and $A_{\pm }\in \widehat{\mathfrak{f%
}}_{\pm },$ the infinitesimal dressing transformation is 
\begin{equation}
\delta _{A}x=^{g}x-x=\pm \left( xAx^{-1}\right) _{\pm }x\mp xA_{\pm }.
\label{Infinitesimal action}
\end{equation}%
We are interested in using the kernel subalgebra $\mathcal{K}$ to generate
actions on the dressing matrices which are carrying the dynamical degrees of
freedom. From (\ref{Infinitesimal action}) and the decomposition $\mathcal{%
K=K}_{-}\mathcal{+K}_{+},$ we find the infinitesimal actions of $A=A_{+}\in 
\mathcal{K}_{+}$ and $A=A_{-}\in \mathcal{K}_{-}$ on $x=\Theta $ and $x=\Pi
, $ respectively%
\begin{equation}
\delta _{A_{+}}\Theta =-\left( \Theta A_{+}\Theta ^{-1}\right) _{-}\Theta 
\text{, \ \ \ \ \ }\delta _{A_{-}}\Pi =+\left( \Pi A_{-}\Pi ^{-1}\right)
_{+}\Pi ,  \label{positive variarion}
\end{equation}%
where $A_{\pm }$ are linear combinations of elements of $\mathcal{K}_{\pm }$
and $\left( \ast \right) _{\pm }$ stands for projections $\mathcal{P}_{\pm }$
along $\widehat{\mathfrak{f}}_{\pm }.$ It is not difficult to show that for $%
A=A_{-}$ and $x=\Theta $ and for $A=A_{+}$ and $x=\Pi $ the variations
vanish, $\delta _{A_{-}}\Theta =0$ and $\delta _{A_{+}}\Pi =0.$ Hence, in
the present form, the dressing matrices (\ref{Dressing matrices}) only
evolve under half of the kernel algebra $\mathcal{K}$ and we have two sets
of decoupled evolution equations.

Taking $A_{\pm }=$ $t_{\pm n}\Lambda ^{(\pm n)}$, $n\in 
\mathbb{Z}
^{+},$ where $\Lambda ^{(\pm n)}\in Cent\left( \mathcal{K}\right) $ belong
to the center of $\mathcal{K}$ and taking the limit $t_{\pm n}\rightarrow 0,$
we obtain\footnote{%
Use $\delta _{A_{+}}\Theta /t_{+}=\left( ^{A_{+}}\Theta -\Theta \right)
/t_{+}\rightarrow \partial _{+n}\Theta $ and a similar expression for $%
\delta _{A_{-}}\Pi .$} the isospectral evolutions of $\Theta $ and $\Pi $%
\begin{equation}
\partial _{+n}\Theta =-\left( \Theta \Lambda ^{(+n)}\Theta ^{-1}\right)
_{-}\Theta \text{, \ \ \ \ \ }\partial _{-n}\Pi =+\left( \Pi \Lambda
^{(-n)}\Pi ^{-1}\right) _{+}\Pi .  \label{Isospectral evolutions}
\end{equation}%
From equations (\ref{Isospectral evolutions}) we obtain the Lax connections%
\begin{equation}
\Lambda _{\Theta }^{(+n)}=\left( \Theta \Lambda ^{(+n)}\Theta ^{-1}\right)
_{+},\text{ \ \ \ \ \ }\Lambda _{\Pi }^{(-n)}=\left( \Pi \Lambda ^{(-n)}\Pi
^{-1}\right) _{-}  \label{Lax connections}
\end{equation}%
and the Lax operators $L_{+n}=\partial _{+n}-\Lambda _{\Theta }^{(+n)},$ $%
L_{-n}=\partial _{-n}+\Lambda _{\Pi }^{(-n)}$ from the dressing relations 
\begin{equation}
\text{\ }L_{-n}=\Pi L_{-n}^{V}\Pi ^{-1},\text{ \ \ \ \ \ }L_{+n}\text{ = }%
\Theta L_{+n}^{V}\Theta ^{-1},  \label{dressin lax}
\end{equation}%
where $L_{\pm n}^{V}=\partial _{\pm n}\mp \Lambda ^{(\pm n)}$ are the vaccum
Lax operators. The Baker-Akhiezer wave functions $\Psi _{\pm }$ are defined
by $L_{\pm n}\Psi _{\mp }$ $=0$ and are given by%
\begin{equation*}
\Psi _{-}=\Theta \exp \left( +\dsum\limits_{n\in 
\mathbb{Z}
^{+}}t_{+n}\Lambda ^{(+n)}\right) \text{, \ \ \ \ \ }\Psi _{+}=\Pi \exp
\left( -\dsum\limits_{n\in 
\mathbb{Z}
^{+}}t_{-n}\Lambda ^{(-n)}\right) .
\end{equation*}

The equations (\ref{Isospectral evolutions}) describe two identical but
decoupled sets of evolution equations as we mentioned above, the coupling of
the two sectors (of positive and negative times) is achieved by imposing the
relation $g=\Psi _{-}^{-1}\Psi _{+}$ with $g=g_{-}^{-1}g_{+}\in \widehat{F}$
a constant loop group element. Alternatively, we have (see also \cite{frank})%
\begin{equation}
\exp \left( +\dsum\limits_{n\in 
\mathbb{Z}
^{+}}t_{+n}\Lambda ^{(+n)}\right) g\exp \left( +\dsum\limits_{n\in 
\mathbb{Z}
^{+}}t_{-n}\Lambda ^{(-n)}\right) =\Theta ^{-1}(t)\Pi (t).  \label{RH}
\end{equation}%
This is the Riemann-Hilbert factorization problem we use in order to extend
the associated integrable hierarchy described by (\ref{positive variarion})
to flow now under the negative times. From (\ref{RH}) we recover (\ref%
{Isospectral evolutions}) and two important extra equations describing the
isospectral evolution of $\Theta $ and $\Pi $ with respect opposite flow
parameters 
\begin{equation}
\partial _{+n}\Pi =+\left( \Theta \Lambda ^{(+n)}\Theta ^{-1}\right) _{+}\Pi 
\text{, \ \ \ \ \ }\partial _{-n}\Theta =-\left( \Pi \Lambda ^{(-n)}\Pi
^{-1}\right) _{-}\Theta .  \label{Isospectral II}
\end{equation}%
These equations are extended to actions of $A_{+}\in \mathcal{K}_{+}$ and $%
A_{-}\in \mathcal{K}_{-}$ on $\Pi ,$ $\Theta $ and besides of (\ref{positive
variarion}) we also have now that%
\begin{equation}
\delta _{A_{+}}\Pi =+\left( \Theta A_{+}\Theta ^{-1}\right) _{+}\Pi \text{,
\ \ \ \ \ }\delta _{A_{-}}\Theta =-\left( \Pi A_{-}\Pi ^{-1}\right)
_{-}\Theta .  \label{negative variation}
\end{equation}

The equations (\ref{positive variarion}), (\ref{Isospectral evolutions}) and
(\ref{Isospectral II}), (\ref{negative variation}) describe the isospectral
evolution and non-Abelian variations of the dressing matrices $\Theta $ and $%
\Pi .$ Note that the flows associated to the positive times are dual to the
ones associated to the negative times, in the sense that $\mathcal{K}%
_{+}^{\ast }\simeq \mathcal{K}_{-}$ under the (assumed to exists)
non-degenerate inner product $\left\langle A,B\right\rangle $ on $\widehat{%
\mathfrak{f}}$ defined by%
\begin{equation}
\left\langle X^{(r)},Y^{(s)}\right\rangle _{\widehat{\mathfrak{f}}}=\delta
_{r+s,0}\times Str\left( X\cdot Y\right) _{\mathfrak{f}},
\label{inner product}
\end{equation}%
where $X^{(r)}=z^{r}\otimes X_{r},$ $X_{r}\in \mathfrak{f}_{r}$ and $Str$ is
the supertrace in some supermatrix representation of $\mathfrak{f.}$

The decomposition in terms of grades is slightly ambiguous because we can
take $\widehat{\mathfrak{f}}=\widehat{\mathfrak{f}}_{-}+\widehat{\mathfrak{f}%
}_{+}$ with $\widehat{\mathfrak{f}}_{-}=\oplus _{r\leq -1/2}$ $\widehat{%
\mathfrak{f}}_{r}$, $\widehat{\mathfrak{f}}_{+}=\oplus _{r\geq 0}$ $\widehat{%
\mathfrak{f}}_{r}$ or $\widehat{\mathfrak{f}}_{-}=\oplus _{r\leq 0}$ $%
\widehat{\mathfrak{f}}_{r}$, $\widehat{\mathfrak{f}}_{+}=\oplus _{r\geq
+1/2} $ $\widehat{\mathfrak{f}}_{r}$ and this turn out to be related to
gauge symmetries of the form $H_{L}\times H_{R},$ as we shall see below. For
the moment let us take into account this difference in order to rewrite the
symmetry flows more explicitly.

All the above evolution equations for $\Theta $ and $\Pi $ are summarized
into the flow equations\footnote{%
One of the most remarkable properties of the flow equations (\ref{flow
equations}), (\ref{gauge equivalent variations}) is that they associate a 2D
symmetry flow to every Lie algebra generator in $\mathcal{K}$ through a Lax
operator of the form $L_{\mathcal{K}}=\delta _{\mathcal{K}}+\mathcal{A}_{%
\mathcal{K}}.$ The symmetry field variations are obtained by dressing the
identities $\left[ L_{\mathcal{K}}^{V},L_{\pm }^{V}\right] =0,$ where $L_{%
\mathcal{K}}^{V},L_{\pm }^{V}$ are the vacuum Lax operators.}%
\begin{eqnarray}
\delta _{A_{+}}\Theta &=&-\left( \Theta A_{+}\Theta ^{-1}\right) _{<0}\Theta 
\text{, \ \ \ \ \ }\delta _{A_{-}}\Pi \text{ }=\text{ }+\left( \Pi A_{-}\Pi
^{-1}\right) _{\geq 0}\Pi ,  \label{flow equations} \\
\delta _{A_{+}}\Pi &=&+\left( \Theta A_{+}\Theta ^{-1}\right) _{\geq 0}\Pi 
\text{, \ \ \ \ \ }\delta _{A_{-}}\Theta \text{ }=\text{ }-\left( \Pi
A_{-}\Pi ^{-1}\right) _{<0}\Theta ,  \notag
\end{eqnarray}%
where we have chosen to put the grade zero part $\widehat{\mathfrak{f}}_{0}$
in $\widehat{\mathfrak{f}}_{+}.$ Let us note that after the coupling the
dynamical degrees of freedom are doubled by the extension because now we
have two dressing matrices $\Theta ,\Pi $ carrying different sets of fields.

The RHS of (\ref{RH}) can be written in an equivalent way because $\Theta
^{-1}\Pi =\Theta ^{\prime -1}\Pi ^{\prime }$, cf (\ref{Dressing matrices}).
In these prime variables the equations (\ref{flow equations}) become%
\begin{eqnarray}
\delta _{A_{+}}\Theta ^{\prime } &=&-\left( \Theta ^{\prime }A_{+}\Theta
^{\prime -1}\right) _{\leq 0}\Theta ^{\prime }\text{, \ \ \ \ \ }\delta
_{A_{-}}\Pi ^{\prime }\text{ }=\text{ }+\left( \Pi ^{\prime }A_{-}\Pi
^{\prime -1}\right) _{>0}\Pi ^{\prime },  \label{gauge equivalent variations}
\\
\delta _{A_{+}}\Pi ^{\prime } &=&+\left( \Theta ^{\prime }A_{+}\Theta
^{\prime -1}\right) _{>0}\Pi ^{\prime }\text{, \ \ \ \ \ }\delta
_{A_{-}}\Theta ^{\prime }\text{ }=\text{ }-\left( \Pi ^{\prime }A_{-}\Pi
^{\prime -1}\right) _{\leq 0}\Theta ^{\prime },  \notag
\end{eqnarray}%
where we have chosen to put the grade zero part $\widehat{\mathfrak{f}}_{0}$
in $\widehat{\mathfrak{f}}_{-}.$ From (\ref{gauge equivalent variations}) we
obtain the Lax connections%
\begin{equation}
\Lambda _{\Theta ^{\prime }}^{(+n)}=\left( \Theta ^{\prime }\Lambda
^{(+n)}\Theta ^{\prime -1}\right) _{>0},\text{ \ \ \ \ \ }\Lambda _{\Pi
^{\prime }}^{(-n)}=\left( \Pi ^{\prime }\Lambda ^{(-n)}\Pi ^{\prime
-1}\right) _{\leq 0}  \label{prime Lax connections}
\end{equation}%
and the Lax operators $L_{+n}^{\prime }=\partial _{+n}-\Lambda _{\Theta
^{\prime }}^{(+n)},$ $L_{-n}^{\prime }=\partial _{-n}+\Lambda _{\Pi ^{\prime
}}^{(-n)}$ from the dressing relations (\ref{dressin lax}) with\ $\Theta
,\Pi $ replaced by $\Theta ^{\prime },\Pi ^{\prime }.$ The prime and
un-prime expressions make clear the projections to be used in computations.
Of course, the two formulations are completely equivalent and it is not
difficult to see that the Lax operators are related by a $B$-conjugation%
\begin{equation}
L_{\pm n}^{\prime }=B^{-1}L_{\pm n}B.  \label{lax conjugation}
\end{equation}%
However, as we will see along the text this is not an ordinary gauge
transformation because in the decomposition $\widehat{\mathfrak{f}}_{0}=%
\widehat{\mathfrak{f}}_{0}^{\perp }+\widehat{\mathfrak{f}}_{0}^{\parallel }$%
, the $\widehat{\mathfrak{f}}_{0}^{\perp }$ is the gauge algebra while in $%
B\in G=\exp \widehat{\mathfrak{f}}_{0}$ we have physical fields in $\widehat{%
\mathfrak{f}}_{0}^{\parallel }.$

In the mKdV hierarchy the grade zero part of the bosonic kernel is empty,
i.e $\mathcal{K}_{B}^{(0)}=\widehat{\mathfrak{f}}_{0}^{\perp }=\oslash $ and
this was the situation already considered in \cite{susyflows-mKdV}. Now we
have that $\mathcal{K}_{B}^{(0)}\neq \oslash ,$ which is more interesting$.$
In this case we have symmetry flows associated to $\widehat{\mathfrak{f}}%
_{0}^{\perp }.$ As mentioned before, this flows are nothing but usual gauge
symmetries and now we proceed to identify them.

The $H_{L}\times H_{R}$ gauge transformations are generated by the elements $%
K_{L/R}^{(0)}\in \mathfrak{h}_{L}\times \mathfrak{h}_{R}=\widehat{\mathfrak{f%
}}_{0}^{\perp }\times \widehat{\mathfrak{f}}_{0}^{\perp }$ and their action
are encoded in the following flow equations%
\begin{eqnarray}
\delta _{L}\Theta &=&-\left( \Theta K_{L}^{(0)}\Theta ^{-1}\right)
_{<0}\Theta \text{, \ \ \ \ \ \ \ \ }\delta _{L}\Pi \text{ }=\text{ }+\left(
\Theta K_{L}^{(0)}\Theta ^{-1}\right) _{\geq 0}\Pi ,  \label{gauge flows} \\
\delta _{R}\Theta ^{\prime } &=&-\left( \Pi ^{\prime }K_{R}^{(0)}\Pi
^{\prime -1}\right) _{\leq 0}\Theta ^{\prime }\text{, \ \ \ \ \ }\delta
_{R}\Pi ^{\prime }\text{ }=\text{ }+\left( \Pi ^{\prime }K_{R}^{(0)}\Pi
^{\prime -1}\right) _{>0}\Pi ^{\prime }.  \notag
\end{eqnarray}%
Consider the first line of (\ref{gauge flows}) and a constant element $%
K_{L}^{(0)}\in \mathfrak{h}_{L}.$ These equations give rise to the dressing
relations 
\begin{equation*}
\Theta \left( \delta _{L}-K_{L}^{(0)}\right) \Theta ^{-1}=\delta
_{L}-K_{L}^{(0)}\text{, \ \ \ \ \ }\Pi \left( \delta _{L}\right) \Pi
^{-1}=\delta _{L}-K_{L}^{(0)}.
\end{equation*}%
The first equation is equivalent to the infinitesimal gauge transformations $%
\delta _{L}\Theta =\left[ K_{L}^{(0)},\Theta \right] $ and the second
equation is equivalent to $\delta _{L}\Pi ^{\prime }\Pi ^{\prime -1}=0$ and $%
\delta _{L}B=K_{L}^{(0)}B,$ where we have used $\Pi =B\Pi ^{\prime }.$
Setting $\Gamma _{L}=\exp K_{L}^{(0)}\in H_{L}=\exp \mathfrak{h}_{L},$ we
get the finite gauge transformations%
\begin{equation}
\widetilde{\Theta }_{L}=\Gamma _{L}\Theta \Gamma _{L}^{-1}\text{, \ \ \ \ }%
\widetilde{B}_{L}=\Gamma _{L}B\text{, \ \ \ \ \ }\widetilde{\Pi ^{\prime }}%
_{L}=\Pi ^{\prime }.  \label{left gauge transformations}
\end{equation}%
Under these left-gauge transformations the Lax connections $\Lambda _{\Theta
}^{(+n)}$ , $\Lambda _{\Pi }^{(-n)}$ transform as 
\begin{equation*}
\Lambda _{\widetilde{\Theta }_{L}}^{(+n)}=\Gamma _{L}\Lambda _{\Theta
}^{(+n)}\Gamma _{L}^{-1},\text{ \ \ \ \ \ }\Lambda _{\widetilde{\Pi }%
_{L}}^{(-n)}=\Gamma _{L}\Lambda _{\Pi }^{(-n)}\Gamma _{L}^{-1}
\end{equation*}%
preserving the compatibility conditions because $\widetilde{F}%
_{mn}^{L}=\Gamma _{L}F_{mn}\Gamma _{L}^{-1}$. \ In an analogous way,
considering the second line of (\ref{gauge flows}) and an element $%
K_{R}^{(0)}\in \mathfrak{h}_{R}$ we have the dressing relations 
\begin{equation*}
\Theta ^{\prime }\left( \delta _{R}\right) \Theta ^{\prime -1}=\delta
_{R}+K_{R}^{(0)}\text{, \ \ \ \ \ }\Pi ^{\prime }\left( \delta
_{R}+K_{R}^{(0)}\right) \Pi ^{\prime -1}=\delta _{R}+K_{R}^{(0)},\ 
\end{equation*}%
which are equivalent to the infinitesimal gauge transformations $\delta
_{R}\Theta \Theta ^{-1}=0,$ $\delta _{R}B=BK_{R}^{(0)}$ and $\delta _{R}\Pi
^{\prime }=-\left[ K_{R}^{(0)},\Pi ^{\prime }\right] $. Setting $\Gamma
_{R}=\exp K_{R}^{(0)}\in H_{R}=\exp \mathfrak{h}_{R},$ we get the finite
gauge transformations%
\begin{equation}
\widetilde{\Theta }_{R}=\Theta \text{, \ \ \ \ \ }\widetilde{B}_{R}=B\Gamma
_{R}\text{, \ \ \ \ \ }\widetilde{\Pi ^{\prime }}_{R}=\Gamma _{R}^{-1}\Pi
^{\prime }\Gamma _{R}.  \label{right gauge transformations}
\end{equation}%
Under (\ref{right gauge transformations}) the Lax connections $\Lambda
_{\Theta }^{(+n)}$, $\Lambda _{\Pi }^{(-n)}$ are gauge invariant and we have 
$\widetilde{F}_{mn}^{R}=F_{mn}.$ This is not an asymmetric behavior, if we
perform the same analysis in the $B$-equivalent representation given by (\ref%
{gauge equivalent variations}) with Lax connections $\Lambda _{\Theta
^{\prime }}^{(+n)}$ , $\Lambda _{\Pi ^{\prime }}^{(-n)}$, the situation is
reversed, i.e $\widetilde{F}_{mn}^{\prime R}=$ $\Gamma
_{R}^{-1}F_{mn}^{\prime }\Gamma _{R}$\ and $\widetilde{F}_{mn}^{\prime
L}=F_{mn}^{\prime }.$ This is a consequence of the position of the Toda
field in (\ref{Dressing matrices}).

Then, combining (\ref{left gauge transformations}) and (\ref{right gauge
transformations}) we have the total finite action of the gauge group $%
H_{L}\times H_{R}$ on the dressing matrices%
\begin{equation}
\widetilde{\Theta }=\Gamma _{L}\Theta \Gamma _{L}^{-1},\ \ \ \ \ \ 
\widetilde{B}=\Gamma _{L}B\Gamma _{R},\ \ \ \ \ \ \widetilde{\Pi ^{\prime }}%
=\Gamma _{R}^{-1}\Pi ^{\prime }\Gamma _{R}
\label{full gauge transformations}
\end{equation}%
and from (\ref{Dressing matrices}) we get the action on each graded subspace 
$\psi ^{(\pm r)}\in \widehat{\mathfrak{f}}_{\pm r}.$ In particular and for
future reference, we write%
\begin{equation}
\widetilde{\psi }^{\left( -1/2\right) }=\Gamma _{L}\psi ^{\left( -1/2\right)
}\Gamma _{L}^{-1},\ \ \ \ \ \ \widetilde{B}=\Gamma _{L}B\Gamma _{R}\ ,\ \ \
\ \ \ \widetilde{\psi }^{\left( +1/2\right) }=\Gamma _{R}^{-1}\psi ^{\left(
+1/2\right) }\Gamma _{R}.  \label{component gauge trans}
\end{equation}

Let us mention that these global gauge symmetries can be related to
Kac-Moody algebras if we promote the gauge parameters to be chiral, see (\ref%
{kac moody symmetry}) below.

\subsection{Relativistic sector of the extended homogeneous hierarchy.}

In what follows we restrict the above construction to the subsystem
associated to the flows $(t_{-1},t_{-1/2},t_{0},t_{+1/2},t_{+1})$ in the
extended homogeneous (also AKNS) hierarchy\footnote{%
The name AKNS is because for the loop algebra $sl(2)^{(1)},$ the times $%
t_{+1},t_{+2}$ leads to the AKNS\ equations while the times $t_{+1},t_{-1}$
leads to the complex sine-Gordon equations (see \cite{AKNS}).}. We have
already identified the flows corresponding to $t_{0}$ with gauge symmetries
and now we want to deduce the Lax operators we are going to use and to find
some immediate consequences. In the next section we initiate the study of
the fermionic symmetry flows associated to $t_{\pm 1/2}.$

The relativistic sector of the extended homogeneous hierarchy is defined by (%
\ref{flow equations}), (\ref{gauge equivalent variations}) for the two
constant elements $\Lambda _{\pm }^{(\pm 1)}\in \mathcal{K}_{\pm }$ of
grades $\pm 1,$ associated to the isospectral times $t_{\pm 1}=-x^{\pm }$.
We are interested in the action of (\ref{full gauge transformations}), (\ref%
{component gauge trans}) as local gauge transformations preserving the
compatibility conditions $\left[ L_{+},L_{-}\right] =\left[ L_{+}^{\prime
},L_{-}^{\prime }\right] =0$ because of the relation of this integrable
hierarchy with the Pohlmeyer reduced models. Then, we add two gauge
connections $A_{\pm }^{(L)}\in \mathfrak{h}_{L}$ and $A_{\pm }^{(R)}\in 
\mathfrak{h}_{R}$ to the Lax operators and transforming as follows 
\begin{equation}
\widetilde{A}_{\pm }^{(L)}=\Gamma _{L}A_{\pm }^{(L)}\Gamma
_{L}^{-1}+\partial _{\pm }\Gamma _{L}\Gamma _{L}^{-1}\text{, \ \ \ \ \ }%
\widetilde{A}_{\pm }^{(R)}=\Gamma _{R}^{-1}A_{\pm }^{(R)}\Gamma _{R}-\Gamma
_{R}^{-1}\partial _{\pm }\Gamma _{R}  \label{gauge transf}
\end{equation}%
getting the desired covariant behavior $\widetilde{L}_{\pm },=\Gamma
_{L}L_{\pm }\Gamma _{L}^{-1}$ and $\widetilde{L}_{\pm }^{\prime }=\Gamma
_{R}^{-1}L_{\pm }^{\prime }\Gamma _{R}.$ Explicitly, they are 
\begin{eqnarray*}
L_{+} &=&\partial _{+}-A_{+}^{(L)}+\left( A_{+}^{(0)}+Q_{+}^{(0)}+\psi
_{+}^{\left( +1/2\right) }+\Lambda _{+}^{(+1)}\right) , \\
L_{-} &=&\partial _{-}-A_{-}^{(L)}-B\left( \psi _{-}^{\left( -1/2\right)
}+\Lambda _{-}^{(-1)}\right) B^{-1}
\end{eqnarray*}%
and%
\begin{eqnarray*}
L_{+}^{\prime } &=&\partial _{+}-A_{+}^{(R)}+B^{-1}\left( \psi _{+}^{\left(
+1/2\right) }+\Lambda _{+}^{(+1)}\right) B, \\
L_{-}^{\prime } &=&\partial _{-}-A_{-}^{(R)}-\left(
A_{-}^{(0)}+Q_{-}^{(0)}+\psi _{-}^{\left( -1/2\right) }+\Lambda
_{-}^{(-1)}\right) ,
\end{eqnarray*}%
where 
\begin{equation}
\psi _{\pm }^{\left( \pm 1/2\right) }=\pm \left[ \psi ^{\left( \mp
1/2\right) },\Lambda _{\pm }^{(\pm 1\text{)}}\right] \in \mathcal{M}%
_{F}^{(\pm 1/2)},\text{ }A_{\pm }^{(0)}=\pm \left[ \psi ^{(\mp 1)},\Lambda
_{\pm }^{(\pm 1)}\right] \in \mathcal{M}_{B}^{(0)},\text{ }Q_{\pm }^{(0)}=%
\frac{1}{2}\left[ \psi ^{(\mp 1/2)},\left[ \psi ^{(\mp 1/2)},\Lambda _{\pm
}^{(\pm 1)}\right] \right] \in \mathcal{K}_{B}^{(0)}.
\label{Baker-Hausdorf fields}
\end{equation}

The dynamical fields are encoded in the expressions (\ref{Baker-Hausdorf
fields}) but these relations are rather obscure, at least for the bosonic
fields\footnote{%
Note that the fermions are automatically in the image subspace $\mathcal{M=}$
$\widehat{\mathfrak{f}}^{\parallel }$. This is an important issue related to
gauge fixing of the residual kappa symmetry in the reduction of superstring
sigma models.}. To identify the fields in a precise way we appeal to the
relation (\ref{lax conjugation}) found above which relates the field content
between $L_{\pm }^{\prime }$ and $L_{\pm }.$ From $L_{\pm }^{\prime
}=B^{-1}L_{\pm }B$ we find the relations%
\begin{equation}
A_{+}^{(L)}-A_{+}^{(0)}=\widehat{A}_{+}^{(L)}\text{, \ \ \ \ \ }%
A_{-}^{(R)}+A_{-}^{(0)}=\widehat{A}_{-}^{(R)},  \label{relations}
\end{equation}%
where%
\begin{equation}
\widehat{A}_{+}^{(L)}\overset{def}{=}\partial
_{+}BB^{-1}+BA_{+}^{(R)}B^{-1}+Q_{+}^{(0)},\text{ \ \ \ \ \ }\widehat{A}%
_{-}^{(R)}\overset{def}{=}-B^{-1}\partial
_{-}B+B^{-1}A_{-}^{(L)}B-Q_{-}^{(0)}.  \label{total missing}
\end{equation}

By projecting (\ref{total missing}) along the gauge algebras $\mathfrak{h}%
_{L},\mathfrak{h}_{R}$ we find the componets $A_{+}^{(L)}$, $A_{-}^{(R)}$ as
functions of the other fields 
\begin{equation}
A_{+}^{(L)}=\mathcal{P}_{\mathfrak{h}_{L}}\left( \widehat{A}%
_{+}^{(L)}\right) ,\text{ \ \ \ \ \ }A_{-}^{(R)}=\mathcal{P}_{\mathfrak{h}%
_{R}}\left( \widehat{A}_{-}^{(R)}\right)  \label{missing componentes}
\end{equation}%
and by projecting along the image subspace $\mathcal{M}$, we find the
components $A_{\pm }^{(0)}$ in terms of the Toda field $B$ and the gauge
fields 
\begin{equation}
A_{+}^{(0)}=-\left( \partial _{+}BB^{-1}+BA_{+}^{(R)}B^{-1}\right)
^{\parallel },\text{ \ \ \ \ \ }A_{-}^{(0)}=-\left( B^{-1}\partial
_{-}B-B^{-1}A_{-}^{(L)}B\right) ^{\parallel }.  \label{image fields}
\end{equation}%
Under the inner product (\ref{inner product}) the relations (\ref{relations}%
) satisfy%
\begin{equation}
\left\langle \left( A_{+}^{(0)}\right) ^{2}\right\rangle =\left\langle
\left( \widehat{A}_{+}^{(L)}\right) ^{2}-\left( A_{+}^{(L)}\right)
^{2}\right\rangle ,\text{ \ \ \ \ \ }\left\langle \left( A_{-}^{(0)}\right)
^{2}\right\rangle =\left\langle \left( \widehat{A}_{-}^{(R)}\right)
^{2}-\left( A_{-}^{(R)}\right) ^{2}\right\rangle  \label{kinetic A's}
\end{equation}%
and are functionals of the physical fields in $\mathcal{M}$, as should be%
\footnote{%
These are contributions of the $T_{\pm \pm }$ components of the stress
tensor $T_{\mu \nu }$, see (\ref{Tensor}) below.}.

Finally, we have the final form of the Lax pairs\footnote{%
This is way the components $A_{+}^{(L)}$, $A_{-}^{(R)}$ were termed as
"missing" in \cite{mira-hollowood}. They do not appear explicitly in the
final form of $L_{\pm }.$}%
\begin{eqnarray}
L_{+}\left( A\right) &=&\partial _{+}-\partial
_{+}BB^{-1}-BA_{+}^{(R)}B^{-1}+\psi _{+}^{\left( +1/2\right) }+\Lambda
_{+}^{(+1)},  \label{Lax operators} \\
L_{-}\left( A\right) &=&\partial _{-}-A_{-}^{(L)}-B\left( \psi _{-}^{\left(
-1/2\right) }+\Lambda _{-}^{(-1)}\right) B^{-1}  \notag
\end{eqnarray}%
and%
\begin{eqnarray}
L_{+}^{\prime }\left( A\right) &=&\partial _{+}-A_{+}^{(R)}+B^{-1}\left(
\psi _{+}^{\left( +1/2\right) }+\Lambda _{+}^{(+1)}\right) B,
\label{prime Lax operators} \\
L_{-}^{\prime }\left( A\right) &=&\partial _{-}+B^{-1}\partial
_{-}B-B^{-1}A_{-}^{(L)}B-\psi _{-}^{\left( -1/2\right) }-\Lambda _{-}^{(-1)},
\notag
\end{eqnarray}%
with%
\begin{equation}
\widetilde{A}_{-}^{(L)}=\Gamma _{L}A_{-}^{(L)}\Gamma _{L}^{-1}+\partial
_{-}\Gamma _{L}\Gamma _{L}^{-1}\text{, \ \ \ \ \ }\widetilde{A}%
_{+}^{(R)}=\Gamma _{R}^{-1}A_{+}^{(R)}\Gamma _{R}-\Gamma _{R}^{-1}\partial
_{+}\Gamma _{R}.  \label{connection transformation}
\end{equation}

The equations of motion of the system are, by definition, given by the zero
curvature $F_{+-}=\left[ L_{+},L_{-}\right] $ of (\ref{Lax operators}) and
they define the fermionic extension of the non-Abelian Toda models on the
bi-quotient $H_{L}\backslash G/H_{R}.$ They are given by 
\begin{eqnarray}
F_{+-}^{(+1/2)} &=&-D_{-}^{(L)}\psi _{+}^{\left( +1/2\right) }+\left[ B\psi
_{-}^{\left( -1/2\right) }B^{-1},\Lambda _{+}^{\text{(}+1\text{)}}\right] ,
\label{NA affine toda equations} \\
F_{+-}^{(0)} &=&D_{-}^{(L)}\left( \partial
_{+}BB^{-1}+BA_{+}^{(R)}B^{-1}\right) -\partial _{+}A_{-}^{(L)}-  \notag \\
&&-\left[ \Lambda _{+}^{(+1)},B\Lambda _{-}^{(-1)}B^{-1}\right] -\left[ \psi
_{+}^{\left( +1/2\right) },B\psi _{-}^{\left( -1/2\right) }B^{-1}\right] , 
\notag \\
F_{+-}^{(-1/2)} &=&B\left( -D_{+}^{(R)}\psi _{-}^{\left( -1/2\right) }+\left[
\Lambda _{-}^{(-1)},B^{-1}\psi _{+}^{\left( +1/2\right) }B\right] \right)
B^{-1}\text{,}  \notag
\end{eqnarray}%
where $D_{-}^{(L)}=\partial _{-}-\left[ A_{-}^{(L)},\text{ }\right] $ and $%
D_{+}^{(R)}=\partial _{+}-\left[ A_{+}^{(R)},\text{ }\right] $ are the
covariant derivatives for the $H_{L}\times H_{R}$ actions of the gauge group$%
.$ The equations given by (\ref{prime Lax operators}) are simply $%
F_{+-}^{\prime }=B^{-1}F_{+-}B.$

The gauge fields $A_{\pm }^{(L)}$, $A_{\pm }^{(R)}$ are flat, an important
property to be used later. To see this, we note that the grade zero
equations of motion $F_{+-}^{(0)}$ and $F_{+-}^{\prime (0)}$ can be written,
with the help of the Jacobi identity and the $F_{+-}^{(\pm 1/2)},$ $%
F_{+-}^{\prime (\pm 1/2)}$ equations of motion, as%
\begin{eqnarray*}
F_{+-}^{(0)} &=&\partial _{-}\widehat{A}_{+}^{(L)}-\partial _{+}A_{-}^{(L)}+%
\left[ \widehat{A}_{+}^{(L)},A_{-}^{(L)}\right] -\left[ \Lambda _{+}^{\text{(%
}+1\text{)}},X_{-}^{(-1)}\right] , \\
F_{+-}^{\prime (0)} &=&\partial _{-}A_{+}^{(R)}-\partial _{+}\widehat{A}%
_{-}^{(R)}+\left[ A_{+}^{(R)},\widehat{A}_{-}^{(R)}\right] +\left[ \Lambda
_{-}^{\text{(}-1\text{)}},X_{+}^{(+1)}\right] ,
\end{eqnarray*}%
where $X_{\pm }^{(\pm 1)}\in \widehat{\mathfrak{f}}.$ Projecting $%
F_{+-}^{(0)}$ and $F_{+-}^{\prime (0)}$ along $\mathfrak{h}_{L}\ $and $%
\mathfrak{h}_{R}$ respectively and taking into account (\ref{missing
componentes}), we conclude that the connections $A_{\pm }^{(L)}$, $A_{\pm
}^{(R)}$ are pure gauge%
\begin{equation}
\left[ \partial _{+}-A_{+}^{(L/R)},\partial _{-}-A_{-}^{(L/R)}\right] =0.
\label{flat connections}
\end{equation}

In the (on-shell) gauge $A_{\pm }^{(L/R)}=0,$ the equations of motion of the
fermionic extension of the non-abelian Toda models (\ref{NA affine toda
equations}), together with the equations (\ref{missing componentes}) become,
respectively,%
\begin{eqnarray}
\partial _{-}\psi _{+}^{(+1/2)} &=&\left[ B\psi _{-}^{(-1/2)}B^{-1},\Lambda
_{+}^{(+1)}\right] ,  \label{gauged fix eq of motion} \\
\partial _{-}\left( \partial _{+}BB^{-1}\right) &=&\left[ \Lambda
_{+}^{(+1)},B\Lambda _{-}^{(-1)}B^{-1}\right] +\left[ \psi
_{+}^{(+1/2)},B\psi _{-}^{(-1/2)}B^{-1}\right] ,  \notag \\
\partial _{+}\psi _{-}^{(-1/2)} &=&\left[ \Lambda _{-}^{(-1)},B^{-1}\psi
_{+}^{(+1/2)}B,\right] ,  \notag \\
\left( \partial _{+}BB^{-1}+Q_{+}^{(0)}\right) ^{\perp } &=&\left(
B^{-1}\partial _{-}B+Q_{-}^{(0)}\right) ^{\perp }\text{ }=\text{ }0.
\label{new constraints}
\end{eqnarray}%
Note that in contrast to the purely bosonic non-Abelian Toda models which
are characterized by the constraints $\left( \partial _{+}BB^{-1}\right)
^{\perp }=\left( B^{-1}\partial _{-}B\right) ^{\perp }=0,$ the new
constraints (\ref{new constraints}) are modified by the fermion bi-linears $%
Q_{\pm }^{(0)}.$ These constraints mean that there are no dynamical degrees
of freedom associated to the kernel subalgebra $\mathcal{K=}$ $\widehat{%
\mathfrak{f}}^{\perp }$, as expected. Note that they are also some sort of
classical bosonization rules.

To compute the classical spins of the fields, it is useful to add a central
extension to the loop algebra (\ref{loop algebra}) and introduce the central
and gradation fields $\nu ,\mu $, in order to restore the conformal
invariance of the equations (\ref{gauged fix eq of motion}). Writing $%
B=\gamma \exp [\eta Q]\exp [\nu C],$ $\gamma \in G,$ we have that the
equations%
\begin{eqnarray}
\partial _{-}\psi _{+}^{\left( +1/2\right) } &=&e^{-\eta /2}\left[ \gamma
\psi _{-}^{\left( -1/2\right) }\gamma ^{-1},\Lambda _{+}^{\text{(}+1\text{)}}%
\right] ,  \label{explicit NA toda eq} \\
\partial _{-}\left( \partial _{+}\gamma \gamma ^{-1}+\gamma
A_{+}^{(R)}\gamma ^{-1}\right) -\partial _{+}A_{-}^{(L)}+\partial
_{-}\partial _{+}\nu C &=&e^{-\eta }\left[ \Lambda _{+}^{(+1)},\gamma
\Lambda _{-}^{(-1)}\gamma ^{-1}\right] +e^{-\eta /2}\left[ \psi _{+}^{\left(
+1/2\right) },\gamma \psi _{-}^{\left( -1/2\right) }\gamma ^{-1}\right] , 
\notag \\
\partial _{+}\psi _{-}^{\left( -1/2\right) } &=&e^{-\eta /2}\left[ \Lambda
_{-}^{(-1)},\gamma ^{-1}\psi _{+}^{\left( +1/2\right) }\gamma \right] , 
\notag \\
\partial _{-}\partial _{+}\eta Q &=&0,  \notag
\end{eqnarray}%
are invariant under conformal transformations $x^{+}\rightarrow \widetilde{x}%
^{+}=f(x^{+}),$ $x^{-}\rightarrow \widetilde{x}^{-}=h(x^{-}),$ with the
fields changing in the following manner%
\begin{eqnarray}
\widetilde{\gamma }(\widetilde{x}^{+},\widetilde{x}^{-}) &=&\gamma
(x^{+},x^{-}),  \label{conformal transf.} \\
e^{-\widetilde{\eta }(\widetilde{x}^{+},\widetilde{x}^{-})} &=&\left(
f^{\prime }\right) ^{-1}\left( h^{\prime }\right) ^{-1}e^{-\eta
(x^{+},x^{-})},  \notag \\
e^{-\widetilde{\nu }(\widetilde{x}^{+},\widetilde{x}^{-})} &=&\left(
f^{\prime }\right) ^{\delta }\left( h^{\prime }\right) ^{\delta }e^{-\nu
(x^{+},x^{-})},  \notag \\
\widetilde{\psi }_{+}^{\left( +1/2\right) }(\widetilde{x}^{+},\widetilde{x}%
^{-}) &=&\left( f^{\prime }\right) ^{-1/2}\psi _{+}^{\left( +1/2\right)
}(x^{+},x^{-}),  \notag \\
\widetilde{\psi }_{-}^{\left( -1/2\right) }(\widetilde{x}^{+},\widetilde{x}%
^{-}) &=&\left( h^{\prime }\right) ^{-1/2}\psi _{-}^{\left( -1/2\right)
}(x^{+},x^{-}),  \notag
\end{eqnarray}%
where $\delta $ is arbitrary and $f^{\prime }=\partial _{+}f,$ $h^{\prime
}=\partial _{-}h.$ Under a Lorentz transformation $x^{\pm }\rightarrow 
\widetilde{x}^{\pm }=\xi ^{\pm 1}x^{\pm }$ we can read off the classical
spin of the fields. The bosonic fields are all scalars and the last two
equations of (\ref{conformal transf.}) imply that 
\begin{equation}
\widetilde{\psi }^{\left( \pm 1/2\right) }(\xi x^{+},\xi ^{-1}x^{-})=\xi
^{\pm 1/2}\psi ^{\left( \pm 1/2\right) }(x^{+},x^{-})
\label{fermion lorentz transformations}
\end{equation}%
which instruct us to consider $\psi ^{\left( \pm 1/2\right) }$ as legitimate
two dimensional real spinors (Majorana-Weyl). This result is important
because in the reduction process of superstring sigma models we start with
world-sheet scalars but end up with world-sheet spinors. The reduction
process will be review in chapter 3.

To see how the Lax pair (\ref{Lax operators}) change under Lorentz
transformations we need to use $\widetilde{A}_{\pm }(\xi x^{+},\xi
^{-1}x^{-})=\xi ^{\mp 1}A_{\pm }(x^{+},x^{-})$. Then, we have 
\begin{equation}
\widetilde{L}_{\pm }\left( A;z\right) =\xi ^{\mp 1}L_{\pm }\left( A;\xi
z\right)  \label{Lax pair Lorentz change}
\end{equation}%
and we see that the net effect of a Lorentz transformations is basically a
rescaling of the spectral parameter $z\rightarrow \xi z.$ From this we
conclude that the equations of motion (\ref{NA affine toda equations}) are
Lorentz invariant because they are $z$-independent. This result is also
important in the context of reduced models because this means that equations
of motion for the transverse degrees of freedom of the superstring are
Lorentz invariant.

\subsection{Introducing the supersymmetry flows variations.}

In this section we try to find the supersymmetry transformations associated
to $t_{\pm 1/2}$ that leave invariant the fermionic non-Abelian Toda
equations (\ref{NA affine toda equations}). The outcome, at this stage, is
that we can defined consistent supersymmetry flows only when the gauge group 
$H_{L}\times H_{R}$ is global and there are no gauge fields $A=0$. A deeper
study of these supersymmetric flows, also introduced in \cite{SSSSG
AdS(5)xS(5)}, will be presented elsewhere \cite{us}.

From the non-Abelian flow evolution equations (\ref{flow equations}), we can
associate to the constant grassmanian elements $D^{(\pm 1/2)}=\epsilon
_{i}F_{i}^{(\pm 1/2)}\in \mathcal{K}_{F}^{(\pm 1/2)},$ $i=1,...,\dim 
\mathcal{K}_{F}^{(\pm 1/2)}$ of the fermionic kernel of grade $\pm 1/2,$ i.e$%
\ \mathcal{K}_{F}^{(\pm 1/2)}$, the following two odd Lax variation operators%
\begin{equation}
L_{+1/2}=\delta _{+1/2}-D_{\Theta }^{(+1/2)},\text{ \ \ \ \ \ }%
L_{-1/2}=\delta _{-1/2}+D_{\Pi }^{(-1/2)},  \label{fractional Lax}
\end{equation}%
where, as in (\ref{Lax connections}),(\ref{dressin lax}), we have 
\begin{equation*}
D_{\Theta }^{(+1/2)}=\left( \Theta D^{(+1/2)}\Theta ^{-1}\right) _{\geq 0},%
\text{ \ \ \ \ \ }D_{\Pi }^{(-1/2)}\text{ }=\text{ }\left( \Pi D^{(-1/2)}\Pi
^{-1}\right) _{\leq -1/2}.
\end{equation*}%
and a similar set of operators in the primed variables are obtained from (%
\ref{gauge equivalent variations}). These satisfy the relation $L_{\pm
1/2}^{\prime }=B^{-1}L_{\pm 1/2}B$ and imply that%
\begin{equation}
\delta _{+1/2}BB^{-1}=\left[ \psi ^{(-1/2)}+\theta ^{(-1/2)},D^{(+1/2)}%
\right] ,\text{ \ \ \ \ \ }B^{-1}\delta _{-1/2}B=-\left[ \psi
^{(+1/2)}+\theta ^{(+1/2)},D^{(-1/2)}\right] .  \label{useful relations}
\end{equation}

At first sight, the supersymmetry variations for the fields would be
extracted from the compatibility relations $\left[ L_{\pm 1/2},L_{+}(A)%
\right] =\left[ L_{\pm 1/2},L_{-}(A)\right] =0$ as was done in the case of
the mKdV hierarchy \cite{susyflows-mKdV} (see also \cite{frank}), i.e, when $%
A_{+}^{(R)}=A_{-}^{(L)}=0$ in (\ref{Lax operators}) with $\mathcal{K}%
_{B}^{(0)}=\varnothing .$ Recall that under the local gauge transformations (%
\ref{full gauge transformations}) the Lax operators transforms covariantly $%
\widetilde{L}_{\pm }(A)=\Gamma _{L}L_{\pm }(A)\Gamma _{L}^{-1}$. Then, in
order for the compatibility relations $\left[ L_{\pm 1/2},L_{+}(A)\right] =%
\left[ L_{\pm 1/2},L_{-}(A)\right] =0$ to transform covariantly as well, we
require the operators $L_{\pm 1/2}$ to transform in the same way as $L_{\pm
}(A)$ do, i.e $\widetilde{L}_{\pm 1/2}=$ $\Gamma _{L}L_{\pm 1/2}\Gamma
_{L}^{-1}.$ Let%
\'{}%
s see when this can occur. Under $\widetilde{\Theta }=$ $\Gamma _{L}\Theta
\Gamma _{L}^{-1}$ and $\widetilde{\Pi }=\Gamma _{L}\Pi \Gamma _{R}$ we have%
\begin{equation*}
D_{\widetilde{\Theta }}^{(+1/2)}=\Gamma _{L}\left( \Theta \Gamma
_{L}^{-1}D^{(+1/2)}\Gamma _{L}\Theta ^{-1}\right) _{\geq 0}\Gamma _{L}^{-1},%
\text{ \ \ \ \ \ }D_{\widetilde{\Pi }}^{(-1/2)}=\Gamma _{L}\left( \Pi \Gamma
_{R}^{-1}D^{(-1/2)}\Gamma _{R}\Pi ^{-1}\right) _{\leq -1/2}\Gamma _{L}^{-1}.
\end{equation*}

Then, the first obstruction for a covariant behavior appears when $\Gamma
_{L}^{-1}D^{(+1/2)}\Gamma _{L}\neq D^{(+1/2)}$ and $\Gamma
_{R}^{-1}D^{(-1/2)}\Gamma _{R}\neq D^{(-1/2)}.$ Now, if we assume that there
are elements of $\mathcal{K}_{F}^{(\pm 1/2)}$ commuting with the entire
gauge algebras $\mathfrak{h}_{L}$ and $\mathfrak{h}_{R}$ and that $\Gamma
_{L}$ is invariant under the $\delta _{\pm 1/2}$ flows i.e $\delta _{\pm
1/2}\Gamma _{L}\Gamma _{L}^{-1}=0$, then the desired gauge transformation
holds$.$ If this is the case we have from $\left[ L_{\pm 1/2},L_{+}(A)\right]
=\left[ L_{\pm 1/2},L_{-}(A)\right] =0,$ the following "supersymmetry"
transformations\footnote{%
These transformations when $\theta _{\pm }^{(0)}=0,$ are essentially the
same as the ones proposed by hand in \cite{tseytlin} to be the on-shell
supersymmetry transformations of a Lagrangian formulation of (\ref{NA affine
toda equations}) in the particular case of the Polhmeyer reduction of the $%
AdS_{5}\times S^{5}$ superstring sigma model.}%
\begin{eqnarray}
\left( \delta _{+1/2}BB^{-1}\right) ^{\parallel }{} &=&\left[ \psi ^{\left(
-1/2\right) },D^{(+1/2)}\right] \text{ ,}  \label{wrong I} \\
\delta _{+1/2}\psi _{+}^{\left( +1/2\right) } &=&\left[ \left(
(D_{+}^{(R)}B)B^{-1}\right) ^{\parallel },D^{(+1/2)}\right] +\left[ \theta
_{+}^{(0)},\psi _{+}^{\left( +1/2\right) }\right] \text{ ,}  \notag \\
\delta _{+1/2}\psi _{-}^{\left( -1/2\right) } &=&-\left[ \Lambda
_{-}^{(-1)},B^{-1}D^{(+1/2)}B\right] \text{\ ,}  \notag \\
\delta _{+1/2}A_{-}^{(L)} &=&-\left[ \left( B\psi _{-}^{\left( -1/2\right)
}B^{-1}\right) ^{\perp },D^{(+1/2)}\right] +D_{-}^{(L)}\left( \theta
_{+}^{(0)}\right) ,  \notag \\
\delta _{+1/2}A_{+}^{(R)} &=&0,  \notag
\end{eqnarray}%
and%
\begin{eqnarray}
\left( B^{-1}\delta _{-1/2}B\right) ^{\parallel } &=&-\left[ \psi ^{\left(
+1/2\right) },D^{(-1/2)}\right] \text{ ,}  \label{wrong II} \\
\delta _{-1/2}\psi _{+}^{\left( +1/2\right) } &=&\left[ \Lambda
_{+}^{(+1)},BD^{(-1/2)}B^{-1}\right] \text{ \ ,}  \notag \\
\delta _{-1/2}\psi _{-}^{\left( -1/2\right) } &=&-\left[ \left(
B^{-1}D_{-}^{(L)}B\right) ^{\parallel },D^{(-1/2)}\right] +\left[ \theta
_{-}^{(0)},\psi _{-}^{\left( -1/2\right) }\right] ,  \notag \\
\text{\ }\delta _{-1/2}A_{+}^{(R)} &=&-\left[ \left( B^{-1}\psi _{+}^{\left(
+1/2\right) }B\right) ^{\perp },D^{(-1/2)}\right] +D_{+}^{(R)}\left( \theta
_{-}^{(0)}\right) ,  \notag \\
\delta _{+1/2}A_{-}^{(L)} &=&0,  \notag
\end{eqnarray}%
where $\theta _{\pm }^{(0)}=\left[ \theta ^{(\mp 1/2)},D^{(\pm 1/2)}\right]
, $ we have used the equations of motion (\ref{NA affine toda equations})
and the important assumption that $\left[ D^{(+1/2)},A_{-}^{(L)}\right] =%
\left[ D^{(-1/2)},A_{+}^{(R)}\right] =0$.

The assumption that an element $D\in \mathcal{K}_{F}$ in the fermionic
kernel $\mathcal{K}_{F}$ commute with the entire gauge algebra is too
stringent because the odd part $\mathfrak{g}_{1}\in \mathfrak{g}$ of a given
Lie superalgebra $\mathfrak{g}=\mathfrak{g}_{0}\oplus \mathfrak{g}_{1}$
provides the carrier space for some faithful representation $R$ of the even
part $\mathfrak{g}_{0}\in \mathfrak{g}$ i.e $\left[ \mathfrak{g}_{0},%
\mathfrak{g}_{1}\right] =R(\mathfrak{g}_{0})\mathfrak{g}_{1}.$ Then,
although tempting, the transformations (\ref{wrong I}), (\ref{wrong II}) are
incorrect in origin and the method used with success in the extended mKdV
hierarchy does not apply here in the extended homogeneous hierarchy anymore.
Another problem is related to the local character of $\Gamma _{L}$ and $%
\Gamma _{R}.$ After a gauge transformation, the supersymmetry parameters $%
\widetilde{D}^{(+1/2)}=\Gamma _{L}^{-1}D^{(+1/2)}\Gamma _{L}$ and $%
\widetilde{D}^{(-1/2)}=\Gamma _{R}^{-1}D^{(-1/2)}\Gamma _{R}$ are not
constants and in principle there is no consistent supersymmetry. However,
consistent supersymmetry flows can be defined when the gauge group is global
because the action of $\mathfrak{h}_{L}\times \mathfrak{h}_{R}$ preserves
the fermionic kernel $\mathcal{K}_{F}^{(\pm 1/2)}.$ They simply rotate the
generators $F_{i}^{(\pm 1/2)},$ say $\left[ K_{L}^{(0)},F_{i}^{(+1/2)}\right]
=R\left( K_{L}^{(0)}\right) _{ij}F_{j}^{(+1/2)}$ , which is equivalent to a
linear combination of the constant grasmannian parameters $\epsilon _{i}$
and $\overline{\epsilon }_{j}$. Then, we see that the global supersymmetry
we are dealing with is of the extended type, namely, they include several
fermionic symmetry flows transforming under the gauge algebra. This should
not come as a surprise because the kernel algebra $\mathcal{K=K}_{B}\oplus 
\mathcal{K}_{F}$ is a sub-superalgebra of $\widehat{\mathfrak{f}}=\mathcal{%
K\oplus M}$ anyway, in which the symmetries $\delta _{\mathcal{K}}$ are
generated by $\mathcal{K}$ through the flow equations (\ref{flow equations}%
),(\ref{gauge equivalent variations}). Below, we will use the
Drinfeld-Sokolov procedure to construct a set of $\dim \mathcal{K}_{F}^{(\pm
1/2)}$ non-local fermionic conserved charges associated to the flows $\delta
_{\pm 1/2}$ and transforming under the global part of the gauge group,
arriving to a meaningful result, see (\ref{AKNS supercharges}).

In the on-shell gauge $A_{-}^{(L)}=A_{+}^{(R)}=0$ with global gauge group,
the supersymmetry transformations of the extended homogeneous hierarchy are
given by (\ref{wrong I}), (\ref{wrong II}) (with $A_{-}^{(L)}=A_{+}^{(R)}=0$%
) and this fix the forms of the non-local terms $\theta ^{(\pm 1/2)}$ to be%
\begin{equation}
\theta ^{(\pm 1/2)}=\partial _{\pm }^{-1}\left( B^{\mp 1}\psi _{\pm
}^{\left( \pm 1/2\right) }B^{\pm 1}\right) ^{\perp },  \label{non-local teta}
\end{equation}%
implying also the invariance of the constraints $A_{+}^{(L)}=A_{-}^{(R)}=0,$
i.e (\ref{new constraints}) under the $\delta _{\pm 1/2}$ flows$.$ Below, we
will show that these supersymmetries are Hamiltonian flows under the second
Poisson structure of the extended homogeneous hierarchy, see (\ref{poisson
susy flows}).

\subsection{Lagrangian formulation of the semi-symmetric space sine-Gordon
models.}

Now that we have identified the dynamical fields, we want to introduce the
action functional which have the fermionic Toda equations (\ref{NA affine
toda equations}) as Euler-Lagrange equations of motion. Below, we shall see
that this is the action functional behind the Pohlmeyer reduction process.

\begin{notation}
The light-cone notation used for the flat Minkowski space $\Sigma $ is $%
x^{\pm }=\frac{1}{2}\left( x^{0}\pm x^{1}\right) $, $\partial _{\pm
}=\partial _{0}\pm \partial _{1}$, $\eta _{+-}=\eta _{-+}=2,\eta ^{+-}=\eta
^{-+}=\frac{1}{2},\epsilon _{+-}=-\epsilon _{-+}=2,\epsilon ^{-+}=-\epsilon
^{+-}=\frac{1}{2}$ corresponding to the metric $\eta _{00}=1,\eta _{11}=-1$
and antisymmetric symbol $\epsilon _{10}=-\epsilon _{01}=+1$. A mass scale
is introduced by setting $\Lambda _{\pm }^{(\pm 1)}\rightarrow \mu \Lambda
_{\pm }^{(\pm 1)}$ and $\psi ^{\left( \pm 1/2\right) }\rightarrow \mu
^{-1/2}\psi ^{\left( \pm 1/2\right) }.$
\end{notation}

Consider the action functional%
\begin{eqnarray}
S_{Toda}[B,\psi ] &=&S_{WZNW}[B]-\frac{k}{4\pi }\dint_{\Sigma }\left\langle
\psi _{+}^{\left( +1/2\right) }\partial _{-}\psi ^{\left( -1/2\right) }+\psi
_{-}^{\left( -1/2\right) }\partial _{+}\psi ^{\left( +1/2\right)
}\right\rangle +  \label{Abelian toda action} \\
&&+\frac{k}{2\pi }\dint_{\Sigma }\left\langle \Lambda _{+}^{(+1)}B\Lambda
_{-}^{(-1)}B^{-1}+\psi _{+}^{\left( +1/2\right) }B\psi _{-}^{\left(
-1/2\right) }B^{-1}\right\rangle ,  \notag \\
S_{WZNW}[B] &=&-\frac{k}{4\pi }\left( \dint_{\Sigma }\left\langle
B^{-1}\partial _{+}BB^{-1}\partial _{-}B\right\rangle -\frac{1}{3}%
\dint_{M}\left\langle \left( B^{-1}dB\right) ^{3}\right\rangle \right) ,
\label{WZNW action}
\end{eqnarray}%
which is the action deduced in \cite{reductions} to describe the
supersymmetric sector of the extended super mKdV hierarchy.

In the present situation, i.e in the extended homogeneous hierarchy, this
action is invariant under the global gauge $H_{L}\times H_{R}$ group
transformations (\ref{component gauge trans}).\ Moreover, it is also
invariant under the Kac-Moody-type transformations%
\begin{eqnarray}
\ \widetilde{B} &=&\Gamma _{L}(x^{+})B\Gamma _{R}(x^{-}),
\label{kac moody symmetry} \\
\widetilde{\psi }^{\left( -1/2\right) } &=&\Gamma _{L}(x^{+})\psi ^{\left(
-1/2\right) }\Gamma _{L}^{-1}(x^{+}),  \notag \\
\widetilde{\psi }^{\left( +1/2\right) } &=&\Gamma _{R}^{-1}(x^{-})\psi
^{\left( +1/2\right) }\Gamma _{R}(x^{-})  \notag
\end{eqnarray}%
as can be seen with the help of the Polyakov-Wiegmann identity\footnote{%
This is given by%
\begin{eqnarray*}
S_{WZNW}\left[ ABC\right] &=&S_{WZNW}\left[ A\right] +S_{WZNW}\left[ B\right]
+S_{WZNW}\left[ C\right] - \\
&&-\frac{k}{2\pi }\dint \left\langle \left( A^{-1}\partial _{-}A\right)
\left( \partial _{+}BB^{-1}\right) +\left( B^{-1}\partial _{-}B\right)
\left( \partial _{+}CC^{-1}\right) +\left( A^{-1}\partial _{-}A\right)
B\left( \partial _{+}CC^{-1}\right) B^{-1}\right\rangle .
\end{eqnarray*}%
}.

Taking $\Gamma _{L}(x^{+})=\exp \omega _{L}(x^{+})$ and $\Gamma
_{R}(x^{-})=\exp \omega _{R}(x^{+})$ we find the transformations 
\begin{equation*}
\delta \psi ^{\left( -1/2\right) }=\left[ \omega _{L}(x^{+}),\psi ^{\left(
-1/2\right) }\right] ,\text{ \ \ }\delta BB^{-1}=\omega _{L}(x^{+})+B\omega
_{R}(x^{-})B^{-1},\text{ \ \ }\delta \psi ^{\left( +1/2\right) }=-\left[
\omega _{R}(x^{-}),\psi ^{\left( +1/2\right) }\right] ,
\end{equation*}%
allowing to compute the variation of (\ref{Abelian toda action}). It is
given by%
\begin{equation*}
\frac{2\pi }{k}\delta S_{Toda}[B,\psi ]=\dint \left\langle \omega
_{L}(x^{+})\partial _{-}\left( \partial _{+}BB^{-1}+Q_{+}^{(0)}\right)
+\omega _{R}(x^{-})\partial _{+}\left( B^{-1}\partial
_{-}B+Q_{-}^{(0)}\right) \right\rangle ,
\end{equation*}%
where we have used the jacobi identity, the ad-invariance of the inner
product and $\left[ \Lambda _{\pm }^{(\pm 1)},\mathcal{K}_{B}^{(0)}\right]
=0.$ This implies the existence of chiral currents $\partial
_{-}J_{+}(x^{+})=\partial _{+}J_{-}(x^{-})=0,$ where (cf. (\ref{new
constraints})) 
\begin{equation*}
J_{+}(x^{+})=\mathcal{P}_{\mathfrak{h}_{L}}\left( \partial
_{+}BB^{-1}+Q_{+}^{(0)}\right) ,\text{ \ \ \ \ \ }J_{-}(x^{-})=\mathcal{P}_{%
\mathfrak{h}_{R}}\left( B^{-1}\partial _{-}B+Q_{-}^{(0)}\right) .
\end{equation*}

The action is also invariant under fermionic shifts 
\begin{equation}
\delta \psi ^{\left( \pm 1/2\right) }=\theta ^{(\pm 1/2)}(x^{\mp }),\text{ }
\label{super kac moody symmetry}
\end{equation}%
where $\theta ^{(\pm 1/2)}\in \mathcal{K}_{F}^{(\pm 1/2)}$ and $\delta
BB^{-1}=0$ leading to the following variation 
\begin{equation*}
\frac{4\pi }{k}\delta S_{Toda}[B,\psi ]=\dint \left\langle \theta
^{(+1/2)}(x^{-})\partial _{+}\psi _{-}^{\left( -1/2\right) }+\theta
^{(-1/2)}(x^{+})\partial _{-}\psi _{+}^{\left( +1/2\right) }\right\rangle .
\end{equation*}

Note the strong resemblance with the super Kac-Moody currents obtained from
a supersymmetric WZNW model \cite{di vecchia}. However, their origin are
quite different as fermions in a supersymmetric WZNW model parametrize the
same Lie algebra as the bosons while in the action (\ref{Abelian toda action}%
) they parametrize the odd subspace of a Lie superalgebra i.e bosons and
fermions are in different subspaces. This difference is very important in
our approach to supersymmetry flows and this will be discussed below when we
study the supersymmetry properties of the Pohlmeyer reduction of superstring
sigma models. Clearly, the invariance under (\ref{kac moody symmetry}), (\ref%
{super kac moody symmetry}) and its current algebra deserves a deeper study.

At this point we can identify the fermionic extension of the non-Abelian
Toda models as the Hamiltonian reduction of the phase space of (\ref{Abelian
toda action}) defined by the vanishing of the Kac-Moody-type currents $%
J_{\pm }(x^{\pm })=0$, which are exactly the constraints given by (\ref{new
constraints}). Then, to obtain a Lagrangian formulation for the equations (%
\ref{gauged fix eq of motion}) we need to impose the constraints $J_{\pm
}(x^{\pm })=0$ off-shell and this can be done by considering a local gauge
group $H_{L}\times H_{R}$ and introducing gauge fields.

The Lagrangian formulation of the system (\ref{gauged fix eq of motion}), (%
\ref{new constraints}) exists, as explained in \cite{Miramontes} (see also 
\cite{bakas} for the original formulation), only when the gauge groups $%
H_{L},$ $H_{R}$ are isomorphic to some Lie group $H=\exp \mathfrak{h}$ with
Lie algebra $\mathfrak{h}$ such that $H_{L}=\epsilon _{L}(H),$ $%
H_{R}=\epsilon _{R}(H),$ where $\epsilon _{L},$ $\epsilon _{R}:H\rightarrow
G $ are two group homomorphisms that descend to embeddings of the
corresponding Lie algebras and that satisfy the anomaly free condition 
\begin{equation*}
\left\langle \epsilon _{L}(a)\epsilon _{L}(b)\right\rangle -\left\langle
\epsilon _{R}(a)\epsilon _{R}(b)\right\rangle =0,\text{ \ }\forall \text{ }%
a,b\in \mathfrak{h.}
\end{equation*}

In this case the gauge group is reduced\footnote{%
This reduction can also be seen as a partially gauge fixing of the $%
H_{L}\times H_{R}$ gauge symmetry \cite{tseytlin}.} from $H_{L}\times H_{R}$
to a diagonal subgroup of $H\subset H_{L}\times H_{R}$ and the Lagrangian is
simply given by an appropriate covariantization of (\ref{Abelian toda action}%
) defining now the semi-symmetric space sine-Gordon model (SSSSG) action%
\begin{eqnarray}
S_{SSSSG}[B,\psi ] &=&S_{gWZNW}[B,A]_{G/H}-\frac{k}{4\pi }\dint_{\Sigma
}\left\langle \psi _{+}^{\left( +1/2\right) }D_{-}^{(L)}\psi ^{\left(
-1/2\right) }+\psi _{-}^{\left( -1/2\right) }D_{+}^{(R)}\psi ^{\left(
+1/2\right) }\right\rangle +  \notag \\
&&+\frac{k}{2\pi }\dint_{\Sigma }\left\langle \Lambda _{+}^{(+1)}B\Lambda
_{-}^{(-1)}B^{-1}+\psi _{+}^{\left( +1/2\right) }B\psi _{-}^{\left(
-1/2\right) }B^{-1}\right\rangle ,  \label{NA Toda action}
\end{eqnarray}%
where $D_{-}^{(L)}=\partial _{-}-\left[ \epsilon _{L}\left( A_{-}\right) ,%
\text{ }\right] $, $D_{+}^{(R)}=\partial _{+}-\left[ \epsilon _{R}\left(
A_{+}\right) ,\text{ }\right] $ are the covariant derivatives for the action
of the gauge group $H$, $A_{\pm }\in \mathfrak{h}$ is the gauge field and 
\begin{equation*}
S_{gWZNW}[B,A]_{G/H}=S_{WZNW}[B]_{G}-\frac{k}{2\pi }\dint_{\Sigma
}\left\langle 
\begin{array}{c}
-\epsilon _{L}\left( A_{-}\right) \partial _{+}BB^{-1}+\epsilon _{R}\left(
A_{+}\right) B^{-1}\partial _{-}B- \\ 
-\epsilon _{L}\left( A_{-}\right) B\epsilon _{R}\left( A_{+}\right)
B^{-1}+\epsilon _{L}(A_{+})\epsilon _{L}(A_{-})%
\end{array}%
\right\rangle
\end{equation*}%
is the standard gauged WZNW action. The action (\ref{NA Toda action}) is
invariant under the following gauge transformations%
\begin{eqnarray}
\widetilde{B} &=&\epsilon _{L}(\Gamma )B\epsilon _{R}(\Gamma ^{-1}),
\label{action gauge transf.} \\
\widetilde{\psi }^{(-1/2)} &=&\epsilon _{L}(\Gamma )\psi ^{(-1/2)}\epsilon
_{L}(\Gamma ^{-1}),  \notag \\
\widetilde{\psi }^{(+1/2)} &=&\epsilon _{R}(\Gamma )\psi ^{(+1/2)}\epsilon
_{R}(\Gamma ^{-1}),  \notag \\
\widetilde{A}_{\pm } &=&\Gamma A_{\pm }\Gamma ^{-1}+\partial _{\pm }\Gamma
\Gamma ^{-1},  \notag
\end{eqnarray}%
where $\Gamma _{L}=\epsilon _{L}\left( \Gamma \right) $ and $\Gamma
_{R}=\epsilon _{R}\left( \Gamma ^{-1}\right) $ in (\ref{kac moody symmetry})
are local elements and we have used (\ref{connection transformation}) in the
last line.

An arbitrary variation of (\ref{NA Toda action}) is given by%
\begin{eqnarray}
\frac{2\pi }{k}\delta S_{SSSSG}[B,\psi ] &=&\dint_{\Sigma }\left\langle
\left( \delta BB^{-1}-B\delta \psi ^{\left( +1/2\right) }B^{-1}-\delta \psi
^{\left( -1/2\right) }\right) F_{+-}\right\rangle +  \notag \\
&&+\dint_{\Sigma }\left\langle \delta A_{-}^{(L)}\left\{
-A_{+}^{(L)}+\partial _{+}BB^{-1}+BA_{+}^{(R)}B^{-1}+Q_{+}^{(0)}\right\}
\right\rangle +  \notag \\
&&+\dint_{\Sigma }\left\langle \delta A_{+}^{(R)}\left\{
-A_{-}^{(R)}-B^{-1}\partial _{-}B+B^{-1}A_{-}^{(L)}B-Q_{-}^{(0)}\right\}
\right\rangle ,  \label{arbitrary variation}
\end{eqnarray}%
where the curvature components $F_{+-}^{(0)}$ and $F_{+-}^{(\pm 1/2)}$ are
given in (\ref{NA affine toda equations}) but now with $A_{-}^{(L)}=\epsilon
_{L}\left( A_{-}\right) $ and $A_{+}^{(R)}=\epsilon _{R}\left( A_{+}\right)
. $ Then, the Lax pair associated to the action (\ref{NA Toda action}) is
simply the reduction of (\ref{Lax operators}), i.e%
\begin{eqnarray}
L_{+}(A) &=&\partial _{+}-\partial _{+}BB^{-1}-B\epsilon _{R}\left(
A_{+}\right) B^{-1}+\psi _{+}^{\left( +1/2\right) }+\Lambda _{+}^{(+1)},
\label{Lax pair for action} \\
L_{-}(A) &=&\partial _{-}-\epsilon _{L}\left( A_{-}\right) -B\left( \psi
_{-}^{\left( -1/2\right) }+\Lambda _{-}^{(-1)}\right) B^{-1},  \notag
\end{eqnarray}%
supplemented by the constraints given by the $A_{\pm }$ equations of motion.
It is not difficult to see that the last two equations of motion in (\ref%
{arbitrary variation}) and the definitions of the missing components (\ref%
{missing componentes}) are the same in this case, and that the two equations
(\ref{flat connections}) reduce to the flatness of the only gauge field
involved $A_{\pm }$ 
\begin{equation*}
\left[ \partial _{+}-A_{+},\partial _{-}-A_{-}\right] =0,
\end{equation*}%
which enables the on-shell gauge $A_{+}=A_{-}=0,$ as in the last line of (%
\ref{gauged fix eq of motion}).

It is important to mention that the $A_{\pm }$ equations of motion can also
be interpreted as a partial gauge fixing of the $H_{L}\times H_{R}$ gauge
symmetry in the definitions (\ref{missing componentes}), see \cite{tseytlin}
for the details.

\subsection{The Drinfeld-Sokolov procedure: the $\dim \mathcal{K}_{F}^{(\pm
1/2)}$ 2D spin $\pm 1/2$ supercharges.}

In section 2.3 we have found some difficulties in finding the supersymmetry
flow variations for the fields due to the local character of the gauge group 
$H_{L}\times H_{R}$. However, a consistent set of transformations were
singled out and we shall return to them later in section 2.7. Here we focus
on the construction of the fermionic conserved charges associated to the
symmetry flows $t_{\pm 1/2}$ in the general situation when the gauge group $%
H_{L}\times H_{R}$ is local.

One of the advantages of the dressing approach adopted above, is that we can
apply the Drinfeld-Sokolov (DS) procedure and viceversa. It provides a
systematic method for constructing the local and non-local conserved charges
associated to the symmetry flows (\ref{flow equations}),(\ref{gauge
equivalent variations}). Inspired by a similar computation done in \cite%
{mira-hollowood}, here we will extract the vector, spinor and tensor
conserved currents associated to the sub-sector $\left( t_{0},t_{\pm
1/2},t_{\pm 1}\right) $.

From (\ref{flow equations}), (\ref{gauge equivalent variations}) but now in
the presence of gauge fields, we have the following dressing relations%
\begin{eqnarray}
L_{+}(A) &=&\Theta \left( \partial _{+}+\Lambda _{+}^{(+1)}\right) \Theta
^{-1}\text{, \ \ \ \ \ \ }L_{-}(A)\text{ }=\text{ }\Theta \left( \partial
_{-}\right) \Theta ^{-1},  \label{dressing} \\
L_{+}^{\prime }(A) &=&\Pi ^{\prime }(\partial _{+})\Pi ^{\prime -1},\text{ \
\ \ \ \ \ \ \ \ \ \ \ \ \ \ \ \ }L_{-}^{\prime }(A)\text{ }=\text{ }\Pi
^{\prime }(\partial _{-}-\Lambda _{-}^{(-1)})\Pi ^{\prime -1},  \notag
\end{eqnarray}%
associated to the Lax operators (\ref{Lax operators}), (\ref{prime Lax
operators}).

The dressing matrix $\Theta $ factorizes as $\Theta =U_{-}S_{-},$ where $%
U_{-}\in \exp u_{-},$ $u_{-}=\oplus _{r\leq -1/2}$ $\widehat{\mathfrak{f}}%
_{r}^{\parallel }\subset \mathcal{M}$ is a local functional of the fields
and $S_{-}\in \exp s_{-},$ $s_{-}=\oplus _{r\leq -1/2}$ $\widehat{\mathfrak{f%
}}_{r}^{\perp }\subset \mathcal{K}$ is a non-local functional of the fields,
splitting the dressing of the vacuum Lax operators, i.e $L_{\pm }=\Theta
L_{\pm }^{V}\Theta ^{-1},$ as a two step process \cite{symmetry flows}. An $%
U_{-}$ and an $S_{-}$ rotation given respectively by 
\begin{eqnarray}
U_{-}^{-1}L_{+}(A)U_{-} &=&\partial _{+}+\Lambda _{+}^{(+1)}+K_{+}^{(-)},%
\text{ \ \ \ \ \ \ \ }U_{-}^{-1}L_{-}(A)U_{-}\text{ }=\text{ }\partial
_{-}+K_{-}^{(-)},  \label{U} \\
\partial _{+}+\Lambda _{+}^{(+1)}+K_{+}^{(-)} &=&S_{-}\left( \partial
_{+}+\Lambda _{+}^{(+1)}\right) S_{-}^{-1},\text{ \ \ \ \ \ \ \ \ }\partial
_{-}+K_{-}^{(-)}\text{ }=\text{ }S_{-}\left( \partial _{-}\right) S_{-}^{-1},
\label{S}
\end{eqnarray}%
where $K_{\pm }^{(-)}$ $=\Sigma _{r\leq 0}K_{\pm }^{(r)},$ $K_{\pm
}^{(r)}\in \mathcal{K}$ are the conserved current components we want to find.

Similarly, for the second line in (\ref{dressing}) we have $\Pi ^{\prime
}=U_{+}S_{+},$ where $U_{+}\in \exp u_{+},$ $u_{+}=\oplus _{r\geq +1/2}$ $%
\widehat{\mathfrak{f}}_{r}^{\parallel }\subset \mathcal{M}$ is local and $%
S_{+}\in \exp s_{+},$ $s_{+}=\oplus _{r\geq +1/2}$ $\widehat{\mathfrak{f}}%
_{r}^{\perp }\subset \mathcal{K}$ is non-local in the fields. Then, we have%
\begin{eqnarray}
U_{+}^{-1}L_{+}^{\prime }(A)U_{+} &=&\partial _{+}+K_{+}^{(+)},\text{ \ \ \
\ \ \ \ \ \ \ \ \ \ }U_{+}^{-1}L_{-}^{\prime }(A)U_{+}\text{ }=\text{ }%
\partial _{-}-\Lambda _{-}^{(-1)}+K_{-}^{(+)},  \label{U gauge} \\
\partial _{+}+K_{+}^{(+)} &=&S_{+}\left( \partial _{+}\right) S_{+}^{-1},%
\text{ \ \ \ \ \ }\partial _{-}-\Lambda _{-}^{(-1)}+K_{-}^{(+)}\text{ }=%
\text{ }S_{-}\left( \partial _{-}-\Lambda _{-}^{(-1)}\right) S_{-}^{-1},
\label{S gauge}
\end{eqnarray}%
where $K_{\pm }^{(+)}$ =$\Sigma _{r\geq 0}K_{\pm }^{(r)},$ $K_{\pm
}^{(r)}\in \mathcal{K}$.

The conservation laws are extracted by projecting the zero curvature
conditions of (\ref{U}) and (\ref{U gauge}) along the kernel subspace $%
\mathcal{K}$ grade by grade. This is roughly speaking, the Drinfeld-Sokolov
procedure. Of course, it provides local conservation laws for all the
isospectral flows $t_{\pm n}.$

It is important to mention that (\ref{U}) and (\ref{U gauge}) are in
canonical form, in the sense that $\Theta =U_{-}S_{-}$ and $\Pi ^{\prime
}=U_{+}S_{+}$ are splitted as local and non-local pieces allowing to obtain $%
K_{\pm }^{(-)},K_{\pm }^{(+)}$ as functionals of the components of $U_{\pm }$
only. However, these relations are subject to an ambiguity induced by the
gauge transformations $U_{\pm }\rightarrow U_{\pm }\widetilde{S}_{\pm }$
with $\widetilde{S}_{\pm }$ parametrized in the same way as $S_{\pm }.$ This
action does not change the LHS of (\ref{U}) and (\ref{U gauge}) but changes
the RHS side. For example, for (\ref{U}) we have 
\begin{equation}
U_{-}^{-1}L_{+}(A)U_{-}=\widetilde{S}_{-}\left( \partial _{+}+\Lambda
_{+}^{(+1)}+\widetilde{K}_{+}^{(-)}\right) \widetilde{S}_{-}^{-1},\text{ \ \
\ \ \ }U_{-}^{-1}L_{-}(A)U_{-}\text{ }=\widetilde{S}_{-}\left( \text{ }%
\partial _{-}+\widetilde{K}_{-}^{(-)}\right) \widetilde{S}_{-}^{-1}
\label{U ambiguous}
\end{equation}%
and a similar expression for (\ref{U gauge}) leading to a non-local gauge
transformation between the current components. This observation will be
useful below because allow for an explanation of an apparent discrepancy
between the non-local supersymmetry flow variations as presented above in (%
\ref{wrong I}), (\ref{wrong II}) in the gauge $A^{(L/R)}=0$ and the local
supersymmetry flow variations induced by the canonical form of the
supercharges extracted from (\ref{U}), (\ref{U gauge}) in the gauge $%
A^{(L/R)}=0$. See (\ref{mKdV supercharges}) and proposition 4 below.

The components of grades $0,-1/2$ and $-1$ of the first equation of (\ref{U}%
) along the kernel $\mathcal{K}$, are given by%
\begin{eqnarray*}
0 &:&K_{+}^{(0)}=-A_{+}^{(L)}, \\
-1/2 &:&K_{+}^{(-1/2)}=\left[ \psi ^{\left( -1/2\right) },\widehat{A}%
_{+}^{(L)}\right] ^{\perp }, \\
-1 &:&\left\langle \Lambda _{+}^{(+1)},K_{+}^{(-1)}\right\rangle =+\frac{1}{2%
}\left\langle \left( \widehat{A}_{+}^{(L)}\right) ^{2}-\left(
A_{+}^{(L)}\right) ^{2}\right\rangle +\frac{1}{2}\left\langle \psi
_{+}^{\left( +1/2\right) }D_{+}^{(L)}\psi ^{\left( -1/2\right)
}\right\rangle +\left\langle \left( Q_{+}^{(0)}\right) ^{2}\right\rangle ,
\end{eqnarray*}%
where we have taken $u^{(-1/2)}=\psi ^{\left( -1/2\right) },$ $u^{(-1)}=\psi
^{\left( -1\right) }$ and used the definitions (\ref{missing componentes})
and (\ref{Baker-Hausdorf fields})$.$ In the last relation we projected along 
$\Lambda _{+}^{(+1)}$ in order to simplify expressions. The components of
grades $0,-1/2$ and $-1$ of the second equation of (\ref{U}) along the
kernel $\mathcal{K}$, are given by%
\begin{eqnarray*}
0 &:&K_{-}^{(0)}=-A_{-}^{(L)}, \\
-1/2 &:&K_{-}^{(-1/2)}=-\left( B\psi _{-}^{\left( -1/2\right) }B^{-1}\right)
^{\perp }, \\
-1 &:&\left\langle \Lambda _{+}^{(+1)},K_{-}^{(-1)}\right\rangle
=-\left\langle \Lambda _{+}^{(+1)}B\Lambda _{-}^{(-1)}B^{-1}+\frac{1}{2}\psi
_{+}^{\left( +1/2\right) }B\psi _{-}^{\left( -1/2\right)
}B^{-1}\right\rangle ,
\end{eqnarray*}%
where in the last relation we have used the equations of motion for $\psi
^{\left( -1/2\right) }.$

In a complete analogous way, we have for the first equation of (\ref{U gauge}%
) along the kernel subspaces of grades $0,+1/2$ and $+1$, the following
components%
\begin{eqnarray*}
0 &:&K_{+}^{\prime (0)}=-A_{+}^{(R)}, \\
+1/2 &:&K_{+}^{(+1/2)}=(B^{-1}\psi _{+}^{\left( +1/2\right) }B)^{\perp }, \\
+1 &:&\left\langle \Lambda _{-}^{(-1)},K_{+}^{(+1)}\right\rangle
=-\left\langle \Lambda _{+}^{(+1)}K_{-}^{(-1)}\right\rangle ,
\end{eqnarray*}%
where $u^{(+1/2)}=-\psi ^{\left( +1/2\right) }$ and $u^{(+1)}=-\psi ^{\left(
+1\right) }$. The second equation of (\ref{U gauge}) provides%
\begin{eqnarray*}
0 &:&K_{-}^{\prime (0)}=-A_{-}^{(R)}, \\
+1/2 &:&K_{-}^{(+1/2)}=-\left[ \psi ^{\left( +1/2\right) },\widehat{A}%
_{-}^{(R)}\right] ^{\perp }, \\
+1 &:&\left\langle \Lambda _{-}^{(-1)},K_{-}^{(+1)}\right\rangle =-\frac{1}{2%
}\left\langle \left( \widehat{A}_{-}^{(R)}\right) ^{2}-\left(
A_{-}^{(R)}\right) ^{2}\right\rangle -\frac{1}{2}\left\langle \psi
_{-}^{\left( -1/2\right) }D_{-}^{(R)}\psi ^{\left( +1/2\right)
}\right\rangle -\left\langle \left( Q_{-}^{(0)}\right) ^{2}\right\rangle .
\end{eqnarray*}

The zero curvatures of the LHS of (\ref{U}) and (\ref{U gauge}) imply%
\footnote{%
The advantage of the DS expressions (\ref{Drinfeld Sokolov}) is that the
process of getting conserved charges can be continued infinitely without
invoking an action functional and without reducing the gauge symmetry to a
diagonal subgroup of $H_{L}\times H_{R}$ as required by a Lagrangian
formulation.}%
\begin{eqnarray}
\partial _{+}K_{-}^{(-)}-\partial _{-}K_{+}^{(-)}+\left[
K_{+}^{(-)},K_{-}^{(-)}\right] &=&0,  \label{Drinfeld Sokolov} \\
\partial _{+}K_{-}^{(+)}-\partial _{-}K_{+}^{(+)}+\left[
K_{+}^{(+)},K_{-}^{(+)}\right] &=&0.  \notag
\end{eqnarray}%
Decomposing the first equation of (\ref{Drinfeld Sokolov}) in terms of $%
0,-1/2$ and $-1$ grades we get 
\begin{eqnarray}
0 &:&\left[ \partial _{+}-A_{+}^{(L)},\partial _{-}-A_{-}^{(L)}\right] =0,
\label{negative conservations} \\
-1/2 &:&D_{+}^{(L)}K_{-}^{(-1/2)}-D_{-}^{(L)}K_{+}^{(-1/2)}=0,  \notag \\
-1 &:&\partial _{-}\left( \left\langle \Lambda
_{+}^{(+1)},K_{+}^{(-1)}\right\rangle \right) +\partial _{+}\left(
-\left\langle \Lambda _{+}^{(+1)},K_{-}^{(-1)}\right\rangle \right) =0. 
\notag
\end{eqnarray}%
Similarly, the second equation of (\ref{Drinfeld Sokolov}) in terms of $%
0,+1/2$ and $+1$ grades gives%
\begin{eqnarray}
0 &:&\left[ \partial _{+}-A_{+}^{(R)},\partial _{-}-A_{-}^{(R)}\right] =0,
\label{positive conservations} \\
+1/2 &:&D_{+}^{(R)}K_{-}^{(+1/2)}-D_{-}^{(R)}K_{+}^{(+1/2)}=0,  \notag \\
+1 &:&\partial _{-}\left( \left\langle \Lambda
_{-}^{(-1)},K_{+}^{(+1)}\right\rangle \right) +\partial _{+}\left(
-\left\langle \Lambda _{-}^{(-1)},K_{-}^{(+1)}\right\rangle \right) =0. 
\notag
\end{eqnarray}%
The first two equations of (\ref{negative conservations}) and (\ref{positive
conservations}) are gauge covariant and the last ones are local and gauge
invariant, as expected.

Let us study first the third equations. Defining the stress-tensor
components of $T_{\mu \nu }$ as 
\begin{equation*}
T_{++}=\left\langle \Lambda _{+}^{(+1)},K_{+}^{(-1)}\right\rangle ,\text{ }%
T_{-+}=-\left\langle \Lambda _{+}^{(+1)},K_{-}^{(-1)}\right\rangle ,\text{ }%
T_{+-}=\left\langle \Lambda _{-}^{(-1)},K_{+}^{(+1)}\right\rangle ,\text{ }%
T_{--}=-\left\langle \Lambda _{-}^{(-1)},K_{-}^{(+1)}\right\rangle .
\end{equation*}%
The last equations of (\ref{negative conservations}), (\ref{positive
conservations}) become $\partial _{-}T_{++}+\partial _{+}T_{-+}=\partial
_{-}T_{+-}+\partial _{+}T_{--}=0$ the conservation laws for the components%
\begin{eqnarray}
T_{++} &=&+\frac{1}{2}\left\langle \left( \widehat{A}_{+}^{(L)}\right)
^{2}-\left( A_{+}^{(L)}\right) ^{2}\right\rangle +\frac{1}{2}\left\langle
\psi _{+}^{\left( +1/2\right) }D_{+}^{(L)}\psi ^{\left( -1/2\right)
}\right\rangle +\left\langle \left( Q_{+}^{(0)}\right) ^{2}\right\rangle , 
\notag \\
T_{--} &=&+\frac{1}{2}\left\langle \left( \widehat{A}_{-}^{(R)}\right)
^{2}-\left( A_{-}^{(R)}\right) ^{2}\right\rangle +\frac{1}{2}\left\langle
\psi _{-}^{\left( -1/2\right) }D_{-}^{(R)}\psi ^{\left( +1/2\right)
}\right\rangle +\left\langle \left( Q_{-}^{(0)}\right) ^{2}\right\rangle , 
\notag \\
T_{+-} &=&T_{-+}=+\left\langle \Lambda _{+}^{(+1)}B\Lambda _{-}^{(-1)}B^{-1}+%
\frac{1}{2}\psi _{+}^{\left( +1/2\right) }B\psi _{-}^{\left( -1/2\right)
}B^{-1}\right\rangle .  \label{Tensor}
\end{eqnarray}%
The conserves charges\footnote{%
We have used the change of basis 
\begin{equation*}
4T_{00}=T_{++}+T_{--}+2T_{+-},\text{ }4T_{01}=T_{++}-T_{--},\text{ }%
T_{10}=T_{01},\text{ }T_{11}=T_{++}+T_{--}-2T_{+-}
\end{equation*}%
and set $T_{00}=H$ and $T_{01}=P.$} for these grade $\pm 1$ equations are%
\begin{equation}
\left( H+P\right) =\frac{1}{2}\dint_{-\infty }^{+\infty }dx^{1}\left(
T_{++}+T_{-+}\right) \text{, \ \ \ \ \ }\left( H-P\right) =\frac{1}{2}%
\dint_{-\infty }^{+\infty }dx^{1}\left( T_{--}+T_{+-}\right) .
\label{H and P}
\end{equation}

Now, we analyze the fermionic equations which are the ones we are mainly
interested. They are covariant but not gauge invariant and at first sight
provide no conservations laws. However, the connections $A_{\mu }^{(L/R)}$
are flat and can be written in the pure gauge form%
\begin{equation*}
A_{\pm }^{(L)}=\partial _{\pm }g_{L}g_{L}^{-1}\text{, \ \ \ \ \ }A_{\pm
}^{(R)}=\partial _{\pm }g_{R}g_{R}^{-1}.
\end{equation*}%
The gauge transformations (\ref{connection transformation}) becomes now%
\begin{equation*}
\widetilde{g}_{L}=\Gamma _{L}g_{L}\Gamma _{cL}^{-1},\text{ \ \ \ \ \ }%
\widetilde{g}_{R}=\Gamma _{R}^{-1}g_{R}\Gamma _{cR},
\end{equation*}%
where $\Gamma _{cL/R}$ are constant elements of the global part of the gauge
group $H_{L}\times H_{R}.$

The flatness of the gauge fields allow to write 
\begin{equation*}
D_{\pm }^{(L)}K_{\mp }^{(-1/2)}=g_{L}\partial _{\pm }\left( g_{L}^{-1}K_{\mp
}^{(-1/2)}g_{L}\right) g_{L}^{-1}\text{, \ \ \ \ \ }D_{\pm }^{(R)}K_{\mp
}^{(+1/2)}=g_{R}\partial _{\pm }\left( g_{R}^{-1}K_{\mp
}^{(+1/2)}g_{R}\right) g_{R}^{-1}.
\end{equation*}%
Then, we can transform the $\pm 1/2$ grade equations into the following
non-local fermionic conservation laws 
\begin{eqnarray*}
\partial _{-}\left( g_{L}^{-1}K_{+}^{(-1/2)}g_{L}\right) +\partial
_{+}\left( -g_{L}^{-1}K_{-}^{(-1/2)}g_{L}\right) &=&0, \\
\partial _{+}\left( g_{R}^{-1}K_{-}^{(+1/2)}g_{R}\right) +\partial
_{-}\left( -g_{R}^{-1}K_{+}^{(+1/2)}g_{R}\right) &=&0,
\end{eqnarray*}%
because of the presence of the Wilson lines%
\begin{equation}
g_{L/R}(p)=P\exp \left( \dint\limits_{p_{0}}^{p}dx^{\mu }A_{\mu
}^{(L/R)}\right) ,  \label{Wilson lines}
\end{equation}%
where $p=(t,x)\in \Sigma $ is an arbitrary point and $p_{0}\in \Sigma $ is a
fixed reference point\footnote{%
Under gauge transformations we have $\Gamma _{cL/R}=$ $\Gamma _{L/R}(p_{0}).$%
}.

These equations are gauge invariant under the local part of the gauge group
because under (\ref{component gauge trans})$,$ we have%
\begin{equation*}
\widetilde{K}_{\pm }^{(-1/2)}=\Gamma _{L}K_{\pm }^{(-1/2)}\Gamma _{L}^{-1},%
\text{ \ \ \ \ \ }\widetilde{K}_{\pm }^{(+1/2)}=\Gamma _{R}^{-1}K_{\pm
}^{(+1/2)}\Gamma _{R}
\end{equation*}%
but transforms under the global part of it by conjugations with $\Gamma
_{L/R}^{c}.$ However, as we discussed before (section 2.3) the global part
of the gauge group preserve the fermionic kernel and it is in this sense
that we have well defined conserved charges. They are given by the following 
$\dim \mathcal{K}_{F}^{(\pm 1/2)}$ non-local conserved supercharges $%
Q(\delta _{\pm 1/2})$ associated to the $\delta _{\pm 1/2}$ symmetry flows%
\begin{eqnarray}
Q\left( \delta _{+1/2}\right) &=&\dint\nolimits_{-\infty }^{+\infty
}dx^{1}G\left( \delta _{+1/2}\right) =Q_{i}^{+}F_{i}^{(-1/2)},
\label{AKNS supercharges} \\
G\left( \delta _{+1/2}\right) &=&g_{L}^{-1}\left( \left[ \psi ^{(-1/2)},%
\widehat{A}_{+}^{(L)}\right] +B\psi _{-}^{(-1/2)}B^{-1}\right) ^{\perp
}g_{L},  \notag \\
Q\left( \delta _{-1/2}\right) &=&\dint\nolimits_{-\infty }^{+\infty
}dx^{1}G\left( \delta _{-1/2}\right) =Q_{i}^{-}F_{i}^{(+1/2)},  \notag \\
G\left( \delta _{-1/2}\right) &=&g_{R}^{-1}\left( -\left[ \psi ^{(+1/2)},%
\widehat{A}_{-}^{(R)}\right] -B^{-1}\psi _{+}^{(+1/2)}B\right) ^{\perp
}g_{R},  \notag
\end{eqnarray}%
where $i=1,...,\dim \mathcal{K}_{F}^{(\pm 1/2)}$. The action of $\Gamma
_{cL/R}$ on the charges is given by 
\begin{equation*}
\widetilde{Q}(\delta _{+1/2})=\Gamma _{cL}^{-1}Q(\delta _{+1/2})\Gamma
_{cL}\in \mathcal{K}_{F}^{(-1/2)}\text{, \ \ \ \ \ }\widetilde{Q}(\delta
_{-1/2})=\Gamma _{cR}Q(\delta _{-1/2})\Gamma _{cR}^{-1}\in \mathcal{K}%
_{F}^{(+1/2)}
\end{equation*}%
and, as mentioned above, we expect to obtain an extended global symmetry
superalgebra $\mathcal{K\subset \widehat{\mathfrak{f}}\rightarrow \delta }_{%
\mathcal{K}}$ $\mathfrak{.}$

From the analysis around (\ref{fermion lorentz transformations}) we can
compute the 2D spin for all the conserved charges extracted from (\ref%
{Drinfeld Sokolov}). In particular, from (\ref{negative conservations}), (%
\ref{positive conservations}) we confirm that $G(\delta _{\pm 1/2})$ are
indeed 2D spinorial currents because the power of the spectral parameter is
half-integer$,$ i.e $z^{\pm 1/2}.$

The most interesting situation to be considered and which is related to the
Pohlmeyer reductions of superstring sigma models, is when the gauge group $%
H_{L}\times H_{R}$ is reduced to $H$ and the only gauge field $A_{\pm }$
involved obey the equations of motion provided by (\ref{arbitrary variation}%
). In this case the gauge field is also flat but their components are
functions of the dynamical fields $B,\psi ^{\left( \pm 1/2\right) }$ turning
the conjugations in $Q(\delta _{\pm 1/2})$ with the Wilson lines (\ref%
{Wilson lines}) non-trivial. This situation will be addressed elsewhere
because we have found some difficulties in trying to obtain the field
variations from these non-local charges. Fortunately, in the on-shell gauge $%
A_{\pm }^{(L/R)}=0,$ the symmetries induced by (\ref{AKNS supercharges}) are
symmetries of the field configurations that solve the equations of motion.
Below in the examples, we will ignore the Wilson lines and explore a little
bit the relation between $Q_{i}^{\pm }$ and some well-known results obtained
from superspace in order to motivate further the study of these new 2D
supersymmetries.

Now, in the following we restrict ourselves to the on-shell gauge $A_{\pm
}^{(L/R)}=0$ in which sharper statements can be made$.$ In this gauge we have%
\begin{eqnarray}
Q(\delta _{+1/2}) &=&\dint_{-\infty }^{+\infty }dx^{1}\left( \left[ \psi
^{\left( -1/2\right) },\partial _{+}BB^{-1}\right] +B\psi _{-}^{\left(
-1/2\right) }B^{-1}\right) ^{\perp },  \label{mKdV supercharges} \\
Q(\delta _{-1/2}) &=&\dint_{-\infty }^{+\infty }dx^{1}\left( \left[ \psi
^{\left( +1/2\right) },B^{-1}\partial _{-}B\right] -B^{-1}\psi _{+}^{\left(
+1/2\right) }B\right) ^{\perp }.  \notag
\end{eqnarray}%
Below, in section 2.7, we shall show by using Poisson brackets, that these
supercharges generate the supersymmetry transformations (\ref{wrong I}), (%
\ref{wrong II}) with $A_{\pm }^{(L/R)}=0,$ showing that they are fermionic
Hamiltonian flows on the phase space of the system.

When there are no gauge symmetries at all, i.e $\mathcal{K}%
_{B}^{(0)}=\varnothing ,$ we recover the supercharges of \cite%
{susyflows-mKdV} for the extended super mKdV hierarchy. The supersymmetry
transformations induced by (\ref{mKdV supercharges}) in this case, are
exactly of the same form (\ref{wrong I}), (\ref{wrong II}) with $A_{\pm
}^{(L/R)}=\theta ^{(\pm 1/2)}=0$ and are symmetries of the action functional
(\ref{Abelian toda action}). The associated Noether conserved charges are
exactly (\ref{mKdV supercharges}) and the same occurs here when $\mathcal{K}%
_{B}^{(0)}\neq \varnothing ,$ with a global gauge group, see (\ref{1/2
variation}) below. Thus, we conclude that the supercharges extracted from
the Drinfeld-Sokolov and Noether procedures coincide.

Finally, we consider the grade zero equations of (\ref{negative
conservations}) and (\ref{positive conservations}). They provide the
following non-local conserved charges%
\begin{equation*}
\Omega _{L/R}=P\exp \left( \frac{1}{2}\dint_{-\infty }^{+\infty
}dx^{1}\left( A_{+}^{(L/R)}-A_{-}^{(L/R)}\right) \right)
\end{equation*}%
and in the case when the gauge algebras $\mathfrak{h}_{L/R}$ are Abelian,
they reduce to 
\begin{equation}
Q_{L/R}=\frac{1}{2}\dint_{-\infty }^{+\infty
}dx^{1}(A_{+}^{(L/R)}-A_{-}^{(L/R)}).  \label{abelian charges}
\end{equation}%
Below, in the examples, we will see that these conserved charges encode some
symmetry data of the target space of the SSSSG model action (\ref{NA Toda
action}).

We will end this section by comparing the supersymmetry involved in the
supersymmetrization of a WZNW model based on an ordinary Lie algebra with
the supersymmetry flows (SF) involved in the formulation of the SSSSG models
(\ref{NA Toda action}), (\ref{Lax pair for action}). \ In the superspace
WZNW based on a Lie algebra $\mathfrak{g,}$ the bosonic fields are replaced
by scalar superfields, as a consequence, bosons $\phi $ and fermions $\psi $
are both parametrized by the same elements in $\mathfrak{g}$ i.e $\phi =\phi
^{i}T_{i},$ $\psi =\psi ^{i}T_{i}$, where $T_{i}$ are the generators of $%
\mathfrak{g.}$ In our case the situation is very different, the physical
fields parametrize the image part $\mathcal{M}$ of a superalgebra in the
decomposition $\widehat{\mathfrak{f}}$ $\mathfrak{=}$ $\mathcal{K+M}$, while
the symmetries $\delta _{\mathcal{K}}$ are generated by the kernel part $%
\mathcal{K}$ of it. Due to the fact that $\left[ \mathcal{K}\text{,}\mathcal{%
M}\right] \subset \mathcal{M}$, we have a map from physical fields to
physical fields $\delta _{\mathcal{K}}:\mathcal{M\rightarrow M}$.
Decomposing $\mathcal{M=M}_{B}\mathcal{\oplus M}_{F}$ in bosonic and
fermionic parts and $\mathcal{K=K}_{B}\mathcal{\oplus K}_{F}$ in a similar
way, we see that a supersymmetry flow obeys $\delta _{SF}:$ $\mathcal{M}%
_{F}\rightarrow \mathcal{M}_{B}$ and $\delta _{SF}:$ $\mathcal{M}%
_{B}\rightarrow \mathcal{M}_{F},$ mapping the odd part into the even one and
viceversa$.$ In a supermatrix representation, such a map is roughly of the
form%
\begin{eqnarray*}
\delta _{SF} &:&%
\begin{pmatrix}
0 & \psi \\ 
\psi ^{\prime } & 0%
\end{pmatrix}%
\in \mathcal{M}_{F}\rightarrow 
\begin{pmatrix}
\delta (\psi ,\psi ^{\prime }) & 0 \\ 
0 & \delta ^{\prime }(\psi ,\psi ^{\prime })%
\end{pmatrix}%
\in \mathcal{M}_{B} \\
\delta _{SF} &:&%
\begin{pmatrix}
\phi & 0 \\ 
0 & \phi ^{\prime }%
\end{pmatrix}%
\in \mathcal{M}_{B}\rightarrow 
\begin{pmatrix}
0 & \delta (\phi ,\phi ^{\prime }) \\ 
\delta ^{\prime }(\phi ,\phi ^{\prime }) & 0%
\end{pmatrix}%
\in \mathcal{M}_{F}.
\end{eqnarray*}

Of course we have to guarantee that the number of bosonic and fermionic
generators in $\mathcal{M}$ match in the appropriate way for $\delta _{%
\mathcal{K}}$ to be considered as a supersymmetry. Fortunately, this is
guaranteed by the finer $%
\mathbb{Z}
_{4}$ grading decomposition entering the Lie algebraic structure of the
integrable hierarchy, which was defined if (\ref{half-integer expansion}).
Recall also that the gauge group $H$ have actions $\delta _{H}$ which mixes
with $\delta _{SF},$ i.e $\left[ \delta _{H},\delta _{SF}\right] =\delta
_{SF}^{\prime }.$

\subsection{Bi-Hamiltonian structure of the extended homogeneous hierarchy.}

In this section we introduce two Hamiltonian structures associated to the
extended homogeneous hierarchy. They can be extracted directly from the Lax
operators but to achieve this it is useful first to introduce the notion of
differential of a functional on a co-adjoint orbit $\Xi $. In this section
and the next, we assume that the constraints (\ref{new constraints}) are
satisfied\footnote{%
In particular, this happens for the soliton solutions constructed from the
dressing method, e.g see \cite{SSSSG AdS(5)xS(5)}.}.

Start by introducing the integrated inner product%
\begin{equation}
\left( X,Y\right) =\dint_{-\infty }^{+\infty }dx^{1}\left\langle
X,Y\right\rangle ,  \label{space inner product}
\end{equation}%
with $\left\langle X,Y\right\rangle _{\widehat{\mathfrak{f}}}$ is defined as
in (\ref{inner product}). The strategy is to look for the two Hamiltonian
structures for functionals defined on the co-adjoint orbits $L_{+}\in 
\widehat{\mathfrak{f}}^{\ast }$ and $L_{-}^{\prime }\in \widehat{\mathfrak{f}%
}^{\ast }.$ From these "light-cone" orbits we select the true phase space as
the spatial component of the Lax operator, i.e $L_{x}$ $=\frac{1}{2}$ $%
\left( L_{+}-L_{-}\right) $. Note that the relativistic counterpart of $%
L_{+} $ is taken as $L_{-}^{\prime }$ and not $L_{-},$ which is more natural
as can be seen from the equations (\ref{Lax operators}), (\ref{prime Lax
operators}). Now that $L_{+},L_{-}^{\prime }$ are considered as phase
spaces, denoted generically by $\Xi $, it is useful to recall the definition
of a differential $d_{\Xi }h$ of a functional on the orbit $\Xi $, i.e of $%
h\left( \Xi \right) :\Xi \rightarrow \mathcal{F}$ ($\mathcal{F}$ is a
field), which is a linear form in $\left( \widehat{\mathfrak{f}}^{\ast
}\right) ^{\ast }\sim \widehat{\mathfrak{f}}$. Then, with $\Xi \in \widehat{%
\mathfrak{f}}^{\ast }$ and a function $h,$ we find the differential $d_{\Xi
}h$ of $h$ through the taylor-like relation%
\begin{equation*}
h\left( \Xi +\delta \Xi \right) =h\left( \Xi \right) +\left( d_{\Xi
}h\right) \circ (\delta \Xi )+\mathcal{O}\left( \delta \Xi ^{2}\right) ,
\end{equation*}%
where $\delta \Xi \in \widehat{\mathfrak{f}}^{\ast }$ is an arbitrary
variation of $\Xi .$

Under the inner products (\ref{space inner product}) or (\ref{inner product}%
), $\widehat{\mathfrak{f}}$ and its dual $\widehat{\mathfrak{f}}^{\ast }$
are identified and we can write $\left( d_{\Xi }h\right) \circ (\delta \Xi
)=\left( d_{\Xi }h,\delta \Xi \right) .$ When the orbit $\Xi $ takes values
on some subspace $\widehat{\mathfrak{f}}_{\Xi }\subset $ $\widehat{\mathfrak{%
f}}$ of the Lie algebra $\widehat{\mathfrak{f}}\mathfrak{,}$ we write the
variation in the form $\delta \Xi =\varepsilon r_{\Xi },$ where $\varepsilon
<<1$, $r_{\Xi }\in $ $\widehat{\mathfrak{f}}_{\Xi }$ and this leads to the
following definition of \ the differential of a function\footnote{%
This definition is equivalent to the usual notion of functional
differentiation after taking the trace \cite{generalized DS II}.} $h$ on $%
\Xi $ 
\begin{equation}
\frac{d}{d\varepsilon }h\left( \Xi +\varepsilon r_{\Xi }\right) \mid
_{\varepsilon =0}\text{ }\equiv \text{ }\left( d_{\Xi }h,r_{\Xi }\right) ,
\label{differentials}
\end{equation}%
where we need to compute explicitly the LHS in order to isolate the
differential $d_{\Xi }h$ in the RHS$.$ From this we immediately conclude
that the differential\footnote{%
The use of the symbol $\perp $ should not be confuse with the same symbol
used before to denote projection along the kernel algebra $\mathcal{K}.$} $%
d_{\Xi }h\in \widehat{\mathfrak{f}}_{\Xi }^{\bot }$ belongs to the
ortho-complement of $\widehat{\mathfrak{f}}_{\Xi }$ in $\widehat{\mathfrak{f}%
}$ because of the operator $\mathcal{P}_{\perp }(\ast )=\left( r_{\Xi },\ast
\right) :\widehat{\mathfrak{f}}_{\Xi }\rightarrow \widehat{\mathfrak{f}}%
_{\Xi }^{\bot }$ is a projector along $\widehat{\mathfrak{f}}_{\Xi }^{\bot
}. $

Under an arbitrary conjugation $\widetilde{\Xi }=S\Xi S^{-1},$ we get $r_{%
\widetilde{\Xi }}=Sr_{\Xi }S^{-1}$ and from (\ref{differentials}) we can see
the effect of a conjugation on a differential, which is given by $d_{%
\widetilde{\Xi }}h=S\left( d_{\Xi }h\right) S^{-1}.$ This relation has to be
used in order to find the differentials on $L_{+}^{\prime },L_{-}$ starting
from those differentials computed on the orbits $L_{+},L_{-}^{\prime }$,
recall that we have $L_{\pm }^{\prime }=B^{-1}L_{\pm }B,$ c.f (\ref{lax
conjugation}). We denote the differentials of a functional $h$ defined on $%
L_{+}$ and $L_{-}^{\prime }$ as $d_{+}h$ and $d_{-}^{\prime }h$
respectively, and a similar notation for $L_{+}^{\prime }$ and $L_{-}$ after
the $B$-conjugation$.$

Now, we proceed to compute the differentials of the functions we are most
interested, namely the conserved charges found in section 2.5 by means of
the Drinfeld-Sokolov procedure.

Consider the stress-tensor components $T_{\mu \nu }$ extracted from the DS
procedure as functionals on $L_{\pm },$ $L_{\pm }^{\prime }.$ They are 
\begin{equation*}
T_{++}=\left\langle \Lambda _{+}^{(+1)},U_{-}^{-1}L_{+}U_{-}\right\rangle ,%
\text{ }T_{-+}=-\left\langle \Lambda
_{+}^{(+1)},U_{-}^{-1}L_{-}U_{-}\right\rangle ,\text{ }T_{+-}=\left\langle
\Lambda _{-}^{(-1)},U_{+}^{-1}L_{+}^{\prime }U_{+}\right\rangle ,\text{ }%
T_{--}=-\left\langle \Lambda _{-}^{(-1)},U_{+}^{-1}L_{-}^{\prime
}U_{+}\right\rangle ,
\end{equation*}%
where we have used (\ref{U}) and (\ref{U gauge}) in the $A=0$ gauge. Taking
into account the variations\footnote{%
The $r_{\geq 0}$, $r_{\leq 0}$ stands for arbitrary terms taking values on
the domains of the Lax connections $\Lambda _{\Theta }^{(+1)},$ $\Lambda
_{\Pi ^{\prime }}^{(-1)}$.} $\delta L_{+}=\varepsilon r_{\geq 0}$, $\delta
L_{-}^{\prime }=\varepsilon r_{\leq 0},$ the fact that $U_{\pm }(\varepsilon
)$ depend on $\varepsilon $ through the fields$,$ the ad-invariance of the
inner product and the relation $\left[ \Lambda _{\pm }^{(\pm 1)},L_{\pm }^{V}%
\right] =0,$ where $L_{\pm }^{V}$ are the vacuum Lax operators, we get from (%
\ref{differentials}), the associated differentials%
\begin{eqnarray*}
d_{+}T_{++} &=&\left( \Theta \Lambda _{+}^{(+1)}\Theta ^{-1}\right) _{\leq 0}%
\text{, \ \ \ \ \ }d_{-}^{\prime }T_{-+}\text{ }=\text{ }-\left( \Theta
^{\prime }\Lambda _{+}^{(+1)}\Theta ^{\prime -1}\right) _{\geq 0}, \\
d_{+}T_{+-} &=&\left( \Pi \Lambda _{-}^{(-1)}\Pi ^{-1}\right) _{\leq 0}\text{%
, \ \ \ \ \ }d_{-}^{\prime }T_{--}\text{ }=\text{ }-\left( \Pi ^{\prime
}\Lambda _{-}^{(-1)}\Pi ^{\prime -1}\right) _{\geq 0},
\end{eqnarray*}%
where we have used the factorizations $\Theta =U_{-}S_{-},$ $\Pi
=BU_{+}S_{+},$ $\Theta ^{\prime }=B^{-1}\Theta ,$ $\Pi ^{\prime }=U_{+}S_{+}$
and the fact that the kernel components $S_{\pm }$ does not contribute.

In a very similar way but this time using the fact that $\widehat{\mathfrak{f%
}}^{\perp }\cap \widehat{\mathfrak{f}}^{\parallel }=\oslash $, we have for
the fermionic functions%
\begin{eqnarray}
m_{+} &=&\left( D^{(+1/2)},K_{+}^{(-1/2)}\right) ,\text{ \ \ \ \ \ }m_{-}%
\text{ }=\text{ }-\left( D^{(+1/2)},K_{-}^{(-1/2)}\right) ,
\label{supercharge pieces} \\
n_{+} &=&-\left( D^{(-1/2)},K_{+}^{(+1/2)}\right) ,\text{ \ \ \ }n_{-}\text{ 
}=\text{ }\left( D^{(-1/2)},K_{-}^{(+1/2)}\right) ,  \notag
\end{eqnarray}%
where $D^{(\pm 1/2)}=\epsilon _{i}F_{i}^{(\pm 1/2)},$\ $i=1,...,\dim 
\mathcal{K}_{F}^{(\pm 1/2)},$ the associated differentials%
\begin{eqnarray}
d_{+}m_{+} &=&\left( U_{-}D^{(+1/2)}U_{-}^{-1}\right) _{\leq 0},\text{ \ \ \
\ \ \ \ \ \ \ \ \ \ \ }d_{-}^{\prime }m_{-}\text{ }=\text{ }-B^{-1}\left(
U_{-}D^{(+1/2)}U_{-}^{-1}\right) _{\geq 0}B,
\label{supercharge differentials} \\
d_{+}n_{+} &=&-B\left( U_{+}D^{(-1/2)}U_{+}^{-1}\right) _{\leq 0}B^{-1}\text{%
, \ \ \ \ }d_{-}^{\prime }n_{-}\text{ }=\text{ }\left(
U_{+}D^{(-1/2)}U_{+}^{-1}\right) _{\geq 0}.  \notag
\end{eqnarray}

To identify the Poisson structures, we follow the approach adopted in \cite%
{generalized DS II} to study the generalized Drinfeld-Sokolov hierarchies
and apply those results to our extended homogeneous hierarchy. For ease of
simplicity, we perform the calculation for a bosonic hierarchy with an
integer homogeneous gradation and at the end we comment on the necessary
changes required in the supersymmetric case.

Consider $L_{+}$ as the phase space and write the equations of motion $\left[
L_{+},L_{-}\right] =0$ with $L_{-}=\partial _{-}-\left( \Pi \Lambda
_{-}^{(-1)}\Pi ^{-1}\right) _{<0},$ in the following two equivalent Lax forms%
\begin{equation}
\partial _{-}L_{+}=\left[ \left( \Pi \Lambda _{-}^{(-1)}\Pi ^{-1}\right)
_{<0},L_{+}\right] =-\left[ \left( \Pi \Lambda _{-}^{(-1)}\Pi ^{-1}\right)
_{\geq 0},L_{+}\right] .  \label{L+ as orbit}
\end{equation}%
The first and second terms on the RHS of (\ref{L+ as orbit}) will lead,
respectively, to the first and second Poisson structures defined on $L_{+}$
as we now see.

Considering the first form of $\partial _{-}L_{+}$ and using $z^{-1}\left(
\Pi z\Lambda _{-}^{(-1)}\Pi ^{-1}\right) _{\leq 0}=\left( \Pi \Lambda
_{-}^{(-1)}\Pi ^{-1}\right) _{<0},$ because $z$ have homogeneous grade $+1$,
we find that%
\begin{equation*}
z^{-1}d_{+}T_{+-}(z)=\left( \Pi \Lambda _{-}^{(-1)}\Pi ^{-1}\right) _{<0},
\end{equation*}%
where $d_{+}T_{+-}(z)$ means that $\Lambda _{-}^{(-1)}$ is replaced by $%
z\Lambda _{-}^{(-1)}$ in the definition of $d_{+}T_{+-}.$ Then, we have the
differential representation of the first form%
\begin{equation}
\partial _{-}L_{+}=\left[ z^{-1}d_{+}T_{+-}(z),L_{+}\right] _{\geq 0},
\label{first form}
\end{equation}%
where we have, for consistency, projected along the same grade space
decomposition of the Lax connection in $L_{+}.$

Now, for the second form of $\partial _{-}L_{+}$ we use $z\left( \Pi
z^{-1}\Lambda _{-}^{(-1)}\Pi ^{-1}\right) _{\geq 0}=\left( \Pi \Lambda
_{-}^{(-1)}\Pi ^{-1}\right) _{>0},$ because $z^{-1}$ have homogeneous grade $%
-1,$ giving%
\begin{equation*}
\partial _{-}L_{+}=-\left[ \left( d_{+}T_{+-}\right) _{0},L_{+}\right] -z%
\left[ \left( \Pi z^{-1}\Lambda _{-}^{(-1)}\Pi ^{-1}\right) _{\geq 0},L_{+}%
\right] .
\end{equation*}%
Noting that $\partial _{z^{-1}(-)}L_{+}=-\left[ \left( \Pi z^{-1}\Lambda
_{-}^{(-1)}\Pi ^{-1}\right) _{\geq 0},L_{+}\right] ,$ where $\partial
_{z^{-1}(-)}L_{+}$ means that $\Lambda _{-}^{(-1)}$ is replaced by $%
z^{-1}\Lambda _{-}^{(-1)}$, we use the first form representation (\ref{first
form}) to write $\partial _{z^{-1}(-)}L_{+}=\left[ z^{-1}d_{+}T_{+-},L_{+}%
\right] _{\geq 0}$ and obtain 
\begin{equation*}
\partial _{-}L_{+}=-\left[ \left( d_{+}T_{+-}\right) _{0},L_{+}\right] +z%
\left[ z^{-1}d_{+}T_{+-},L_{+}\right] _{\geq 0}.
\end{equation*}%
Finally, using $z\left[ z^{-1}d_{+}T_{+-},L_{+}\right] _{\geq 0}=\left[
d_{+}T_{+-},L_{+}\right] _{>0},$ we obtain the differential representation
of the second form%
\begin{equation}
\partial _{-}L_{+}=-\left[ \left( d_{+}T_{+-}\right) _{0},L_{+}\right] +%
\left[ d_{+}T_{+-},L_{+}\right] _{>0}.  \label{second form}
\end{equation}

From these results we can write two equivalent forms for the $\partial _{-}$
evolution of a functional $\varphi $ on $L_{+},$ i.e 
\begin{eqnarray}
\partial _{-}\varphi &=&\left( d_{+}\varphi ,\partial _{-}L_{+}\right)
=\left( d_{+}\varphi ,\left[ z^{-1}d_{+}T_{+-}(z),L_{+}\right] \right) ,
\label{equivalent forms} \\
\partial _{-}\varphi &=&\left( d_{+}\varphi ,\partial _{-}L_{+}\right)
=\left( d_{+}\varphi ,-\left[ \left( d_{+}T_{+-}\right) _{0},L_{+}\right] +%
\left[ d_{+}T_{+-},L_{+}\right] _{>0}\right) ,  \notag
\end{eqnarray}%
leading to the following two Kostant-Kirillov brackets on the orbit $L_{+}$ 
\begin{eqnarray}
\left\{ \varphi ,\psi \right\} _{1}(L_{+}) &=&-\left( L_{+},z^{-1}\left[
d_{+}\varphi ,d_{+}\psi \right] \right) ,  \label{brackets L mas} \\
\left\{ \varphi ,\psi \right\} _{2}(L_{+}) &=&\left( L_{+},\left[
(d_{+}\varphi )_{0},(d_{+}\psi )_{0}\right] -\left[ (d_{+}\varphi
)_{<0},(d_{+}\psi )_{<0}\right] \right) ,  \notag
\end{eqnarray}%
where we have used the decomposition $d_{+}f=(d_{+}f)_{0}+(d_{+}f)_{<0}$ in
order to simplify the second bracket. From (\ref{brackets L mas}) we can
write (\ref{equivalent forms}) as a recursion relation%
\begin{equation}
\partial _{-}\varphi =\left\{ T_{+-}(z),\varphi \right\} _{1}(L_{+})=\left\{
T_{+-},\varphi \right\} _{2}(L_{+}).  \label{recursion 1}
\end{equation}

In a complete analogous way, we consider $L_{-}^{\prime }$ as the phase
space and write $\left[ L_{+}^{\prime },L_{-}^{\prime }\right] =0$ in the
two equivalent Lax forms. We get%
\begin{eqnarray}
\partial _{+}L_{-}^{\prime } &=&\left[ zd_{-}^{\prime
}T_{-+}(z^{-1}),L_{-}^{\prime }\right] _{\leq 0},\text{ }
\label{two rep for L minus prime} \\
\partial _{+}L_{-}^{\prime } &=&-\left[ \left( d_{-}^{\prime }T_{-+}\right)
_{0},L_{-}^{\prime }\right] +\left[ d_{-}^{\prime }T_{-+},L_{-}^{\prime }%
\right] _{<0},  \notag
\end{eqnarray}%
leading to the following two Kostant-Kirillov brackets on the orbit $%
L_{-}^{\prime }$ 
\begin{eqnarray}
\left\{ \varphi ,\psi \right\} _{1}(L_{-}^{\prime }) &=&\left( L_{-}^{\prime
},z\left[ d_{-}^{\prime }\varphi ,d_{-}^{\prime }\psi \right] \right) ,
\label{brackets L menos} \\
\left\{ \varphi ,\psi \right\} _{2}(L_{-}^{\prime }) &=&-\left(
L_{-}^{\prime },\left[ (d_{-}^{\prime }\varphi )_{0},(d_{-}^{\prime }\psi
)_{0}\right] -\left[ (d_{-}^{\prime }\varphi )_{>0},(d_{-}^{\prime }\psi
)_{>0}\right] \right) ,  \notag
\end{eqnarray}%
where we have, for convenience, multiplied by a global factor $-1.$ The
recursion relation becomes now 
\begin{equation}
\partial _{+}\varphi =-\left\{ T_{-+}(z^{-1}),\varphi \right\}
_{1}(L_{-}^{\prime })=-\left\{ T_{-+},\varphi \right\} _{2}(L_{-}^{\prime }).
\label{recursion 2}
\end{equation}

It is instructive to repeat the same analysis for the identities $\left[
L_{+},L_{+}\right] =0$ and $\left[ L_{-}^{\prime },L_{-}^{\prime }\right]
=0. $ We find that $\partial _{+}L_{+}$ and $\partial _{-}L_{-}^{\prime }$
can be written in two equivalent ways 
\begin{eqnarray*}
\partial _{+}L_{+} &=&\left[ z^{-1}d_{+}T_{++}(z),L_{+}\right] _{\geq 0}=-%
\left[ \left( d_{+}T_{++}\right) _{0},L_{+}\right] +\left[ d_{+}T_{++},L_{+}%
\right] _{>0}, \\
\partial _{-}L_{-}^{\prime } &=&\left[ zd_{-}^{\prime
}T_{--}(z^{-1}),L_{-}^{\prime }\right] =-\left[ \left( d_{-}^{\prime
}T_{--}\right) _{0},L_{-}^{\prime }\right] +\left[ d_{-}^{\prime
}T_{--},L_{-}^{\prime }\right] _{<0},
\end{eqnarray*}%
showing that the mixed components of the stress-tensor, i.e $T_{\pm \mp },$
are responsible for coupling the two sectors of the relativistic part of the
extended integrable hierarchy as shown by (\ref{first form}),(\ref{second
form}).

The brackets (\ref{brackets L mas}),(\ref{brackets L menos}) can be written
in the following compact $r$-bracket forms%
\begin{eqnarray}
\left\{ \varphi ,\psi \right\} _{\mu }(L_{+}) &=&\left( L_{+},\left[
d_{+}\varphi ,d_{+}\psi \right] _{\mathcal{R}_{<,\mu }}\right) ,
\label{compact R bracket form} \\
\left\{ \varphi ,\psi \right\} _{\nu }(L_{-}^{\prime }) &=&\left(
L_{-}^{\prime },\left[ d_{-}^{\prime }\varphi ,d_{-}^{\prime }\psi \right] _{%
\mathcal{R}_{>,\nu }}\right) ,  \notag
\end{eqnarray}%
where ($\mu ,\nu \in 
\mathbb{C}
$) \ 
\begin{eqnarray}
\mathcal{R}_{<,\mu } &=&\left( \mathcal{R}_{<}-\mu \cdot z^{-1}\right) ,%
\text{ \ \ \ \ \ }\mathcal{R}_{<}\text{ }=\text{ }\frac{1}{2}\left( \mathcal{%
P}_{0}-\mathcal{P}_{<0}\right) ,  \label{R matrices} \\
\mathcal{R}_{>,\nu } &=&-\left( \mathcal{R}_{>}-\nu \cdot z\right) ,\text{ \
\ \ \ \ \ }\mathcal{R}_{>}\text{ }=\text{ }\frac{1}{2}\left( \mathcal{P}_{0}-%
\mathcal{P}_{>0}\right) ,  \notag
\end{eqnarray}%
are the corresponding $r$-matrices entering the $r$-bracket\footnote{%
Explicitly: $\left[ X,Y\right] _{\mathcal{R}_{\lessgtr }}=\pm \left( \left[
X_{0},Y_{0}\right] -\left[ X_{_{\lessgtr }0},Y_{_{\lessgtr }0}\right]
\right) .$} $\left[ X,Y\right] _{\mathcal{R}}=\left[ \mathcal{R}(X),Y\right]
+\left[ X,\mathcal{R}(Y)\right] $ and $\mathcal{P}_{0},\mathcal{P}_{<0},%
\mathcal{P}_{>0}$ are the projectors along zero, negative and positive
grades.

The Jacobi identity and the compatibility of (\ref{compact R bracket form})
are guaranteed because each $r$-matrix $\mathcal{R}_{<,\mu },\mathcal{R}%
_{>,\nu }$ satisfy separately the classical modified Yang-Baxter equation 
\begin{equation}
\left[ \mathcal{R}_{i}\left( X\right) ,\mathcal{R}_{i}\left( Y\right) \right]
-\mathcal{R}_{i}\left( \left[ \mathcal{R}_{i}\left( X\right) ,Y\right] +%
\left[ X,\mathcal{R}_{i}\left( Y\right) \right] \right) =\lambda _{i}\left[
X,Y\right] ,\text{ \ \ }\lambda _{i}\in 
\mathbb{C}
,i=1,2  \label{class. mod. super YB}
\end{equation}%
with $\mathcal{R}_{1}=\mathcal{R}_{<,\mu }$, $\mathcal{R}_{2}=\mathcal{R}%
_{>,\nu }$ and%
\begin{equation*}
\lambda _{1}=-\left( \frac{1}{2}+\mu \cdot z^{-1}\right) ^{2},\text{ \ \ \ \
\ }\lambda _{2}=-\left( \frac{1}{2}+\nu \cdot z\right) ^{2}.
\end{equation*}

In what follows we will choose the second symplectic structure because, for
the homogeneous hierarchies, it is nothing but the canonical symplectic
structure associated to the Lagrangian (\ref{Abelian toda action}), see (\ref%
{SSSSG model bracket}) below. In the superalgebra case, we only have to
replace trace by supertrace and the integer gradation by half-integer one.
For the moment, it is of great importance to mention that the form of the
first Hamiltonian structure above is a consequence of the special properties
of the only gradation (naively) used, i.e, the integer homogenous gradation.
Note that we are not taking into account any other gradations. This does not
necessarily mean that the integrable system described by the action (\ref%
{Abelian toda action}) admits such an structure in the form presented above,
i.e, that the supersymmetrization of the first structure takes the same form
as in (\ref{brackets L mas}), (\ref{brackets L menos}), because the
existence of the finer $%
\mathbb{Z}
_{4}$ decomposition entering in (\ref{loop algebra}) could preclude it. In
that case, the recursion relations (\ref{recursion 1}), (\ref{recursion 2})
have to be modified in an appropriate way and this changes the explicit form
of the first Poisson structure, in fact it becomes non-local \cite{us}.
However, the second structure is enough for our present purposes\footnote{%
The first bracket is requiered only when we want to link the sigma/SSSSG
model symplectic structures.}.

The $r$-matrices (\ref{R matrices}) satisfy (\ref{class. mod. super YB})
separately and based on this fact we now propose a bracket on $L_{x}$
through the following definition.

\begin{definition}
The bracket on the spatial orbit $L=L_{x}$ is given by\footnote{%
We have absorbed the $1/2$ in $L=\frac{1}{2}\left( L_{+}-L_{-}\right) .$} 
\begin{equation}
\left\{ \varphi ,\psi \right\} _{2}(L)=\left\{ \varphi ,\psi \right\}
_{2}(L_{+})-\left\{ \varphi ,\psi \right\} _{2}(L_{-})=\left( L_{+},\left[
d_{+}\varphi ,d_{+}\psi \right] _{\mathcal{R}_{<}}\right) +\left( L_{-},%
\left[ d_{-}\varphi ,d_{-}\psi \right] _{\mathcal{R}_{>}}\right) ,
\label{spatial 2 bracket}
\end{equation}%
where we find $\left\{ \varphi ,\psi \right\} _{2}(L_{-})\ $and the
differentials $d_{-}\varphi ,d_{-}\psi $ by using the map $L_{-}^{\prime
}=B^{-1}L_{-}B.$ The differentials $d\varphi ,d\psi $ on $L$ are constructed
by restricting the functionals $\varphi ,\psi $ on the respective domains of
definitions of $L_{+}$ and $L_{-}^{\prime }$. The bracket for the $B$%
-equivalent representation $L^{\prime }$ is defined accordingly.
\end{definition}

For computational purposes, we use $L_{+}=\partial _{+}+\Lambda _{\Theta
}^{(+1)}$ and $L_{-}^{\prime }=\partial _{-}-\Lambda _{\Pi ^{\prime
}}^{(-1)} $ in order to rewrite the brackets expressions. We get\footnote{%
Concerning the first expression right below, a similar bracket was introduce
in \cite{Delduc-Gallot} by using a superspace approach to the
Drinfeld-Sokolov reduction.}%
\begin{eqnarray}
\left\{ \varphi ,\psi \right\} _{2}(L_{+}) &=&\left( \Lambda _{\Theta
}^{(+1)},\left[ d_{+}\varphi ,d_{+}\psi \right] _{\mathcal{R}_{<}}\right)
+\left( \partial _{+}\left( d_{+}\varphi \right) ,\mathcal{R}_{<}(d_{+}\psi
)\right) -\left( \mathcal{R}_{<}\left( d_{+}\varphi \right) ,\partial
_{+}(d_{+}\psi )\right) ,  \label{simplified brackets} \\
\left\{ \varphi ,\psi \right\} _{2}(L_{-}^{\prime }) &=&\left( \Lambda _{\Pi
^{\prime }}^{(-1)},\left[ d_{-}^{\prime }\varphi ,d_{-}^{\prime }\psi \right]
_{\mathcal{R}_{>}}\right) -\left( \partial _{-}\left( d_{-}^{\prime }\varphi
\right) ,\mathcal{R}_{>}(d_{-}^{\prime }\psi )\right) +\left( \mathcal{R}%
_{>}\left( d_{-}^{\prime }\varphi \right) ,\partial _{-}(d_{-}^{\prime }\psi
)\right) ,  \notag
\end{eqnarray}%
where we have used the ad-invariance of the inner product to write%
\begin{equation}
\left( \partial ,\left[ X,Y\right] _{\mathcal{R}}\right) =\left( \partial X,%
\mathcal{R}(Y)\right) -\left( \mathcal{R}(X),\partial Y\right) .
\label{derivative term}
\end{equation}%
Note that these derivative terms only receive contributions from the zero
grade parts of the algebra.

\subsection{The second Hamiltonian structure: the supercharges algebra and
the field variations.}

The purpose of this section is threefold. We will find part of the algebra
obeyed by the supercharges (\ref{mKdV supercharges}) and comment later on
some issues related to the role of gauge group. We shall obtain, by using
the second Poisson bracket (\ref{spatial 2 bracket}), the supersymmetric
field variations induced by the supercharges and also show that the second
symplectic structure found above coincides with the symplectic structure of
the WZNW model (\ref{Abelian toda action}).

Let us start with the following proposition.

\begin{proposition}
Under the second Poisson structure, the functionals (\ref{supercharge pieces}%
) satisfy the relations%
\begin{eqnarray*}
\left\{ m_{+},\widetilde{m}_{+}\right\} _{2}(L_{+}) &=&2\text{ }\epsilon
\cdot \widetilde{\epsilon }\text{ }\dint_{-\infty }^{+\infty }dx^{1}T_{++},%
\text{ \ \ \ \ }\left\{ m_{-},\widetilde{m}_{-}\right\} _{2}(L_{-})\text{ }=%
\text{ }-2\text{ }\epsilon \cdot \widetilde{\epsilon }\text{ }\dint_{-\infty
}^{+\infty }dx^{1}T_{-+} \\
\left\{ n_{+},\widetilde{n}_{+}\right\} _{2}(L_{+}^{\prime }) &=&2\text{ }%
\overline{\epsilon }\cdot \widetilde{\overline{\epsilon }}\text{ }%
\dint_{-\infty }^{+\infty }dx^{1}T_{+-},\text{ \ \ \ \ \ \ }\left\{ n_{-},%
\widetilde{n}_{-}\right\} _{2}(L_{-}^{\prime })\text{ }=\text{ }-2\text{ }%
\overline{\epsilon }\cdot \widetilde{\overline{\epsilon }}\text{ }%
\dint_{-\infty }^{+\infty }dx^{1}T_{--}.
\end{eqnarray*}%
Moreover, under (\ref{spatial 2 bracket}) the supercharges (\ref{mKdV
supercharges}) written as $Q^{+}=m_{+}+m_{-},$ $Q^{-}=n_{+}+n_{-}$ satisfy%
\begin{equation}
\left\{ Q^{+},\widetilde{Q}^{+}\right\} _{2}(L)=2\text{ }\epsilon \cdot 
\widetilde{\epsilon }\text{ }\left( H+P\right) ,\text{ \ \ \ \ \ }\left\{
Q^{-},\widetilde{Q}^{-}\right\} _{2}(L^{\prime })=2\text{ }\overline{%
\epsilon }\cdot \widetilde{\overline{\epsilon }}\text{ }\left( H-P\right) .
\label{charge algebra}
\end{equation}
\end{proposition}

\begin{proof}
We will compute one of the brackets only, as the others computations are
quite similar. Starting with, cf. (\ref{simplified brackets}), 
\begin{equation*}
\left\{ m_{+},\widetilde{m}_{+}\right\} _{2}(L_{+})=\left( \Lambda _{\Theta
}^{(+1)},\left[ d_{+}m_{+},d_{+}\widetilde{m}_{+}\right] _{\mathcal{R}%
_{<}}\right) +\left( \partial _{+}\left( d_{+}m_{+}\right) ,\mathcal{R}%
_{<}(d_{+}\widetilde{m}_{+})\right) -\left( \mathcal{R}_{<}\left(
d_{+}m_{+}\right) ,\partial _{+}(d_{+}\widetilde{m}_{+})\right) ,
\end{equation*}%
using the explicit form of the Lax connection 
\begin{equation*}
\Lambda _{\Theta }^{(+1)}=A_{+}^{(0)}+Q_{+}^{(0)}+\psi _{+}^{\left(
+1/2\right) }+\Lambda _{+}^{(+1)},
\end{equation*}%
the involved $r$-matrix $\mathcal{R}_{<}=\frac{1}{2}\left( \mathcal{P}_{0}-%
\mathcal{P}_{<0}\right) $ and the projected components of the differentials 
\begin{eqnarray*}
\left( d_{+}m_{+}\right) _{0} &=&\left[ \psi ^{(-1/2)},D^{(+1/2)}\right] , \\
\left( d_{+}m_{+}\right) _{-1/2} &=&\left[ \psi ^{(-1)},D^{(+1/2)}\right] +%
\frac{1}{2}\ \left[ \psi ^{(-1/2)},\left[ \psi ^{(-1/2)},D^{(+1/2)}\right] %
\right] ,
\end{eqnarray*}%
we identify the relevant contributions we need to compute:%
\begin{eqnarray*}
A &=&\left( \Lambda _{\Theta }^{(+1)},\left[ \left( d_{+}m_{+}\right)
_{0},\left( d_{+}\widetilde{m}_{+}\right) _{0}\right] \right) =\left(
Q_{+}^{(0)},\left[ \left[ \psi ^{(-1/2)},D^{(+1/2)}\right] ,\left[ \psi
^{(-1/2)},\widetilde{D}^{(+1/2)}\right] \right] \right) , \\
B &=&\left( \Lambda _{\Theta }^{(+1)},\left[ (d_{+}m_{+})_{-1/2},(d_{+}%
\widetilde{m}_{+})_{-1/2}\right] \right) =\left( \Lambda _{+}^{(+1)},\left[ %
\left[ \psi ^{(-1)},D^{(+1/2)}\right] ,\left[ \psi ^{(-1)},\widetilde{D}%
^{(+1/2)}\right] \right] \right) , \\
C &=&\text{derivative term}=\frac{1}{2}\left( \partial _{+}\psi ^{(-1/2)},%
\left[ D^{(+1/2)},\left[ \psi ^{(-1/2)},\widetilde{D}^{(+1/2)}\right] \right]
-\left[ \widetilde{D}^{(+1/2)},\left[ \psi ^{(-1/2)},D^{(+1/2)}\right] %
\right] \right) .
\end{eqnarray*}

Considering $A$, we anti-symmetrize and use the ad-invariance of the inner
product to get%
\begin{equation*}
A=\frac{1}{2}\left( \left[ Q_{+}^{(0)},\left[ \psi ^{(-1/2)},D^{(+1/2)}%
\right] \right] ,\left[ \psi ^{(-1/2)},\widetilde{D}^{(+1/2)}\right] \right)
-\frac{1}{2}\left( \left[ Q_{+}^{(0)},\left[ \psi ^{(-1/2)},\widetilde{D}%
^{(+1/2)}\right] \right] ,\left[ \psi ^{(-1/2)},D^{(+1/2)}\right] \right) .
\end{equation*}%
Now, use the Jacobi identity to write 
\begin{equation*}
\left[ Q_{+}^{(0)},\left[ \psi ^{(-1/2)},X\right] \right] =-\left[ X,\left[
Q_{+}^{(0)},\psi ^{(-1/2)}\right] \right] -\left[ \psi ^{(-1/2)},\left[
X,Q_{+}^{(0)}\right] \right] ,
\end{equation*}%
where $X=D^{(+1/2)}$ and $\widetilde{D}^{(+1/2)}.$ Using the ad-invariance
of the inner product again we obtain%
\begin{eqnarray*}
A &=&\left( \left[ Q_{+}^{(0)},\psi ^{(-1/2)}\right] ,\left[ D^{(+1/2)},%
\left[ \psi ^{(-1/2)},\widetilde{D}^{(+1/2)}\right] \right] -\left[ 
\widetilde{D}^{(+1/2)},\left[ \psi ^{(-1/2)},D^{(+1/2)}\right] \right]
\right) - \\
&&-\frac{1}{2}\left( \left[ Q_{+}^{(0)},D^{(+1/2)}\right] ,\left[ \psi
^{(-1/2)},\left[ \psi ^{(-1/2)},\widetilde{D}^{(+1/2)}\right] \right]
\right) +\frac{1}{2}\left( \left[ Q_{+}^{(0)},\widetilde{D}^{(+1/2)}\right] ,%
\left[ \psi ^{(-1/2)},\left[ \psi ^{(-1/2)},D^{(+1/2)}\right] \right]
\right) .
\end{eqnarray*}%
The second line right above vanishes because of the anti-commutativity of
the constant Grassmannian parameters $\epsilon ,\widetilde{\epsilon }$ and
the supersymmetry of the trace and we get%
\begin{equation*}
A=\left( \left[ Q_{+}^{(0)},\psi ^{(-1/2)}\right] ,\left[ D^{(+1/2)},\left[
\psi ^{(-1/2)},\widetilde{D}^{(+1/2)}\right] \right] -\left[ \widetilde{D}%
^{(+1/2)},\left[ \psi ^{(-1/2)},D^{(+1/2)}\right] \right] \right) .
\end{equation*}%
Using once again the Jacobi identity we obtain%
\begin{equation}
\left[ D^{(+1/2)},\left[ \psi ^{(-1/2)},\widetilde{D}^{(+1/2)}\right] \right]
-\left[ \widetilde{D}^{(+1/2)},\left[ \psi ^{(-1/2)},D^{(+1/2)}\right] %
\right] =\left[ \psi ^{(-1/2)},\left[ D^{(+1/2)},\widetilde{D}^{(+1/2)}%
\right] \right] =2\text{ }\epsilon \cdot \widetilde{\epsilon }\left[ \psi
^{(-1/2)},\Lambda _{+}^{(+1)}\right] ,  \label{result}
\end{equation}%
where $\epsilon \cdot \widetilde{\epsilon }=\epsilon _{i}\widetilde{\epsilon 
}_{i},$ $i=1,...,\dim \mathcal{K}_{F}^{(+1/2)}$ and where we have assumed
that\footnote{%
In particular, this is satisfied by the superalgebras entering the Pohlmeyer
reduction of $AdS_{n}\times S^{n},$ $n=2,3,5$ and $AdS_{4}\times 
\mathbb{C}
P^{3}$ superstring sigma models. For explicit examples, see below (\ref{psu
anticom}) and (\ref{psuxpsu anticom}).} 
\begin{equation*}
\left\{ F_{i}^{(+1/2)},F_{j}^{(+1/2)}\right\} =2\delta _{ij}\Lambda
_{+}^{(+1)},\text{ \ \ \ }i=1,...,\dim \mathcal{K}_{F}^{(+1/2)}.
\end{equation*}%
Now, we use the ad-invariance of the inner product and the relations (\ref%
{Baker-Hausdorf fields}) to obtain%
\begin{equation*}
A=2\text{ }\epsilon \cdot \widetilde{\epsilon }\text{ }\left( \left(
Q_{+}^{(0)}\right) ^{2}\right) \text{ , \ }C=2\text{ }\epsilon \cdot 
\widetilde{\epsilon }\text{ }\frac{1}{2}\left( \psi _{+}^{\left( +1/2\right)
}\partial _{+}\psi ^{(-1/2)}\right) ,
\end{equation*}%
where we have used (\ref{result}) to find $C.$

It remains to compute $B$. Anti-symmetrizing and using ad-invariance of the
trace we get%
\begin{equation*}
B=\frac{1}{2}\left( \left[ \Lambda _{+}^{(+1)},\left[ \psi ^{(-1)},D^{(+1/2)}%
\right] \right] ,\left[ \psi ^{(-1)},\widetilde{D}^{(+1/2)}\right] \right) -%
\frac{1}{2}\left( \left[ \Lambda _{+}^{(+1)},\left[ \psi ^{(-1)},\widetilde{D%
}^{(+1/2)}\right] \right] ,\left[ \psi ^{(-1)},D^{(+1/2)}\right] \right) .
\end{equation*}%
Using the Jacobi identity, the fact that $\left[ \Lambda _{+}^{(+1)},X\right]
=0,$ with $X$ as above and the relations (\ref{Baker-Hausdorf fields}) we
write%
\begin{equation*}
\left[ \Lambda _{+}^{(+1)},\left[ \psi ^{(-1)},X\right] \right] =\left[
X,A_{-}^{(0)}\right] .
\end{equation*}%
This last result, the equation (\ref{result}) with $\psi
^{(-1/2)}\rightarrow \psi ^{(-1)}$ and once again (\ref{Baker-Hausdorf
fields}) gives 
\begin{equation*}
B=-2\text{ }\epsilon \cdot \widetilde{\epsilon }\text{ }\frac{1}{2}\left(
\left( A_{+}^{(0)}\right) ^{2}\right) .
\end{equation*}%
Finally, putting all together, i.e $A-B+C$ we arrive to the final expression%
\begin{equation*}
\left\{ m_{+},\widetilde{m}_{+}\right\} _{2}(L_{+})=2\text{ }\epsilon \cdot 
\widetilde{\epsilon }\text{ }\dint_{-\infty }^{+\infty }dx^{1}T_{++},
\end{equation*}%
where we have used (\ref{kinetic A's}) in the gauge $A=0$. A similar
computation follows for the other brackets, the only difference is that for
the brackets involving the mixed components $T_{-+},T_{+-}$ we have to use
the fermion equations of motion in (\ref{gauged fix eq of motion}) and also
the relation $\left\{ F_{i}^{(-1/2)},F_{j}^{(-1/2)}\right\} =-2\delta
_{ij}\Lambda _{-}^{(-1)}$. The supercharge algebra follows directly after
noting that the differentials defined on $L_{+}$ and $L_{-}^{\prime }$
vanishes when restricted to $L_{-}^{\prime }$ and $L_{+}$ respectively,
getting $\left\{ Q^{+},\widetilde{Q}^{+}\right\} _{2}(L)=\left\{ m_{+},%
\widetilde{m}_{+}\right\} _{2}(L_{+})-\left\{ m_{-},\widetilde{m}%
_{-}\right\} _{2}(L_{-})$ and a similar relation for $\left\{ Q^{-},%
\widetilde{Q}^{-}\right\} _{2}(L^{\prime }).$ Taking into account (\ref{H
and P}), the result follows.
\end{proof}

Now, in order to find the supersymmetry field variations we need to find the
differentials associated to the physical fields, which are the simplest to
find. Let us first introduce the following quantities%
\begin{eqnarray*}
G_{\mu }^{(-1/2)} &=&\mu _{i}G_{i}^{(-1/2)},\text{ \ }G_{\nu }^{(+1/2)}\text{
}=\text{ }\nu _{i}G_{i}^{(+1/2)},\text{ \ }M_{x}^{(0)}\text{ }=\text{ }%
x_{a}M_{a}^{(0)}, \\
\left\langle G_{i}^{(-1/2)},G_{j}^{(+1/2)}\right\rangle &=&-\left\langle
G_{j}^{(-1/2)},G_{i}^{(+1/2)}\right\rangle \text{ }=\text{ }\delta _{ij},%
\text{ \ }\left\langle M_{a}^{(0)},M_{b}^{(0)}\right\rangle \text{ }=\text{ }%
\left\langle M_{b}^{(0)},M_{a}^{(0)}\right\rangle \text{ }=\text{ }\delta
_{ab},
\end{eqnarray*}%
where $G_{i}^{(-1/2)}\in \mathcal{M}_{F}^{(-1/2)},G_{i}^{(+1/2)}\in \mathcal{%
M}_{F}^{(+1/2)},$ $M_{a}^{(0)}\in \mathcal{M}_{B}^{(0)}$ span the image part 
$\mathcal{M}$ of the algebra $\widehat{\mathfrak{f}}$ and $x$ and $\mu ,\nu $
are even/odd constant parameters, respectively. Setting 
\begin{equation*}
\psi _{+}^{\left( +1/2\right) }=\psi _{i}G_{i}^{(+1/2)},\text{ \ }\psi
_{-}^{\left( -1/2\right) }=\overline{\psi }_{i}G_{i}^{(-1/2)},\text{ \ }%
A_{+}^{(0)}=A_{+a}^{(0)}\text{ }M_{a}^{(0)},\text{ \ }%
A_{-}^{(0)}=A_{-a}^{(0)}\text{ }M_{a}^{(0)}
\end{equation*}%
and using $L_{+},L_{-}^{\prime }$ in their original forms 
\begin{equation*}
L_{+}=\partial _{+}+\left( A_{+}^{(0)}+Q_{+}^{(0)}+\psi _{+}^{\left(
+1/2\right) }+\Lambda _{+}^{(+1)}\right) ,\text{ \ \ }L_{-}^{\prime
}=\partial _{-}-\left( A_{-}^{(0)}+Q_{-}^{(0)}+\psi _{-}^{\left( -1/2\right)
}+\Lambda _{-}^{(-1)}\right) ,
\end{equation*}%
we easily find the differentials associated to the dynamical fields%
\begin{eqnarray}
\mu \cdot \psi &=&\left\langle G_{\mu }^{(-1/2)},L_{+}\right\rangle \text{ }%
\rightarrow \text{ \ \ \ \ }d_{+}\left( \mu \cdot \psi \right) \text{ }=%
\text{ }G_{\mu }^{(-1/2)},  \label{image fields differentials} \\
\nu \cdot \overline{\psi } &=&-\left\langle G_{\nu }^{(+1/2)},L_{-}^{\prime
}\right\rangle \text{ }\rightarrow \text{ \ }d_{-}^{\prime }\left( \nu \cdot 
\overline{\psi }\right) \text{ }=\text{ }-G_{\nu }^{(+1/2)},  \notag \\
x\cdot A_{+}^{(0)} &=&\left\langle M_{x}^{(0)},L_{+}\right\rangle \text{ }%
\rightarrow \text{ \ \ \ }d_{+}\left( x\cdot A_{+}^{(0)}\right) \text{ }=%
\text{ }M_{x}^{(0)},  \notag \\
y\cdot A_{-}^{(0)} &=&\left\langle M_{y}^{(0)},L_{-}^{\prime }\right\rangle 
\text{ }\rightarrow \text{ \ \ \ }d_{-}^{\prime }\left( y\cdot
A_{-}^{(0)}\right) \text{ }=\text{ }-M_{y}^{(0)}.  \notag
\end{eqnarray}

The Poisson form of the supersymmetry flow transformations $\delta _{\pm
1/2} $ of the field components is encoded in the following proposition.

\begin{proposition}
The supersymmetry transformations of the fields $\Phi (x)$ are Hamiltonian
flows on the reduced phase space$\ $and are induced by the supercharges $%
Q^{\pm }$ in terms of the second Hamiltonian structure%
\begin{equation}
\rho \cdot \delta _{+1/2}\Phi (x)=-\left\{ Q^{+},\Phi (x)\right\} _{2}(L),%
\text{ \ \ \ \ \ }\lambda \cdot \delta _{-1/2}\Phi (x)=+\left\{ Q^{-},\Phi
(x)\right\} _{2}(L^{\prime }),  \label{poisson susy flows}
\end{equation}%
where $\rho ,\lambda $ are constant even/odd parameters depending
respectively, if the field $\Phi $ is even/odd. This definition\ is
equivalent to the following matrix supersymmetry transformations%
\begin{eqnarray}
\left( \delta _{+1/2}BB^{-1}\right) ^{\parallel }{} &=&\left[ \psi ^{\left(
-1/2\right) },D^{(+1/2)}\right] \text{ ,}  \label{poisson susy 1/2} \\
\delta _{+1/2}\psi _{+}^{\left( +1/2\right) } &=&\left[ \left( \partial
_{+}BB^{-1}\right) ^{\parallel },D^{(+1/2)}\right] \text{ ,}  \notag \\
\delta _{+1/2}\psi _{-}^{\left( -1/2\right) } &=&-\left[ \Lambda
_{-}^{(-1)},B^{-1}D^{(+1/2)}B\right] ,\text{\ }  \notag
\end{eqnarray}%
and%
\begin{eqnarray}
\left( B^{-1}\delta _{-1/2}B\right) ^{\parallel } &=&-\left[ \psi ^{\left(
+1/2\right) },D^{(-1/2)}\right] \text{ ,}  \label{poisson susy -1/2} \\
\delta _{-1/2}\psi _{+}^{\left( +1/2\right) } &=&\left[ \Lambda
_{+}^{(+1)},BD^{(-1/2)}B^{-1}\right] \text{ \ ,}  \notag \\
\delta _{-1/2}\psi _{-}^{\left( -1/2\right) } &=&-\left[ \left(
B^{-1}\partial _{-}B\right) ^{\parallel },D^{(-1/2)}\right] .  \notag
\end{eqnarray}%
which should be compared with (\ref{wrong I}), (\ref{wrong II}) in the $%
A=\theta =0$ gauge.
\end{proposition}

\begin{proof}
We will prove for the $\delta _{+1/2}$ supersymmetry flow variations as for
the $\delta _{-1/2}$ the proof follows exactly the same lines. For the
components $\psi $ in $\psi _{+}^{\left( +1/2\right) }=\psi
_{i}G_{i}^{(+1/2)},$ we have (with $\mu $ fermionic) that 
\begin{equation*}
\mu \cdot \delta _{+1/2}\psi =-\left\{ m_{+},\mu \cdot \psi \right\}
_{2}(L_{+}),
\end{equation*}%
because $\psi $ and its differential vanishes when restricted to $%
L_{-}^{\prime }.$ Using the explicit form of the Lax connection 
\begin{equation*}
\Lambda _{\Theta }^{(+1)}=A_{+}^{(0)}+Q_{+}^{(0)}+\psi _{+}^{\left(
+1/2\right) }+\Lambda _{+}^{(+1)},
\end{equation*}%
the $r$-matrix $\mathcal{R}_{<}=\frac{1}{2}\left( \mathcal{P}_{0}-\mathcal{P}%
_{<0}\right) $ and the first equation in (\ref{simplified brackets}), we
conclude that there is not contribution from the derivative terms because $%
d_{+}\left( \mu \cdot \psi \right) $ has $Q$ grade $-1/2.$ Then, the
non-zero contribution to the bracket is\footnote{%
We have drop the integration of $\mu \cdot \psi (x)$ inside the bracket but
it have to be taken into account always when considering quantities in which
the derivative terms contribute.}%
\begin{equation*}
\left\{ m_{+},\mu \cdot \psi \right\} _{2}(L_{+})=\left\langle \Lambda
_{\Theta }^{(+1)},\left[ d_{+}m_{+},d_{+}(\mu \cdot \psi )\right] _{\mathcal{%
R}_{<}}\right\rangle =-\left\langle \Lambda _{+}^{(+1)},\left[ \left(
d_{+}m_{+}\right) _{-1/2},G_{\mu }^{(-1/2)}\right] \right\rangle .
\end{equation*}%
The only contribution from the differential $\left( d_{+}m_{+}\right)
_{-1/2} $ is the one along the image, which is given by $\left(
d_{+}m_{+}\right) _{-1/2}^{\parallel }=\left[ \psi ^{(-1)},D^{(+1/2)}\right]
.$ Now, using the ad-invariance of the inner product, the Jacobi identity,
the definition $A_{+}^{(0)}=\left[ \psi ^{(-1)},\Lambda _{+}^{(+1)}\right] $
and (\ref{image fields}) in the gauge $A=0,$ we find that%
\begin{equation*}
\text{ }\mu \cdot \delta _{+1/2}\psi =\left\langle \left[ \left( \partial
_{+}BB^{-1}\right) ^{\parallel },D^{(+1/2)}\right] ,G_{\mu
}^{(-1/2)}\right\rangle ,\text{ \ \ \ \ \ }\delta _{+1/2}\psi _{+}^{\left(
+1/2\right) }=\left[ \left( \partial _{+}BB^{-1}\right) ^{\parallel
},D^{(+1/2)}\right] ,
\end{equation*}%
where we have used $\mu \cdot \delta _{+1/2}\psi =\left\langle \delta
_{+1/2}\psi _{+}^{\left( +1/2\right) },G_{\mu }^{(-1/2)}\right\rangle $ in
order to isolate the supermatrix variation $\delta _{+1/2}\psi _{+}^{\left(
+1/2\right) }.$

Similarly, for the components $\overline{\psi }$ in $\psi _{-}^{\left(
-1/2\right) }=\overline{\psi }_{i}G_{i}^{(-1/2)},$ we get by the same
restriction argument applied this time to $\overline{\psi }$ on $L_{-},$
that 
\begin{equation*}
\nu \cdot \delta _{+1/2}\overline{\psi }=\left\{ m_{-},\nu \cdot \overline{%
\psi }\right\} _{2}(L_{-}).
\end{equation*}%
Using the explicit form of the Lax connection 
\begin{equation*}
\Lambda _{\Pi }^{(-1)}=B\left( \psi _{-}^{\left( -1/2\right) }+\Lambda
_{-}^{(-1)}\right) B^{-1},
\end{equation*}%
the projected components of the differential 
\begin{equation*}
\left( d_{-}m_{-}\right) _{+1/2}=-D^{(+1/2)},\text{ \ \ \ \ \ }d_{-}\left(
\nu \cdot \overline{\psi }\right) =-BG_{\nu }^{(+1/2)}B^{-1},
\end{equation*}%
the second line of (\ref{simplified brackets}) and noting that the
derivative terms does not contribute, we obtain%
\begin{equation*}
\left\{ m_{-},\nu \cdot \overline{\psi }\right\} _{2}(L_{-})=\left\langle
\Lambda _{\Pi }^{(-1)},\left[ d_{-}m_{-},d_{-}(\nu \cdot \overline{\psi })%
\right] _{\mathcal{R}_{>}}\right\rangle =\left\langle B\Lambda
_{-}^{(-1)}B^{-1},\left[ \left( d_{-}m_{-}\right) _{+1/2},BG_{\nu
}^{(+1/2)}B^{-1}\right] \right\rangle .
\end{equation*}%
From this, we easily obtain the variation%
\begin{equation*}
\text{ }\nu \cdot \delta _{+1/2}\overline{\psi }=-\left\langle \left[
\Lambda _{-}^{(-1)},B^{-1}D^{(+1/2)}B\right] ,G_{\nu }^{(+1/2)}\right\rangle
,\text{ \ \ \ \ \ }\delta _{+1/2}\psi _{-}^{\left( -1/2\right) }=-\left[
\Lambda _{-}^{(-1)},B^{-1}D^{(+1/2)}B\right] ,
\end{equation*}%
where we have used the ad-invariance of the inner product and $\nu \cdot
\delta _{+1/2}\overline{\psi }=\left\langle \delta _{+1/2}\psi _{-}^{\left(
-1/2\right) },G_{\nu }^{(+1/2)}\right\rangle $ to isolate $\delta
_{+1/2}\psi _{-}^{\left( -1/2\right) }.$

The transformation for $B$ is going to be found in an indirect way through
the transformation of the quantity $A_{+}^{(0)},$ which is more natural.
From the restriction of $A_{+}^{(0)}$ to $L_{+}$ we find (with $x$ bosonic)
that%
\begin{equation*}
x\cdot \delta _{+1/2}A_{+}^{(0)}=-\left\{ m_{+},x\cdot A_{+}^{(0)}\right\}
_{2}(L_{+}).
\end{equation*}%
The differentials entering the bracket are%
\begin{equation*}
\left( d_{+}m_{+}\right) _{0}=\left[ \psi ^{\left( -1/2\right) },D^{(+1/2)}%
\right] ,\text{ \ \ \ \ \ }d_{+}\left( x\cdot A_{+}^{(0)}\right) =M_{x}^{(0)}
\end{equation*}%
and by using (\ref{simplified brackets}) once more we obtain the non-zero
contributions to it%
\begin{eqnarray*}
\left\{ m_{+},x\cdot A_{+}^{(0)}\right\} _{2}(L_{+}) &=&\left\langle \Lambda
_{\Theta }^{(+1)},\left[ \left( d_{+}m_{+}\right) _{0},d_{+}\left( x\cdot
A_{+}^{(0)}\right) \right] \right\rangle +\left\langle \partial _{+}\left(
d_{+}m_{+}\right) _{0},d_{+}\left( x\cdot A_{+}^{(0)}\right) \right\rangle \\
&=&\left\langle \left[ Q_{+}^{(0)},\left[ \psi ^{\left( -1/2\right)
},D^{(+1/2)}\right] \right] +\partial _{+}\left[ \psi ^{\left( -1/2\right)
},D^{(+1/2)}\right] ,M_{x}^{(0)}\right\rangle
\end{eqnarray*}%
implying that%
\begin{equation}
\text{ }x\cdot \delta _{+1/2}A_{+}^{(0)}=-\left\langle \left[ Q_{+}^{(0)},%
\left[ \psi ^{\left( -1/2\right) },D^{(+1/2)}\right] \right] +\partial _{+}%
\left[ \psi ^{\left( -1/2\right) },D^{(+1/2)}\right] ,M_{x}^{(0)}\right%
\rangle .  \label{A}
\end{equation}

Now, in order to get the variation of the Toda field $B,$ we perform a $%
\delta _{+1/2}$ variation on the definition of $A_{+}^{(0)}$, i.e (\ref%
{image fields}) in the $A=0$ gauge$.$ This gives%
\begin{equation*}
\delta _{+1/2}A_{+}^{(0)}=-\partial _{+}\left( \delta _{+1/2}BB^{-1}\right)
^{\parallel }-\left[ \left( \delta _{+1/2}BB^{-1}\right) ^{\parallel
},\left( \partial _{+}BB^{-1}\right) ^{\perp }\right] -\left[ \left( \delta
_{+1/2}BB^{-1}\right) ^{\perp },\left( \partial _{+}BB^{-1}\right)
^{\parallel }\right] .
\end{equation*}%
Using (\ref{image fields}), the last line of (\ref{gauged fix eq of motion})
and projecting the resulting expression with $M_{x}^{(0)},$ we obtain%
\begin{equation}
\text{ }x\cdot \delta _{+1/2}A_{+}^{(0)}=-\left\langle \left[
Q_{+}^{(0)},\left( \delta _{+1/2}BB^{-1}\right) ^{\parallel }\right]
+\partial _{+}\left( \delta _{+1/2}BB^{-1}\right) ^{\parallel
},M_{x}^{(0)}\right\rangle -\left\langle \left[ A_{+}^{(0)},\left( \delta
_{+1/2}BB^{-1}\right) ^{\perp }\right] ,M_{x}^{(0)}\right\rangle .  \label{B}
\end{equation}%
Finally, by comparing (\ref{A}) and (\ref{B}) we conclude that%
\begin{equation*}
\left( \delta _{+1/2}BB^{-1}\right) ^{\parallel }=\left[ \psi ^{\left(
-1/2\right) },D^{(+1/2)}\right] ,\text{ \ \ \ \ \ }\left( \delta
_{+1/2}BB^{-1}\right) ^{\perp }=0.
\end{equation*}

Following exactly the same steps for the $\delta _{-1/2}$ variation we
arrive to the final result.
\end{proof}

At this point we have an aparent discrepancy between the transformations (%
\ref{poisson susy 1/2}), (\ref{poisson susy -1/2}) and (\ref{wrong I}), (\ref%
{wrong II}). One set is local while the other is not because of the presence
of the $\theta ^{(\pm 1/2)}$ terms (\ref{non-local teta}). The first set do
not preserve the constraints (\ref{new constraints}) while the second set
do. The easiest way to explain this difference is to work out the Noether
procedure on reverse\footnote{%
I want to thank Tim Hollowood for discussions on this point.}. We will do it
only for the $\delta _{+1/2}$ variations, for the obvious reason.

Consider the arbitrary variation (\ref{arbitrary variation}) with $\delta
=\delta _{+1/2}$ and without gauge fields, i.e the action functional (\ref%
{Abelian toda action}). Taking into account (\ref{U}) and (\ref{U ambiguous}%
) we get%
\begin{eqnarray}
\frac{2\pi }{k}\delta _{+1/2}S_{Toda}[B,\psi ] &=&\dint_{\Sigma
}\left\langle \left( \partial _{+}K_{-}^{(-1/2)}-\partial
_{-}K_{+}^{(-1/2)}\right) ,\left( U_{-}^{-1}\delta _{+1/2}\Phi U_{-}\right)
\right\rangle  \label{1/2 variation} \\
&=&\dint_{\Sigma }\left\langle \left( \partial _{+}\widetilde{K}%
_{-}^{(-1/2)}-\partial _{-}\widetilde{K}_{+}^{(-1/2)}\right) ,\left( 
\widetilde{S}_{-}^{-1}U_{-}^{-1}\delta _{+1/2}\Phi U_{-}\widetilde{S}%
_{-}\right) \right\rangle ,  \notag
\end{eqnarray}%
where $\delta _{+1/2}\Phi =\delta _{+1/2}BB^{-1}-B\delta _{+1/2}\psi
^{(+1/2)}B^{-1}-\delta _{+1/2}\psi ^{(-1/2)}$ and where we have to solve the
relations%
\begin{equation*}
\left( U_{-}^{-1}\delta _{+1/2}\Phi U_{-}\right) _{+1/2}^{\perp }=\left( 
\widetilde{S}_{-}^{-1}U_{-}^{-1}\delta _{+1/2}\Phi U_{-}\widetilde{S}%
_{-}\right) _{+1/2}^{\perp }=D^{(+1/2)}.
\end{equation*}%
Taking respectively, the following solutions 
\begin{equation*}
\delta _{+1/2}\Phi =\left( U_{-}D^{(+1/2)}U_{-}^{-1}\right) _{\geq
-1/2}=\left( U_{-}\widetilde{S}_{-}D^{(+1/2)}\widetilde{S}%
_{-}^{-1}U_{-}^{-1}\right) _{\geq -1/2},
\end{equation*}%
with $\widetilde{S}_{-}=\exp \left( \theta ^{(-1/2)}+\theta
^{(-1)}+...\right) $, $\theta ^{(r)}\in $ $\mathcal{K}$ and using (\ref%
{Baker-Hausdorf fields}), we get the local form (\ref{poisson susy 1/2}), (%
\ref{poisson susy -1/2}) corresponding to the local canonical DS
supercharges (\ref{mKdV supercharges}) and the non-local form (\ref{wrong I}%
), (\ref{wrong II}) corresponding to the gauge transformation of the
supercharges (\ref{mKdV supercharges}), respectively$.$ This shows that we
can always perform a compensating gauge transformation $U_{\pm }\rightarrow
U_{\pm }\widetilde{S}_{\pm }$ in order to preserve the constraints (\ref{new
constraints}) and this selects the special value of $\theta ^{(\pm 1/2)}$ to
be (\ref{non-local teta}). This means that the flows $\delta _{\pm 1/2}$ are
symmetry flows on the coadjoint orbits $L_{+},L_{-}^{\prime }.$

Let us compute the algebra of the supersymmetry flows (\ref{poisson susy 1/2}%
), (\ref{poisson susy -1/2}) in the perturbative limit $B=1+b.$ Using the
fermionic equations of motion given by (\ref{gauged fix eq of motion}), we
find, as expected from (\ref{charge algebra}) that%
\begin{equation}
\left[ \delta _{+1/2},\delta _{+1/2}^{\prime }\right] =2\text{ }\epsilon
\cdot \epsilon ^{\prime }\text{ }\partial _{+},\text{ \ \ \ \ \ }\left[
\delta _{-1/2},\delta _{-1/2}^{\prime }\right] =2\text{ }\overline{\epsilon }%
\cdot \overline{\epsilon }^{\prime }\text{ }\partial _{-},
\label{un-deformed}
\end{equation}%
where we have used $\left[ D^{(+1/2)},D^{\prime (+1/2)}\right] =2$ $\epsilon
\cdot \epsilon ^{\prime }\Lambda _{+}^{(+1)},$ $\left[ D^{(-1/2)},D^{\prime
(-1/2)}\right] =-2$ $\overline{\epsilon }\cdot \overline{\epsilon }^{\prime
}\Lambda _{-}^{(-1)},$ the jacobi identity and the relations (\ref%
{Baker-Hausdorf fields}). However, for the mixed terms we obtain%
\begin{equation}
\left[ \delta _{+1/2},\delta _{-1/2}\right] =\delta _{H},  \label{deformed}
\end{equation}%
where $\delta _{H}=\left[ \omega ,\right] $ is a gauge transformation with
parameter $\omega =\left[ D^{(+1/2)},D^{(-1/2)}\right] \in \mathcal{K}%
_{B}^{(0)}$. In principle, for the mixed bracket we should obtain $\left\{
Q^{+},Q^{-}\right\} _{2}\sim Q_{H},$ where $Q_{H}$ is the Noether charge
corresponding to the gauge group $H,$ but the corresponding Noether current
components $J^{\mu }$ are precisely the constraints (\ref{new constraints})
which vanishes identically. Then, the non-zero contribution to the conserved
charge must come from the ambiguous term in the definition of a conserved
current $\mathcal{J}^{\mu }$=$J^{\mu }+\varepsilon ^{\mu \nu }\partial _{\nu
}\mathcal{F}$, with $\mathcal{F}$ and arbitrary function. For this reason we
have not writen the mixed bracket $\left\{ Q^{+},Q^{-}\right\} _{2}$ above$.$
A more refined study of the supercharge algebra in terms of the subtracted
monodromy\footnote{%
In \cite{class-quant solitons} it was shown, by using monodromy matrix
arguments, that $\mathcal{F}$ is unambiguously fixed and that the conserved
charge $Q_{H}$ is of a kink type given by $B(+\infty )B(-\infty )^{-1}=\exp
Q_{H}.$} matrix will be done elsewhere \cite{us}. For the moment, we content
ourselves with the result (\ref{charge algebra}), but let us note that the
perturbative excitations $b^{\parallel },\psi _{\pm }^{(\pm 1/2)}$ do
transform under the kernel algebra $\mathcal{K=}$ $\widehat{\mathfrak{f}}%
^{\perp }$ through the variations $\delta _{\mathcal{K}}.$

Let us end this section by writing the brackets for the fundamental fields $%
B $ and $\psi ^{(\pm 1/2)}$. To do this, it is more convenient to use (\ref%
{Lax operators}) and (\ref{prime Lax operators}) in the gauge $A=0.$
Defining the currents $J_{+}=-\partial _{+}BB^{-1},$ $J_{-}=B^{-1}\partial
_{-}B,$ which have to be considered as functionals on $L_{+}$ and $%
L_{-}^{\prime }$ respectively, and introducing a normalized basis for $%
\widehat{\mathfrak{f}}_{0}=span(T_{a}^{(0)})$ with $\left\langle
T_{a}^{(0)},T_{b}^{(0)}\right\rangle =\delta _{ab},$ we find the current
differentials%
\begin{equation*}
d(\lambda \cdot J_{+})=T_{\lambda }^{(0)},\text{ \ \ \ \ \ }d_{-}^{\prime
}(\lambda \cdot J_{-})=T_{\lambda }^{(0)},
\end{equation*}%
where $\lambda \cdot J_{\pm }=\lambda _{a}J_{\pm }^{a},$ $T_{\lambda
}^{(0)}=\lambda _{a}T_{a}^{(0)},$ $a=1,...,\dim \widehat{\mathfrak{f}}^{(0)}$
with $\lambda $ bosonic constants$.$ Adding the fermionic field
differentials (\ref{image fields differentials}) and using (\ref{spatial 2
bracket}), (\ref{simplified brackets}) we find the brackets%
\begin{eqnarray}
\left\{ \lambda \cdot J_{+}(x),\rho \cdot J_{+}(y)\right\} _{2}(L) &=&\delta
\left( x-y\right) \left\langle J_{+}(x),\left[ T_{\lambda }^{(0)},T_{\rho
}^{(0)}\right] \right\rangle -\delta ^{\prime }\left( x-y\right)
\left\langle T_{\lambda }^{(0)},T_{\rho }^{(0)}\right\rangle ,
\label{SSSSG model bracket} \\
\left\{ \mu \cdot \psi (x),\nu \cdot \psi (y)\right\} _{2}(L) &=&-\mu \cdot
a^{+}\cdot \nu \delta \left( x-y\right) ,  \notag
\end{eqnarray}%
where we have assumed that $\left[ \Lambda _{+}^{(+1)},G_{i}^{(-1/2)}\right]
=a_{ij}^{+}G_{j}^{(+1/2)}.$ For $J_{-}$ and $\overline{\psi },$ the brackets
in $L^{\prime }$ are quite similar but with $a_{ij}^{+}\rightarrow
-a_{ij}^{-}.$ Then, the second Hamiltonian structure is of Kac-Moody type%
\footnote{%
This was already noticed in \cite{generalized DS II} in the purely bosonic
case. See the last example worked out in that reference.}.

Finally, if we compute the canonical brackets for the action functional (\ref%
{Abelian toda action}) we arrive to the conclusion that $\left\{ X,Y\right\}
_{2}\sim \left\{ X,Y\right\} _{WZNW}.$ It would be interesting to
investigate more deeply the relation between the brackets $\left\{
X,Y\right\} _{1,2}$ and the brackets $\left\{ X,Y\right\} ^{[\pm 2],[0]}$
introduced in \cite{Mikhailov}, which studies the connection between the
Hamiltonian structures of the Green-Schwarz action and the generalized
sine-Gordon models involved in the reduction of the $AdS_{5}\times S^{5}$
superstring sigma model. This will be done elsewhere \cite{us}.

\section{Pohlmeyer reduction of GS superstring sigma models.}

In this chapter we briefly review the steps involved in the reduction of
Green-Schwarz (GS) superstring sigma models on semi-symmetric superspaces.
The aim is to show how emerges the extended homogeneous hierarchy integrable
structure defined above and also to show the connection between the number
of fermionic symmetry flows with the elements in the rank of the kappa
symmetry.

A coset $F/G$ of a supergroup $F$ is a semi-symmetric superspace if it is
invariant under a $%
\mathbb{Z}
_{4}$ symmetry and the superalgebra admits the decomposition (\ref{Z4
grading}), $\mathfrak{f=f}_{0}\mathfrak{\oplus f}_{1}\mathfrak{\oplus f}_{2}%
\mathfrak{\oplus f}_{3},$ which is consistent with the relations $[\mathfrak{%
f}_{i},\mathfrak{f}_{j}]\subset \mathfrak{f}_{(i+j)\func{mod}4},$ the
subspace $\mathfrak{f}_{j}$ is defined by $\Omega (\mathfrak{f}_{j})=i^{j}$ $%
\mathfrak{f}_{j}$ and the denominator subalgebra is the invariant subspace $%
\mathfrak{f}_{0}.$

We assume that the following decomposition of the bosonic subalgebra $%
\mathfrak{f}_{B}=\mathfrak{f}_{0}\mathfrak{\oplus f}_{2}$ also holds%
\begin{eqnarray}
\mathfrak{f}_{0} &=&\mathfrak{m}_{0}\oplus \mathfrak{h}_{0},\text{ \ }%
\mathfrak{f}_{2}\text{ }=\text{ }\mathfrak{a}_{2}\oplus \mathfrak{n}_{2},%
\text{ \ }\left[ \mathfrak{a}_{2},\mathfrak{a}_{2}\right] \text{ }=\text{ }0,%
\text{ \ }\left[ \mathfrak{a}_{2},\mathfrak{h}_{0}\right] \text{ }=\text{ }0,%
\text{ \ }\left[ \mathfrak{h}_{0},\mathfrak{h}_{0}\right] \text{ }\subset 
\text{ }\mathfrak{h}_{0},\text{\ }  \label{finite decomposition} \\
\left[ \mathfrak{m}_{0},\mathfrak{m}_{0}\right] &\subset &\mathfrak{h}_{0},%
\text{ \ }\left[ \mathfrak{h}_{0},\mathfrak{m}_{0}\right] \text{ }\subset 
\text{ }\mathfrak{m}_{0},\text{ \ }\left[ \mathfrak{m}_{0},\mathfrak{a}_{2}%
\right] \text{ }\subset \text{ }\mathfrak{n}_{2},\text{ \ }\left[ \mathfrak{a%
}_{2},\mathfrak{n}_{2}\right] \text{ }\subset \text{ }\mathfrak{m}_{0}, 
\notag
\end{eqnarray}%
where $\mathfrak{a}_{2}$ is a maximal Abelian subalgebra and $\mathfrak{h}%
_{0}$ its centralizer. The algebra $\mathfrak{h}_{0}$ turns out to be
related to the gauge flows $H_{L}\times H_{R}$ introduced above.

The action of the sigma model can be written in terms of the $%
\mathbb{Z}
_{4}$ decomposition of the current%
\begin{equation}
J_{\mu }=f^{-1}\partial _{\mu }f=\mathcal{A}_{\mu }+Q_{1\mu }+P_{\mu
}+Q_{3\mu },  \label{z current decomposition}
\end{equation}%
where $\mathcal{A}_{\mu }\in \mathfrak{f}_{0},$ $Q_{1\mu }\in \mathfrak{f}%
_{1},$ $P_{\mu }\in \mathfrak{f}_{2},$ $Q_{3\mu }\in \mathfrak{f}_{3}.$ The
current $J$ is invariant under global left $F$-gauge transformations $%
f(x)\rightarrow f^{\prime }f(x)$ and under local right $G$-gauge
transformations $f(x)\rightarrow f(x)g(x),$ the component $\mathcal{A}$
transform as a gauge connection, while the others components transform
covariantly (in the adjoint).

The action of the non linear sigma model must be gauge invariant and $%
\mathbb{Z}
_{4}$ symmetric and we are interested here in those models which are
described by the following (Green-Schwarz) action%
\begin{equation}
S_{GS}=\frac{1}{2\kappa ^{2}}\dint \left\langle \gamma ^{\mu \nu }P_{\mu
}P_{\nu }+\epsilon ^{\mu \nu }Q_{1\mu }Q_{3\nu }\right\rangle ,
\label{GS lagrangian}
\end{equation}%
where $\gamma ^{\mu \nu }=\sqrt{-g}g^{\mu \nu },$ $g_{\mu \nu }$ is the two
dimensional world-sheet metric, $\epsilon _{\mu \nu }$ the anti-symmetric
symbol and $\kappa $ is the sigma model coupling. The action (\ref{GS
lagrangian}) is invariant under 2D conformal transformations, $G$-gauge
transformations and $\kappa $-symmetry, which is a local fermionic gauge
transformation and all this means that some of the bosonic and fermionic
degrees of freedom in the lagrangian (\ref{GS lagrangian}) are un-physical.

The Pohlmeyer reduction consist of a classical removal of all the
un-physical degrees of freedom and such a reduction is performed by gauge
fixing all the above mentioned symmetries. We now briefly show how this is
achieved, for full details see the original paper \cite{tseytlin}, see also 
\cite{Mikha II}.

Start by writing (\ref{GS lagrangian}) in the conformal gauge\footnote{%
We use here the light-cone convention $\gamma ^{+-}=\gamma ^{-+}=1$ and $%
\epsilon ^{+-}=-\epsilon ^{-+}=1$ of \cite{tseytlin}.}%
\begin{equation}
\mathcal{L}_{GS}=\left\langle P_{+}P_{-}+\frac{1}{2}\left(
Q_{1+}Q_{3-}-Q_{1-}Q_{3+}\right) \right\rangle .
\label{light cone GS lagrangian}
\end{equation}%
The classical field theory is describe by the equations of motion extracted
from varying the action functional of (\ref{light cone GS lagrangian}) 
\begin{eqnarray}
0 &=&\partial _{+}P_{-}+\left[ \mathcal{A}_{+},P_{-}\right] +\left[
Q_{3+},Q_{3-}\right] ,  \label{sigma model eq of motion} \\
0 &=&\partial _{-}P_{+}+\left[ \mathcal{A}_{-},P_{+}\right] +\left[
Q_{1-},Q_{1+}\right] ,  \notag \\
0 &=&\left[ P_{+},Q_{1-}\right] =\left[ P_{-},Q_{3+}\right]  \notag
\end{eqnarray}%
and by the Maurer-Cartan identity for the flat current (\ref{z current
decomposition})%
\begin{equation}
\partial _{+}J_{-}-\partial _{-}J_{+}+\left[ J_{+},J_{-}\right] =0,
\label{maurer-cartan}
\end{equation}%
both supplemented with the Virasoro constraints%
\begin{equation}
\left\langle P_{+},P_{+}\right\rangle =0,\text{ \ \ \ \ \ }\left\langle
P_{-},P_{-}\right\rangle =0.  \label{virasoro constraints}
\end{equation}

Now, by using the polar decomposition theorem and the $G$-gauge freedom on
the coset we can go to the so-called reduction gauge in which the components
of the current $P_{\mu }\in \mathfrak{f}_{2}$ can be put in the following
form 
\begin{equation}
P_{+}=\mu _{+}(x^{+})\Lambda _{+},\text{ \ \ \ \ \ }P_{-}=\mu
_{-}(x^{-})B\Lambda _{-}B^{-1},  \label{x-dependet virasoro sol}
\end{equation}%
where $\Lambda _{\pm }\in \mathfrak{a}_{2}$ are constant elements in $%
\mathfrak{f},$ $\mu _{\pm }(x^{\pm })$ are functions of $x^{\pm }$ only and $%
B\in G=\exp \mathfrak{f}_{0}$ belongs to the sigma model gauge group, i.e
the denominator group in the coset $F/G.$ For the semi-symmetric-space
cosets of interest, we have that $\dim \mathfrak{a}_{2}=1$ meaning that $%
\Lambda _{+}=\Lambda _{-}=\Lambda $ is unique$.$ We use this constant
element $\Lambda ,$ which is also semisimple$,$ to introduce the $%
\mathbb{Z}
_{2}$ superalgebra decomposition (\ref{kernel-image finite}) with $\mathfrak{%
f}=$ $\mathfrak{f}^{\perp }\oplus \mathfrak{f}^{\parallel }.$

The $\kappa $-symmetry can be partially fixed by the gauge condition $%
Q_{1-}=0$ and $Q_{3+}=0$ and this simplifies the equations (\ref{sigma model
eq of motion})$.\ $Replacing (\ref{x-dependet virasoro sol}) into (\ref%
{sigma model eq of motion}), using (\ref{virasoro constraints}) and the
residual conformal transformations we can set the functions $\mu _{\pm
}(x^{\pm })$ to constants $\mu _{\pm }$ . The remaining equations of motion
in (\ref{sigma model eq of motion}) are equations for the gauge field
components $\mathcal{A}_{\pm }$ only and allow the following solution%
\begin{equation}
\mathcal{A}_{+}=-(\partial _{+}BB^{-1}+BA_{+}^{(R)}B^{-1}),\text{ \ \ \ \ \ }%
\mathcal{A}_{-}^{\perp }=-A_{-}^{(L)},\text{ \ \ \ \ \ }\mathcal{A}%
_{-}^{\parallel }=0,  \notag
\end{equation}%
where $A_{+}^{(R)},A_{-}^{(L)}\in \mathfrak{f}_{0}^{\perp }=\mathfrak{h}_{0}$
belong now to a subalgebra of the former sigma model gauge algebra $%
\mathfrak{f}_{0}$, see(\ref{finite decomposition}).$\ $All this exhaust the
equations of motion in (\ref{sigma model eq of motion}).

The solutions of the Virasoro constraints (\ref{virasoro constraints}) are
now\footnote{%
From now we normalize $\mu _{\pm }\rightarrow 1$.}%
\begin{equation}
P_{+}=\mu _{+}\Lambda _{+},\text{ \ \ \ \ \ }P_{-}=\mu _{-}B\Lambda
_{-}B^{-1}\text{\ }  \label{virasoro solution}
\end{equation}%
and the gauge group that preserves these solutions and the Virasoro
"surface" (\ref{virasoro constraints}) is precisely the $H_{L}\times H_{R}$
gauge symmetry introduced above, $\widetilde{B}=\Gamma _{L}B\Gamma _{R}$.
This explains the use of the notation $L,R.$

In terms of the new bosonic field variables $A_{+}^{(R)},A_{-}^{(L)}\in 
\mathfrak{f}_{0}^{\perp },$ $B\in G$ and in the gauge $Q_{1-}=0,$ $Q_{3+}=0,$
the (remaining) Maurer-Cartan equations (\ref{maurer-cartan}) can be put in
the form 
\begin{eqnarray}
D_{-}^{(L)}Q_{1+} &=&\left[ \Lambda _{+},Q_{3-}\right] ,
\label{pre final maurer cartan} \\
D_{-}^{(L)}\left( \partial _{+}BB^{-1}+BA_{+}^{(R)}B^{-1}\right) -\partial
_{+}A_{-}^{(L)} &=&-\left[ \Lambda _{+},B\Lambda _{+}B^{-1}\right] -\left[
Q_{1+},Q_{3-}\right] ,  \notag \\
D_{+}^{(R)}\left( B^{-1}Q_{3-}B\right) &=&\left[ \Lambda _{-},B^{-1}Q_{1+}B%
\right] ,  \notag
\end{eqnarray}%
where $D_{-}^{(L)}=\partial _{-}-\left[ A_{-}^{(L)},\right] ,$ $%
D_{+}^{(R)}=\partial _{+}-\left[ A_{+}^{(R)},\right] $ are the covariant
derivatives for the $H_{L}\times H_{R}$ action of the gauge group and $%
\widetilde{Q}_{1+}=\Gamma _{L}Q_{1+}\Gamma _{L}^{-1},$ $\widetilde{Q}%
_{3-}=\Gamma _{L}Q_{3-}\Gamma _{L}^{-1}$ are the gauge transformations for
the remaining fermionic current components$.$ At this point it is worth it
to compare these equations with (\ref{NA affine toda equations}).

We now make the following change of field variables $Q_{1+},Q_{3-}%
\rightarrow \Psi _{1},\Psi _{3}$ defined by 
\begin{equation*}
\Psi _{1}=Q_{1+}\text{, \ \ \ \ \ }\Psi _{3}=-B^{-1}Q_{3-}B.\text{\ }
\end{equation*}%
The $\kappa $-symmetry is completely fixed by putting to zero all the
components of $\Psi _{1},\Psi _{3}$ belonging to the fermionic part of
kernel, i.e $\mathfrak{f}_{1}^{\perp },\mathfrak{f}_{3}^{\perp }.\ $The
remaining components in the fermionic part in the Image, i.e $\mathfrak{f}%
_{1}^{\parallel },\mathfrak{f}_{3}^{\parallel }$ are the truly fermionic
physical degrees of freedom%
\begin{equation}
\Psi _{L}=\Psi _{1}^{\parallel },\text{ \ \ \ \ \ }\Psi _{R}=\Psi
_{3}^{\parallel }.  \label{physical fermions}
\end{equation}

In terms of the Pohlmeyer reduced model field content $%
A_{+}^{(R)},A_{-}^{(L)},$ $B$, $\Psi _{R}$ and $\Psi _{L},$ the
Maurer-Cartan equations (\ref{pre final maurer cartan}) can be written, with
the help of a spectral parameter $z,$ as the compatibility condition of the
following Lax pair%
\begin{eqnarray}
L_{+}(A) &=&\partial _{+}-\partial
_{+}BB^{-1}-BA_{+}^{(R)}B^{-1}+(z^{+1/2}\Psi _{L})+(z\Lambda _{+}),
\label{Tseytlin Lax} \\
L_{-}(A) &=&\partial _{-}-A_{-}^{(L)}-B\left[ (z^{-1/2}\Psi
_{R})+(-z^{-1}\Lambda _{-})\right] B^{-1}.  \notag
\end{eqnarray}%
Making the following identifications%
\begin{eqnarray*}
(z^{+1/2}\Psi _{R}) &=&\psi _{+}^{\left( +1/2\right) }\text{, \ \ \ \ \ \ \
\ \ \ \ }(z\Lambda _{+})\text{ }=\text{ }\Lambda _{+}^{(+1)}, \\
\text{\ }(z^{-1/2}\Psi _{L}) &=&\psi _{-}^{\left( -1/2\right) }\text{, \ \ \
\ \ }(-z^{-1}\Lambda _{-})\text{ }=\text{ }\Lambda _{-}^{(-1)}
\end{eqnarray*}%
and comparing with (\ref{Lax operators}) we see that the Lax operators (\ref%
{Tseytlin Lax}) and (\ref{Lax pair for action}) describe the same integrable
hierarchy (the extended homogeneous hierarchy). From (\ref{loop algebra}) we
see that the fermionic identifications are consistent%
\begin{equation*}
(z^{+1/2}\Psi _{L})=\left[ \psi ^{\left( -1/2\right) },\Lambda _{+}^{(+1)}%
\right] ,\text{ \ \ \ \ \ }(z^{-1/2}\Psi _{R})=-\left[ \psi ^{\left(
+1/2\right) },\Lambda _{-}^{(-1)}\right] \text{\ }
\end{equation*}%
because $\Psi _{L}\in \mathfrak{f}_{1}^{\parallel },$ $\psi ^{\left(
-1/2\right) }\in \widehat{\mathfrak{f}}_{3}^{\parallel }$ and $\Psi _{R}\in 
\mathfrak{f}_{3}^{\parallel },$ $\psi ^{\left( +1/2\right) }\in \widehat{%
\mathfrak{f}}_{1}^{\parallel }$. Recall that $ad(\Lambda )$ maps the
subspaces $\mathfrak{f}_{1,3}\rightarrow \mathfrak{f}_{3,1}$ each other. To
obtain the equivalent formulation $L_{\pm }^{\prime }(A)$ we use the $B$%
-conjugated solution of (\ref{virasoro solution}).

These equations are invariant under the gauge group $H_{L}\times H_{R}$ and
all the results found above also applies here$.$ Note that $\mathfrak{f}%
_{0}^{\perp }=\mathfrak{h}_{0}\subset \ker \left( ad(\Lambda )\right) $
generates gauge flows associated to grade zero elements after embedding $%
\mathfrak{f}$ into the loop algebra $\widehat{\mathfrak{f}}$ in the form (%
\ref{loop algebra})$.$ The Pohlmeyer reduction in now clear. It states that (%
\ref{GS lagrangian}) and (\ref{NA Toda action}) describes the same classical
field theory\footnote{%
The action (\ref{NA Toda action}) is slightly different than the one
originally constructed in \cite{tseytlin} in the particular case of the $%
AdS_{5}\times S^{5}$ superstring sigma model, the difference is in the
potential term and it is not essential.}. Note that the net effect of the
reduction is to trade the Euler-Lagrange equations (\ref{sigma model eq of
motion}) of (\ref{GS lagrangian}) by the Maurer-Cartan identities (\ref%
{maurer-cartan}) with associated Lax Pair (\ref{Tseytlin Lax}), which
correspond now to the Euler-Lagrange equations of the SSSSG action (\ref{NA
Toda action}). The local gauge symmetry is reduced from the right $G$-gauge
action on the coset $F/G$ of the GS sigma model to the left-right $H$-gauge
action on the coset $G/H$ of the gauged WZNW model coupled to fermions and
the fields which include the un-physical degrees of freedom are reduced from 
$P_{\mu },Q_{1\mu },Q_{3\mu }$ to $B,\psi _{\pm }^{(\pm 1/2)},A_{\pm
}^{(R/L)},$ which have physical degrees of freedom only$.$ Concerning the
local $\kappa $-symmetry, it seems to be that its global remnant is related
to the existence of 2D world-sheet extended supersymmetry in the reduced
models. We do not have a formal proof of this statement but we will provide
some evidence that this is the case because, as discussed above, the
elements in the fermionic kernel $\mathcal{K}_{F}^{(\pm 1/2)}$ generate $dim%
\mathcal{K}_{F}^{(\pm 1/2)}$ global symmetry flows with conserved
supercharges (\ref{AKNS supercharges}). Note that physical fields are
parametrized in $\mathcal{M}$ while the symmetries are associated to the
sub-superalgebra $\mathcal{K\subset }$ $\widehat{\mathfrak{f}}$ and we can
see now the role of the dressing flow equations (\ref{flow equations}), (\ref%
{gauge equivalent variations}): they generate global world-sheet symmetries $%
\delta _{\mathcal{K}}$ in the reduced model from the loop algebra $\mathcal{K%
}$ constructed out of the subalgebra $\mathfrak{f}^{\perp }\subset \mathfrak{%
f}$ of the global target space symmetry of the sigma model$.$

We consider now the rank of the $\kappa $-symmetry. Following \cite{zarembo}%
, it is defined to be the number of fermionic generators in $\mathfrak{f}%
_{1} $ and $\mathfrak{f}_{3}$ which are annihilated by the adjoint action of 
$P_{\pm }\in \mathfrak{f}_{2}$%
\begin{equation}
N_{\kappa }=\dim \ker \left( ad(P_{-})\mid \mathfrak{f}_{3}\right) ,\text{ \
\ \ \ \ }N_{\widetilde{\kappa }}=\dim \ker \left( ad(P_{+})\mid \mathfrak{f}%
_{1}\right)  \label{kappa rank}
\end{equation}%
with $P_{\pm }$ satisfying the Virasoro constraints (\ref{virasoro
constraints}). We assume that we are in the situation of generic classical
solutions which is when the number of zero modes $N_{\kappa },$ $N_{%
\widetilde{\kappa }}$ are field (i.e $P_{\pm }$) independent and are
determined entirely by the Lie superalgebra properties of the algebra $%
\mathfrak{f}$ of interest. This means that we can use $ad\left( \Lambda
_{\pm }\right) $ instead of $ad(P_{\pm })$ to compute the dimensions (\ref%
{kappa rank}). In \cite{zarembo} the dimensions (\ref{kappa rank}) were
calculated for several Lie superalgebras admitting the $%
\mathbb{Z}
_{4}$ decomposition (\ref{Z4 grading}). Here we write some of those
dimensions of our interest%
\begin{eqnarray}
AdS_{2}\times S^{2} &\rightarrow &N_{\kappa }=N_{\widetilde{\kappa }}=2,
\label{table} \\
AdS_{3}\times S^{3} &\rightarrow &N_{\kappa }=N_{\widetilde{\kappa }}=4, 
\notag \\
AdS_{4}\times 
\mathbb{C}
P^{3} &\rightarrow &N_{\kappa }=N_{\widetilde{\kappa }}=4,  \notag \\
AdS_{5}\times S^{5} &\rightarrow &N_{\kappa }=N_{\widetilde{\kappa }}=8. 
\notag
\end{eqnarray}%
These are exactly the number of supersymmetry flows generated by the
supercharges (\ref{AKNS supercharges}). In our notation and for the first
two models in the list we have $\dim \mathcal{K}_{F}^{(\pm 1/2)}=2$ and $%
\dim \mathcal{K}_{F}^{(\pm 1/2)}=4,$ respectively. See (\ref{psu susy flows
I}),(\ref{psu susy flows II}),(\ref{psuxpsu susy flows I}) and (\ref{psuxpsu
susy flows II}) for the explicit expressions of the generators in $\mathcal{K%
}_{F}$.

The $AdS_{5}\times S^{5}$ case is treated in detail in \cite{SSSSG
AdS(5)xS(5)}, where it is shown that the symmetry algebra of the solitonic
spectrum of the reduced model, in semi-classical quantization, is precisely
the kernel algebra $\mathcal{K}$, which in this case turns out to be
isomorphic to a centrally extended $su(2\mid 2)^{\times 2}$ superalgebra$.$
This algebra has $8+8$ fermionic elements in its odd part generating $16$
supersymmetry flows and $12$ bosonic elements in its even part generating
the gauge algebra $su(2)^{\times 4}$. A similar set of conserved
supercharges (\ref{AKNS supercharges}) and supersymmetry transformations (%
\ref{wrong I}), (\ref{wrong II}) in the on-shell $A_{\pm }=0$ gauge are also
constructed for this case.

\section{Examples: supercharges from superspace and from symmetry flows.}

Because of the supersymmetry we are dealing with is quite non-standard, it
is important to study its relation with the usual supersymmetry obtained
from superspace by working out some examples. Here we consider the Pohlmeyer
reduction of the superstring on $AdS_{2}\times S^{2}$ and $AdS_{3}\times
S^{3}$ which are already known$.$ As we are mainly interested in
understanding the role of the supercharges found above, we will try to be as
close as possible to the superspace results. However, we have to mention
that in the $AdS_{3}\times S^{3}$ we will ignore the Wilson lines appearing
in (\ref{AKNS supercharges}) and make the computations in some field limits.

\subsection{Supercharges of the Landau-Ginzburg models.}

\begin{notation}
In this section, the light-cone notation used is $x^{\pm }=x^{0}\pm x^{1}$, $%
\partial _{\pm }=\frac{1}{2}\left( \partial _{0}\pm \partial _{1}\right) $, $%
\eta _{+-}=\eta _{-+}=-\frac{1}{2},\eta ^{+-}=\eta ^{-+}=-2,$ corresponding
to the metric $\eta _{00}=-1,\eta _{11}=+1.$
\end{notation}

A Landau-Ginzburg model is defined by a Lagrangian density of the from (e.g
see \cite{Mirror book})%
\begin{equation}
\mathcal{L=}\dint d^{4}\theta K(\Phi ^{i},\overline{\Phi }^{\overline{i}})+%
\frac{1}{2}\left( \dint d^{2}\theta W(\Phi ^{i})+c.c\right) \text{ \ , \ }i,%
\overline{i}=1,...,n,  \label{LG model}
\end{equation}%
where $\Phi =\phi (y^{\pm })+\theta ^{\alpha }\psi _{\alpha }(y^{\pm
})+\theta ^{+}\theta ^{-}F(y^{\pm })$ is a chiral superfield, $y^{\pm
}=x^{\pm }-i\theta ^{\pm }\overline{\theta }^{\pm }$, $\overline{\theta }%
^{\pm }\equiv (\theta ^{\pm })^{\ast }$ and $(\psi _{1}\psi _{2})^{\ast
}\equiv (\psi _{2}^{\ast }\psi _{1}^{\ast })$ stands for the complex
conjugation convention acting on fermions$.$

In components, the lagrangian density is%
\begin{eqnarray}
\mathcal{L} &=&-g_{i\overline{j}}\partial ^{\mu }\phi ^{i}\partial _{\mu }%
\overline{\phi }^{\overline{j}}+2ig_{i\overline{j}}\overline{\psi }_{-}^{%
\overline{j}}D_{+}\psi _{-}^{i}+2ig_{i\overline{j}}\overline{\psi }_{+}^{%
\overline{j}}D_{-}\psi _{+}^{i}+R_{i\overline{j}k\overline{l}}\psi
_{+}^{i}\psi _{-}^{k}\overline{\psi }_{-}^{\overline{j}}\overline{\psi }%
_{+}^{\overline{l}}-  \notag \\
&&-\frac{1}{4}g^{\overline{i}j}\partial _{\overline{i}}\overline{W}\partial
_{j}W-\frac{1}{2}D_{i}\partial _{j}W\psi _{+}^{i}\psi _{-}^{j}-\frac{1}{2}D_{%
\overline{i}}\partial _{\overline{j}}\overline{W}\overline{\psi }_{-}^{%
\overline{i}}\overline{\psi }_{+}^{\overline{j}},  \label{LG lagrangian}
\end{eqnarray}%
where $g_{i\overline{j}}=\partial _{i}\partial _{\overline{j}}K(\phi ^{i},%
\overline{\phi }^{\overline{i}})$, $K$ is the K\"{a}hler potential and $%
D_{\mu }\psi _{\pm }^{i}=\partial _{\mu }\psi _{\pm }^{i}+\partial _{\mu
}\phi ^{j}\Gamma _{jk}^{i}\psi _{\pm }^{k}$ , $D_{i}\partial _{j}W=\partial
_{i}\partial _{j}W-\Gamma _{ij}^{k}\partial _{k}W.$ \ The $(2,2)$ Noether
supercharges associated to the model (\ref{LG lagrangian}) are given by%
\begin{equation*}
Q_{\pm }=\dint_{-\infty }^{+\infty }dx^{1}G_{\pm }^{0},\text{ \ \ \ \ \ }%
G_{\pm }^{0}=2g_{i\overline{j}}\partial _{\pm }\overline{\phi }^{\overline{j}%
}\psi _{\pm }^{i}\mp \frac{i}{2}\overline{\psi }_{\mp }^{\overline{i}%
}\partial _{\overline{i}}\overline{W}.
\end{equation*}

\subsubsection{The $(2,2)$ sine-Gordon model.}

The first model of interest is when $i=1,$ $K(\Phi ^{1},\overline{\Phi }^{%
\overline{1}})=\overline{\Phi }\Phi $ and $W\rightarrow 2W$ in the
Lagrangian (\ref{LG lagrangian}). In this case we have%
\begin{equation}
\mathcal{L}=2\left( \partial _{+}\varrho \partial _{-}\overline{\varrho }%
+\partial _{-}\varrho \partial _{+}\overline{\varrho }\right) +2i\overline{%
\psi }_{-}\partial _{+}\psi _{-}+2i\overline{\psi }_{+}\partial _{-}\psi
_{+}-\left\vert W^{\prime }(\varrho )\right\vert ^{2}-\left[ W^{\prime
\prime }(\varrho )\psi _{+}\psi _{-}+\overline{W}^{\prime \prime }(\overline{%
\varrho })\overline{\psi }_{-}\overline{\psi }_{+}\right] ,
\label{flat lagrangian}
\end{equation}%
where we have denoted by $\varrho $ the complex scalar field component of
the superfield $\Phi .$ The densities are%
\begin{equation}
G_{\pm }^{0}=2\partial _{\pm }\overline{\varrho }\psi _{\pm }\mp i\overline{%
\psi }_{\mp }\overline{W}^{\prime }(\overline{\varrho }).
\label{flat supercharges}
\end{equation}

Taking now the following choice of field components 
\begin{equation*}
\varrho =\varphi +i\phi ,\text{ \ \ }\psi _{-}=\lambda _{1}+i\lambda _{2},%
\text{\ \ \ }\psi _{+}=\overline{\lambda }_{1}+i\overline{\lambda }_{2},%
\text{ \ \ }W(\varrho )=2\mu \cos \varrho ,
\end{equation*}%
we have from (\ref{flat lagrangian}), the $N=(2,2)$ supersymmetric extension
of the sine-Gordon model \cite{uematsu}%
\begin{eqnarray}
\frac{1}{4}\mathcal{L} &=&\partial _{+}\phi \partial _{-}\phi +\partial
_{+}\varphi \partial _{-}\varphi +\frac{i}{2}\left( \lambda _{1}\partial
_{-}\lambda _{1}+\lambda _{2}\partial _{-}\lambda _{2}+\overline{\lambda }%
_{1}\partial _{+}\overline{\lambda }_{1}+\overline{\lambda }_{2}\partial _{+}%
\overline{\lambda }_{2}\right) -\frac{\mu ^{2}}{2}\left( \cosh 2\phi -\cos
2\varphi \right) +  \label{N=2 SS sine-Gordon lagrangian} \\
&&+\mu i\left\{ \cos \varphi \cosh \phi \left( \lambda _{1}\overline{\lambda 
}_{2}+\lambda _{2}\overline{\lambda }_{1}\right) -\sin \varphi \sinh \phi
\left( \lambda _{1}\overline{\lambda }_{1}-\lambda _{2}\overline{\lambda }%
_{2}\right) \right\} .  \notag
\end{eqnarray}

The fermionic densities (\ref{flat supercharges}) can be written as $G_{\pm
}^{0}=2\left( q_{1}^{\pm }+iq_{2}^{\pm }\right) $ in terms of the following $%
2+2$ real components\footnote{%
SS means Super-Space.}%
\begin{eqnarray}
^{SS}q_{1}^{+} &=&\lambda _{1}\partial _{+}\varphi +\lambda _{2}\partial
_{+}\phi +\mu \left( \overline{\lambda }_{1}\cosh \phi \sin \varphi +%
\overline{\lambda }_{2}\sinh \phi \cos \varphi \right) ,
\label{N=2 SS sine-Gordon supercharges} \\
^{SS}q_{2}^{+} &=&-\lambda _{1}\partial _{+}\phi +\lambda _{2}\partial
_{+}\varphi +\mu \left( -\overline{\lambda }_{1}\sinh \phi \cos \varphi +%
\overline{\lambda }_{2}\cosh \phi \sin \varphi \right) ,  \notag \\
^{SS}q_{1}^{-} &=&\overline{\lambda }_{1}\partial _{-}\varphi +\overline{%
\lambda }_{2}\partial _{-}\phi -\mu \left( \lambda _{1}\cosh \phi \sin
\varphi +\lambda _{2}\sinh \phi \cos \varphi \right) ,  \notag \\
^{SS}q_{2}^{-} &=&-\overline{\lambda }_{1}\partial _{-}\phi +\overline{%
\lambda }_{2}\partial _{-}\varphi +\mu \left( \lambda _{1}\sinh \phi \cos
\varphi -\lambda _{2}\cosh \phi \sin \varphi \right) .  \notag
\end{eqnarray}

\subsubsection{The $(2,2)$ complex sine-Gordon and its hyperbolic
counterpart.}

The second models of interest are when $i=1$ in the "compact" and
"non-compact" cases, which are possibly related to two different truncations
of the Pohlmeyer reduced $AdS_{3}\times S^{3}$ superstring. For the compact
and non-compact models we choose, respectively, the following superfields
components%
\begin{eqnarray*}
\phi ^{1} &=&\ln \cos \varphi +i\theta \text{, \ \ \ \ \ \ \ \ }\phi ^{1}%
\text{ }=\text{ }\ln \cosh \phi +i\chi , \\
\psi _{\pm }^{1} &=&\tan \varphi e^{-i\theta }\chi _{\pm }^{1}\text{,\ \ \ \
\ \ \ }\psi _{\pm }^{1}\text{ }=\text{ }\tanh \phi e^{-i\chi }\rho _{\pm
}^{1}, \\
W(\phi ^{1}) &=&4\mu e^{\phi ^{1}},\text{ \ \ \ \ \ \ \ \ \ \ }W(\phi ^{1})%
\text{ }=\text{ }4\mu e^{\phi ^{1}},
\end{eqnarray*}%
where\footnote{%
The bar over the real fermion components should not be confused with the
complex conjugation which is denoted by the same symbol.} $\chi
_{+}^{1}=(\lambda _{1}+i\lambda _{2}),$ $\chi _{-}^{1}=(\overline{\lambda }%
_{1}+i\overline{\lambda }_{2}),$ $\rho _{+}^{1}=(\lambda _{3}+i\lambda _{4})$
and $\rho _{-}^{1}=(\overline{\lambda }_{3}+i\overline{\lambda }_{4}).$

The kahler potentials are chosen such that 
\begin{equation*}
g_{1\overline{1}}=-\frac{1}{1-\left\vert e^{-\phi ^{1}}\right\vert ^{2}}%
=\cot ^{2}\varphi ,\text{ \ \ \ \ \ }g_{1\overline{1}}=\frac{1}{1-\left\vert
e^{-\phi ^{1}}\right\vert ^{2}}=\coth ^{2}\phi ,
\end{equation*}%
which imply 
\begin{equation*}
\Gamma _{11}^{1}=\frac{1}{\sin ^{2}\varphi },\text{ \ \ \ \ \ }R_{1\overline{%
1}1\overline{1}}=\frac{\cot ^{4}\varphi }{\sin ^{2}\varphi }\text{ \ and \ }%
\Gamma _{11}^{1}=-\frac{1}{\sinh ^{2}\phi },\text{ \ \ \ \ \ }R_{1\overline{1%
}1\overline{1}}=-\frac{\coth ^{4}\phi }{\sinh ^{2}\phi },
\end{equation*}%
respectively$.$

In the real variables, the Lagrangian density (\ref{LG lagrangian}) is (see
also \cite{napolitano})%
\begin{eqnarray}
\frac{1}{4}\mathcal{L} &=&\partial _{+}\varphi \partial _{-}\varphi +\cot
^{2}\varphi \mathcal{\partial }_{+}\theta \mathcal{\partial }_{-}\theta +%
\frac{i}{2}(\lambda _{1}\partial _{-}\lambda _{1}+\lambda _{2}\partial
_{-}\lambda _{2}+\overline{\lambda }_{1}\partial _{+}\overline{\lambda }_{1}+%
\overline{\lambda }_{2}\partial _{+}\overline{\lambda }_{2})+  \notag \\
&&+i\cot ^{2}\varphi \left[ \mathcal{\partial }_{-}\theta \lambda
_{1}\lambda _{2}+\mathcal{\partial }_{+}\theta \overline{\lambda }_{1}%
\overline{\lambda }_{2}\right] +\frac{1}{\sin ^{2}\varphi }\lambda
_{1}\lambda _{2}\overline{\lambda }_{1}\overline{\lambda }_{2}-  \notag \\
&&-\mu ^{2}\sin ^{2}\varphi +\mu i\cos \varphi \left[ \cos \theta (\lambda
_{1}\overline{\lambda }_{2}+\lambda _{2}\overline{\lambda }_{1})-\sin \theta
(\lambda _{1}\overline{\lambda }_{1}-\lambda _{2}\overline{\lambda }_{2})%
\right] .  \label{N=2 complex sine-gordon}
\end{eqnarray}%
A similar lagrangian for $\mu =0$ was constructed in \cite{nakatsu} by using
a conventional (Kazama-Suzuki) gauged super WZNW model.

The supercharge densities can be expressed as $G_{\pm }^{0}=2(q_{1}^{\pm
}+iq_{2}^{\pm })$ in terms of the following $2+2$ real components%
\begin{eqnarray}
^{SS}q_{1}^{+} &=&\frac{1}{Y_{3}}\left( \lambda _{1}\partial
_{+}Y_{1}+\lambda _{2}\partial _{+}Y_{2}\right) -\mu \overline{\lambda }%
_{2}Y_{3},  \notag \\
^{SS}q_{2}^{+} &=&\frac{1}{Y_{3}}\left( -\lambda _{1}\partial
_{+}Y_{2}+\lambda _{2}\partial _{+}Y_{1}\right) -\mu \overline{\lambda }%
_{1}Y_{3},  \notag \\
^{SS}q_{1}^{-} &=&\frac{1}{Y_{3}}\left( \overline{\lambda }_{1}\partial
_{-}Y_{1}+\overline{\lambda }_{2}\partial _{-}Y_{2}\right) +\mu \lambda
_{2}Y_{3},  \notag \\
^{SS}q_{2}^{-} &=&\frac{1}{Y_{3}}\left( -\overline{\lambda }_{1}\partial
_{-}Y_{2}+\overline{\lambda }_{2}\partial _{-}Y_{1}\right) +\mu \lambda
_{1}Y_{3},  \label{SS sine-gordon}
\end{eqnarray}%
where $Y_{1}=\cos \theta \cos \varphi ,$ $Y_{2}=\sin \theta \cos \varphi $
and $Y_{3}=\sin \varphi .$

Similarly, for the non-compact model we have the Lagrangian density 
\begin{eqnarray}
\frac{1}{4}\mathcal{L} &=&\partial _{+}\phi \partial _{-}\phi +\coth
^{2}\phi \mathcal{\partial }_{+}\chi \mathcal{\partial }_{-}\chi +\frac{i}{2}%
(\lambda _{3}\partial _{-}\lambda _{3}+\lambda _{4}\partial _{-}\lambda _{4}+%
\overline{\lambda }_{3}\partial _{+}\overline{\lambda }_{3}+\overline{%
\lambda }_{4}\partial _{+}\overline{\lambda }_{4})+  \notag \\
&&+i\coth ^{2}\phi \left[ \mathcal{\partial }_{-}\chi \lambda _{3}\lambda
_{4}+\mathcal{\partial }_{+}\chi \overline{\lambda }_{3}\overline{\lambda }%
_{4}\right] +\frac{1}{\sinh ^{2}\phi }\lambda _{3}\lambda _{4}\overline{%
\lambda }_{3}\overline{\lambda }_{4}-  \notag \\
&&-\mu ^{2}\sinh ^{2}\phi -\mu i\cosh \phi \left[ \cos \chi (\lambda _{3}%
\overline{\lambda }_{4}+\lambda _{4}\overline{\lambda }_{3})-\sin \chi
(\lambda _{3}\overline{\lambda }_{3}-\lambda _{4}\overline{\lambda }_{4})%
\right]  \label{N=2 complex sinh-Gordon}
\end{eqnarray}%
and the supercharge densities $G_{\pm }^{0}=2(q_{1}^{\pm }+iq_{2}^{\pm })$
written in terms of the following $2+2$ real components%
\begin{eqnarray}
^{SS}q_{1}^{+} &=&\frac{1}{X_{3}}\left( \lambda _{3}\partial
_{+}X_{1}+\lambda _{4}\partial _{+}X_{2}\right) -\mu \overline{\lambda }%
_{4}X_{3},  \notag \\
^{SS}q_{2}^{+} &=&\frac{1}{X_{3}}\left( -\lambda _{3}\partial
_{+}X_{2}+\lambda _{4}\partial _{+}X_{1}\right) -\mu \overline{\lambda }%
_{3}X_{3},  \notag \\
^{SS}q_{1}^{-} &=&\frac{1}{X_{3}}\left( \overline{\lambda }_{3}\partial
_{-}X_{1}+\overline{\lambda }_{4}\partial _{-}X_{2}\right) +\mu \lambda
_{4}X_{3},  \notag \\
^{SS}q_{2}^{-} &=&\frac{1}{X_{3}}\left( -\overline{\lambda }_{3}\partial
_{-}X_{2}+\overline{\lambda }_{4}\partial _{-}X_{1}\right) +\mu \lambda
_{3}X_{3},  \label{SS sinh-gordon}
\end{eqnarray}%
where $X_{1}=\cos \chi \cosh \phi ,$ $X_{2}=\sin \chi \cosh \phi $ and $%
X_{3}=\sinh \phi .$

Now, we proceed to use our formulation.

\subsection{Supercharges of the Pohlmeyer reduced models.}

As mentioned above, the aim is to try to relate the superspace expressions (%
\ref{N=2 SS sine-Gordon supercharges}), (\ref{SS sine-gordon}) and (\ref{SS
sinh-gordon}) with the supersymmetry flow result (\ref{AKNS supercharges}).

\subsubsection{Reduction of the $AdS_{2}\times S^{2}$ superstring and $(2,2)$
2D SUSY.}

This is the only known case in which a reduced model posses 2D world-sheet
supersymmetry \cite{tseytlin}. However, the supersymmetry of this $%
AdS_{2}\times S^{2}$ reduced model was identified through its superspace
description (\ref{N=2 SS sine-Gordon lagrangian}), i.e the $N=2$
supersymmetric sine-Gordon model. Here we use our general flow approach to
confirm this fact from a different point of view.

From the general discussion we identify $\frac{F}{G}=\frac{PSU(1,1\mid 2)}{%
U(1)\times U(1)}$ and $\frac{G}{H}=G$ because $H=\oslash ,$ as can be seen
from (\ref{psu gauge group}). Then, this model has no gauge symmetries $%
\mathcal{K}_{B}^{(0)}$=$\oslash $ and from (\ref{psu anticom}) we expect to
obtain a reduced model with an ordinary extended $(2,2)$ supersymmetry.
Fortunately, in this case we do note have gauge fields and Wilson lines and
this means that the variations $\delta _{\pm 1/2}$ can be lifted easily to
the Lagrangian level (\ref{Abelian toda action}). Note also that we do not
have to deal with the constraints (\ref{new constraints}).

Using the basis (\ref{psu loop basis}), we parametrize the physical fields as%
\begin{equation*}
B=diag\left( B_{A},B_{S}\right) ,\text{ \ \ \ \ \ }\psi ^{\left( -1/2\right)
}=\psi _{i}G_{i}^{(-1/2)},\text{ \ \ \ \ \ }\psi ^{\left( +1/2\right) }=%
\overline{\psi }_{i}G_{i}^{(+1/2)},
\end{equation*}%
where $i=1,2$ and $B_{A},$ $B_{S}$ are given by\footnote{%
The $A,$ $S$ stands for $AdS_{n}$ and $S^{n}.$} 
\begin{eqnarray*}
B_{A} &=&\exp \left( \phi M_{1}^{(0)}\right) \text{ }=\text{ }%
\begin{pmatrix}
\cosh \phi & \sinh \phi \\ 
\sinh \phi & \cosh \phi%
\end{pmatrix}%
\in U(1), \\
B_{S} &=&\exp \left( \varphi M_{2}^{(0)}\right) \text{ }=\text{ }%
\begin{pmatrix}
\cos \varphi & i\sin \varphi \\ 
i\sin \varphi & \cos \varphi%
\end{pmatrix}%
\in U(1).
\end{eqnarray*}

With $\Lambda $ defined in (\ref{T for psu}), we can compute all the terms
entering the action (\ref{Abelian toda action}). They are 
\begin{eqnarray*}
\left\langle B^{-1}\partial _{+}BB^{-1}\partial _{-}B\right\rangle
&=&2\left( \partial _{+}\phi \partial _{-}\phi +\partial _{+}\varphi
\partial _{-}\varphi \right) , \\
\left\langle \psi _{+}^{\left( +1/2\right) }\partial _{-}\psi ^{\left(
-1/2\right) }+\psi _{-}^{\left( -1/2\right) }\partial _{+}\psi ^{\left(
+1/2\right) }\right\rangle &=&2\left( \psi _{i}\partial _{-}\psi _{i}+%
\overline{\psi }_{i}\partial _{+}\overline{\psi }_{i}\right) , \\
\mu ^{2}\left\langle \Lambda _{+}^{(+1)}B\Lambda
_{-}^{(-1)}B^{-1}\right\rangle &=&-\frac{\mu ^{2}}{2}\left( \cos 2\varphi
-\cosh 2\phi \right) , \\
\mu \left\langle \psi _{+}^{\left( +1/2\right) }B\psi _{-}^{\left(
-1/2\right) }B^{-1}\right\rangle &=&-2\mu \left\{ \cos \varphi \cosh \phi
\left( \psi _{1}\overline{\psi }_{1}+\psi _{2}\overline{\psi }_{2}\right)
+\sin \varphi \sinh \phi \left( \psi _{2}\overline{\psi }_{1}-\psi _{1}%
\overline{\psi }_{2}\right) \right\}
\end{eqnarray*}%
and the total Lagrangian density of the corresponding reduced model is 
\begin{eqnarray}
-\frac{2\pi }{k}\mathcal{L} &=&\partial _{+}\phi \partial _{-}\phi +\partial
_{+}\varphi \partial _{-}\varphi +\psi _{1}\partial _{-}\psi _{1}+\psi
_{2}\partial _{-}\psi _{2}+\overline{\psi }_{1}\partial _{+}\overline{\psi }%
_{1}+\overline{\psi }_{2}\partial _{+}\overline{\psi }_{2}-V
\label{N=2 sine-Gordon} \\
V &=&\mu ^{2}\left\langle \Lambda _{+}^{(+1)}B\Lambda
_{-}^{(-1)}B^{-1}\right\rangle +\mu \left\langle \psi _{+}^{\left(
+1/2\right) }B\psi _{-}^{\left( -1/2\right) }B^{-1}\right\rangle .  \notag
\end{eqnarray}

Now comes the interesting part. To compute the supercharges associated to (%
\ref{N=2 sine-Gordon}) we use the general formula (\ref{AKNS supercharges})
in the gauge $A_{\pm }=0$ because $\mathcal{K}_{B}^{(0)}=\oslash ,$ which
reduce to (\ref{mKdV supercharges}) in the mKdV hierarchy. Writing the
supercharges in the form $Q(\delta _{\pm 1/2})=\dint_{-\infty }^{+\infty
}dx^{1}G(\delta _{\pm 1/2}),$ we find the densities $G(\delta _{\pm
1/2})=-q_{i}^{\pm }F_{i}^{(\mp 1/2)}$ in terms of the following $2+2$ real
components\footnote{%
SF means Supersymmetry Flows.} 
\begin{eqnarray}
^{SF}q_{1}^{+} &=&\psi _{1}\partial _{+}\varphi +\psi _{2}\partial _{+}\phi
+\mu \left( \overline{\psi }_{1}\cosh \phi \sin \varphi +\overline{\psi }%
_{2}\sinh \phi \cos \varphi \right) ,
\label{SF N=2 sine-Gordon supercharges} \\
^{SF}q_{2}^{+} &=&\psi _{1}\partial _{+}\phi -\psi _{2}\partial _{+}\varphi
+\mu \left( \overline{\psi }_{1}\sinh \phi \cos \varphi -\overline{\psi }%
_{2}\cosh \phi \sin \varphi \right) ,  \notag \\
^{SF}q_{1}^{-} &=&-\overline{\psi }_{1}\partial _{-}\varphi +\overline{\psi }%
_{2}\partial _{-}\phi +\mu \left( \psi _{1}\cosh \phi \sin \varphi -\psi
_{2}\sinh \phi \cos \varphi \right) ,  \notag \\
^{SF}q_{2}^{-} &=&\overline{\psi }_{1}\partial _{-}\phi +\overline{\psi }%
_{2}\partial _{-}\varphi -\mu \left( \psi _{1}\sinh \phi \cos \varphi +\psi
_{2}\cosh \phi \sin \varphi \right)  \notag
\end{eqnarray}%
and this is because we have $\dim \mathcal{K}_{F}^{(\pm 1/2)}=2,$ as shown
in (\ref{table}), (\ref{psu susy flows I}), (\ref{psu susy flows II}).
Compare (\ref{SF N=2 sine-Gordon supercharges}) with the superspace result (%
\ref{N=2 SS sine-Gordon supercharges}).

Inserting $D^{(+1/2)}=\epsilon _{i}F_{i}^{(+1/2)},D^{(-1/2)}=\overline{%
\epsilon }_{i}F_{i}^{(-1/2)}$ in the supersymmetry variations (\ref{poisson
susy 1/2}), (\ref{poisson susy -1/2}) and using (\ref{psu anticom}) we
obtain the $(2,2)$ supersymmetry algebra 
\begin{equation}
\left[ \delta _{+1/2},\delta _{+1/2}^{\prime }\right] =2\left( \epsilon
_{1}\epsilon _{1}^{\prime }+\epsilon _{2}\epsilon _{2}^{\prime }\right)
\partial _{+},\text{ \ \ \ }\left[ \delta _{-1/2},\delta _{-1/2}^{\prime }%
\right] =2\left( \overline{\epsilon }_{1}\overline{\epsilon }_{1}^{\prime }+%
\overline{\epsilon }_{2}\overline{\epsilon }_{2}^{\prime }\right) \partial
_{-},\text{ \ \ \ }\left[ \delta _{+1/2},\delta _{-1/2}\right] =0  \notag
\end{equation}%
in agreement with the result (\ref{un-deformed}), (\ref{deformed}) in the
absence of gauge group. These supersymmetry transformations $\delta _{\pm
1/2}$ are the same as the ones induced by the supercharges (\ref{N=2 SS
sine-Gordon supercharges}). See also \cite{susyflows-mKdV} for similar
models constructed from the twisted superalgebras $sl(2\mid 1)^{(2)}$ and $%
psl(2\mid 2)^{(2)}$ .

\subsubsection{Reduction of the $AdS_{3}\times S^{3}$ superstring and
possible $(4,4)$ 2D SUSY.}

This model is more complicated because it has gauge symmetries and it is the
first non-trivial case in which we want to test our construction, then we
will study it in some detail$.$ Although the Pohlmeyer reduced Lagrangian
for the $AdS_{3}\times S^{3}$ superstring was already computed in \cite%
{tseytlin II}, the existence of 2D world-sheet supersymmetry was conjectured
to be of type $N=(2,2)$. \ Here we provide some evidence that the 2D
supersymmetry is of the extended type $N=(4,4)$ instead of $N=(2,2).$

From the general discussion we identify $\frac{F}{G}=\frac{PSU(1,1\mid
2)\times PSU(1,1\mid 2)}{SU(1,1)\times SU(2)}$ and $\frac{G}{H}=\frac{%
SU(1,1)\times SU(2)}{U(1)\times U(1)},$ then this model has gauge group $%
H=U(1)\times U(1)$ as shown in (\ref{gauge group for psuxpsu})$.$ We will
concentrate on the vector gauge only, i.e $\epsilon _{L}=\epsilon _{R}=I$ as
the axial gauge i.e $\epsilon _{L}=I,$ $\epsilon _{R}=-I$ follows exactly
the same lines.

Using the basis (\ref{psuxpsu loop basis}), we parametrize the physical
fields as follows%
\begin{equation*}
B=diag\left( B_{A},B_{S}\right) ,\text{ \ \ \ \ \ }\psi ^{\left( -1/2\right)
}=\psi _{i}G_{i}^{(-1/2)},\text{ \ \ \ \ \ }\psi ^{\left( +1/2\right) }=%
\overline{\psi }_{i}G_{i}^{(+1/2)},
\end{equation*}%
where $i=1,...,4$ and $B_{A}\in SU(1,1)$, $B_{S}\in SU(2).$ The group
elements $B_{A}$ and $B_{S}$ are given by%
\begin{eqnarray*}
B_{A} &=&\exp \left( \frac{1}{2}(\chi +t)K_{1}^{(0)}\right) \exp \left( \phi
M_{1}^{(0)}\right) \exp \left( \frac{1}{2}(\chi -t)K_{1}^{(0)}\right) \in
SU(1,1), \\
B_{S} &=&\exp \left( \frac{1}{2}(\theta +t^{\prime })K_{2}^{(0)}\right) \exp
\left( \varphi M_{4}^{(0)}\right) \exp \left( \frac{1}{2}(\theta -t^{\prime
})K_{2}^{(0)}\right) \in SU(2).
\end{eqnarray*}

The gauge transformations (\ref{action gauge transf.}) acting on $B$ are
simply the shifts $(t,t^{\prime })\rightarrow (t+\alpha ,t^{\prime }+\beta )$
and we fix the gauge by taking $t=t^{\prime }=0$ to get, in terms of $%
2\times 2$ block matrices, the coset elements%
\begin{eqnarray}
B_{A}^{\prime } &=&%
\begin{pmatrix}
X_{1}+iX_{2} & X_{3} \\ 
X_{3} & X_{1}-iX_{2}%
\end{pmatrix}%
\in \frac{SU(1,1)}{U(1)},  \label{non compact gauged} \\
\text{ \ \ }B_{S}^{\prime } &=&%
\begin{pmatrix}
Y_{1}+iY_{2} & iY_{3} \\ 
iY_{3} & Y_{1}-iY_{2}%
\end{pmatrix}%
\text{\ }\in \frac{SU(2)}{U(1)},  \label{compact gauged}
\end{eqnarray}%
where 
\begin{eqnarray}
X_{1} &=&\cos \chi \cosh \phi ,\text{ \ }X_{2}\text{ }=\text{ }\sin \chi
\cosh \phi ,\text{ \ }X_{3}\text{ }=\text{ }\sqrt{-1+X_{1}^{2}+X_{2}^{2}}%
\text{ }=\text{ }\sinh \phi ,  \label{X y Y} \\
Y_{1} &=&\cos \theta \cos \varphi ,\text{ \ \ \ \ }Y_{2}\text{ }=\text{ }%
\sin \theta \cos \varphi ,\text{ \ \ \ }Y_{3}\text{ }=\text{ }\sqrt{%
1-Y_{1}^{2}-Y_{2}^{2}}\text{ }=\text{ }\sin \varphi .  \notag
\end{eqnarray}

The other necessary ingredient is the constant element $\Lambda $ defined in
(\ref{T for psuxpsu}) and the gauge fields $A_{\pm }=a_{\pm
}K_{1}^{(0)}+b_{\pm }K_{2}^{(0)}.$ Using the gauge field $A_{\pm }$
equations of motion given by (\ref{arbitrary variation}) we find%
\begin{equation}
a_{\pm }=\mp \frac{1}{2}\left( \coth ^{2}\phi \partial _{\pm }\chi \pm \frac{%
1}{\sinh ^{2}\phi }F_{\pm }\right) ,\text{ \ \ \ \ \ }b_{\pm }=\pm \frac{1}{2%
}\left( \cot ^{2}\varphi \partial _{\pm }\theta \mp \frac{1}{\sin
^{2}\varphi }F_{\pm }\right) ,  \label{gauge field components}
\end{equation}%
where $F_{+}=(\psi _{1}\psi _{2}-\psi _{3}\psi _{4}),$ $F_{-}=(\overline{%
\psi }_{1}\overline{\psi }_{2}-\overline{\psi }_{3}\overline{\psi }_{4})$
and where we have used (\ref{currents}) and $Q_{\pm }^{(0)}=F_{\pm
}(K_{1}^{(0)}-K_{2}^{(0)})$.

Once the $A_{\pm }$ are solved through their equations of motion (\ref{gauge
field components}) we put them back in the gauged fixed Lagrangian obtaining
the Pohlmeyer reduced action functional. However, instead of doing this we
will integrate them out in the path integral without taking account of the
quantum measure Jacobian, which gives the same classical answer. To do this,
it is useful to consider the general integration formula%
\begin{equation}
\dint DAD\overline{A}\exp \left[ -\frac{k}{2\pi }\dint \left( \overline{A}MA+%
\overline{A}N+\overline{N}A\right) \right] =\exp \left[ -\frac{k}{2\pi }%
\dint \left( -\overline{N}M^{-1}N\right) \right] ,
\label{integration formula}
\end{equation}%
where $A,\overline{A},N$ and $\overline{N}$ are vectors and $M$ is an
invertible matrix.

The gauge field $A_{\pm }$ independent quantities entering the SSSSG model
action (\ref{NA Toda action}) are%
\begin{eqnarray}
\left\langle B^{-1}\partial _{+}BB^{-1}\partial _{-}B\right\rangle
&=&2\left( \partial _{+}\phi \partial _{-}\phi -\cosh ^{2}\phi \partial
_{+}\chi \partial _{-}\chi +\partial _{+}\varphi \partial _{-}\varphi +\cos
^{2}\varphi \partial _{+}\theta \partial _{-}\theta \right) ,
\label{quantities in the total action} \\
\left\langle \psi _{+}^{\left( +1/2\right) }\partial _{-}\psi ^{\left(
-1/2\right) }+\psi _{-}^{\left( -1/2\right) }\partial _{+}\psi ^{\left(
+1/2\right) }\right\rangle &=&2\left( \psi _{i}\partial _{-}\psi _{i}+%
\overline{\psi }_{i}\partial _{+}\overline{\psi }_{i}\right) ,  \notag \\
\mu ^{2}\left\langle \Lambda _{+}^{(+1)}B\Lambda
_{-}^{(-1)}B^{-1}\right\rangle &=&-\frac{1}{2}\mu ^{2}\left( \cos 2\varphi
-\cosh 2\phi \right) ,  \notag \\
\mu \left\langle \psi _{+}^{\left( +1/2\right) }B\psi _{-}^{\left(
-1/2\right) }B^{-1}\right\rangle &=&2\mu \left\{ 
\begin{array}{c}
\sin \varphi \sinh \phi \left( \psi _{1}\overline{\psi }_{3}-\psi _{3}%
\overline{\psi }_{1}+\psi _{2}\overline{\psi }_{4}-\psi _{4}\overline{\psi }%
_{2}\right) + \\ 
+\cos \varphi \cosh \phi \left[ 
\begin{array}{c}
\cos (\theta +\chi )\left( \psi _{1}\overline{\psi }_{2}-\psi _{2}\overline{%
\psi }_{1}-\psi _{3}\overline{\psi }_{4}+\psi _{4}\overline{\psi }%
_{3}\right) + \\ 
+\sin (\theta +\chi )\left( \psi _{1}\overline{\psi }_{1}+\psi _{2}\overline{%
\psi }_{2}+\psi _{3}\overline{\psi }_{3}+\psi _{4}\overline{\psi }_{4}\right)%
\end{array}%
\right]%
\end{array}%
\right\} ,  \notag
\end{eqnarray}%
where we have used (\ref{currents}) in the first line and (\ref{traces}) in
the last line.

The gauge field $A_{\pm }$ dependent part of the action is 
\begin{eqnarray}
I^{V} &=&\left\langle -A_{-}\left( \partial _{+}BB^{-1}+Q_{+}^{(0)}\right)
+A_{+}\left( B^{-1}\partial _{-}B+Q_{-}^{(0)}\right)
-A_{-}BA_{+}B^{-1}+A_{+}A_{-}\right\rangle  \notag \\
&=&a_{+}M_{A}a_{-}+a_{+}N_{A}+\overline{N}%
_{A}a_{-}+b_{+}M_{S}b_{-}+b_{+}N_{S}+\overline{N}_{S}b_{-},
\label{gauge dependent part}
\end{eqnarray}%
where $M_{A}=4\sinh ^{2}\phi ,$ $N_{A}=-2X_{-},$ $\overline{N}_{A}=2X_{+},$ $%
M_{S}=4\sin ^{2}\varphi ,$ $N_{S}=2Y_{-},\overline{N}_{S}=-2Y_{+}$ and%
\begin{eqnarray*}
X_{+} &=&\partial _{+}\chi \cosh ^{2}\phi +F_{+},\text{ \ \ \ \ \ }X_{-}%
\text{ }=\text{ }\partial _{-}\chi \cosh ^{2}\phi -F_{-}, \\
Y_{+} &=&\partial _{+}\theta \cos ^{2}\varphi -F_{+},\text{ \ \ \ \ \ \ \ \ }%
Y_{-}\text{ }=\text{ }\partial _{+}\theta \cos ^{2}\varphi +F_{-}.
\end{eqnarray*}

Before performing the full $a_{\pm },b_{\pm }$ integration, by using (\ref%
{integration formula}), we will consider first the following two consistent
truncations of the total model which are defined by%
\begin{eqnarray*}
I &:&\text{ \ }B_{S}=Id,\text{ }\psi _{1}=\psi _{2}=\overline{\psi }_{1}=%
\overline{\psi }_{2}=0\text{ \ and \ }b_{\pm }=0, \\
II &:&\text{ \ }B_{A}=Id,\text{ }\psi _{3}=\psi _{4}=\overline{\psi }_{3}=%
\overline{\psi }_{4}=0\text{ \ and \ }a_{\pm }=0.
\end{eqnarray*}

One comment is in order. The gauges $b_{\pm }=0$ or $a_{\pm }=0$ are only
valid on-shell so they does not make any sense at the Lagrangian level. What
we are doing is localizing in (\ref{Abelian toda action}) only one part of
the gauge symmetry (\ref{kac moody symmetry}) while keeping the other part
intact, i.e global. This is equivalent to the vanishing of some components
of $A_{\pm }$ in the action (\ref{NA Toda action}). Then, in the limit $I$
we get from (\ref{gauge dependent part}) and (\ref{integration formula}) 
\begin{equation*}
I_{I}^{V}=-\overline{N}_{A}M_{A}^{-1}N_{A}=\frac{X_{+}X_{-}}{\sinh ^{2}\phi }%
.
\end{equation*}%
Putting all together with (\ref{quantities in the total action}) in this
particular limit, we get the Lagrangian density 
\begin{eqnarray}
-\frac{2\pi }{k}\mathcal{L}_{I} &=&\partial _{+}\phi \partial _{-}\phi
+\coth ^{2}\phi \mathcal{\partial }_{+}\chi \mathcal{\partial }_{-}\chi
+\psi _{3}\partial _{-}\psi _{3}+\psi _{4}\partial _{-}\psi _{4}+\overline{%
\psi }_{3}\partial _{+}\overline{\psi }_{3}+\overline{\psi }_{3}\partial _{+}%
\overline{\psi }_{4}-  \notag \\
&&-\coth ^{2}\phi \left[ \mathcal{\partial }_{-}\chi \psi _{3}\psi _{4}-%
\mathcal{\partial }_{+}\chi \overline{\psi }_{3}\overline{\psi }_{4}\right] -%
\frac{1}{\sinh ^{2}\phi }\psi _{3}\psi _{4}\overline{\psi }_{3}\overline{%
\psi }_{4}-  \notag \\
&&-\mu ^{2}\sinh ^{2}\phi -2\mu \cosh \phi \left[ \cos \chi (-\psi _{3}%
\overline{\psi }_{4}+\psi _{4}\overline{\psi }_{3})+\sin \chi (\psi _{3}%
\overline{\psi }_{3}+\psi _{4}\overline{\psi }_{4})\right] ,
\label{N=2 sinh-Gordon prime}
\end{eqnarray}%
which should be compared with (\ref{N=2 complex sinh-Gordon}).

After the $a_{\pm }$ gauge field integration there a residual $U(1)$ global
symmetry which we now proceed to identify. The lagrangian (\ref{N=2
sinh-Gordon prime}) is separately invariant under the following global
transformations%
\begin{equation}
\widetilde{\chi }=\chi +\alpha _{1},\text{ \ \ \ \ \ }\widetilde{\lambda }%
_{\pm }=e^{\pm i\beta _{1}}\lambda _{\pm },\text{ \ }  \label{U(1)}
\end{equation}%
where $\alpha _{1}=2\pi n,$ $\beta _{1}\in 
\mathbb{R}
,$ $\lambda _{+}=\psi _{3}+i\psi _{4}$ and $\lambda _{-}=\overline{\psi }%
_{3}-i\overline{\psi }_{4}.$ The Noether procedure gives the corresponding
charges 
\begin{equation*}
Q_{\chi }=\dint_{-\infty }^{+\infty }dx^{1}\coth ^{2}\phi \left\{ \left( 
\mathcal{\partial }_{+}\chi +\mathcal{\partial }_{-}\chi \right) -\left(
\psi _{3}\psi _{4}-\overline{\psi }_{3}\overline{\psi }_{4}\right) \right\} ,%
\text{ \ \ \ }Q_{\lambda }=\dint_{-\infty }^{+\infty }dx^{1}\left( \psi
_{3}\psi _{4}-\overline{\psi }_{3}\overline{\psi }_{4}\right) ,
\end{equation*}%
where $Q_{\chi }$ is associated to the isometry of the metric and $%
Q_{\lambda }$ to the fermion electric charge.

Let us extract the information encoded in the conservation law given by the
zero grade equations of the Drinfeld-Sokolov procedure in this Vector gauge.
From (\ref{abelian charges}) we have an Abelian charge%
\begin{equation*}
Q_{U(1)}=\frac{1}{2}\dint_{-\infty }^{+\infty }dx^{1}\left(
a_{+}-a_{-}\right) .
\end{equation*}%
Now, taking the gauge field components $a_{\pm }$ defined by (\ref{gauge
field components}) in this particular limit and using $Q_{\chi },$ $%
Q_{\lambda },$ we find that $Q_{U(1)}=-\left( Q_{\chi }+Q_{\lambda }\right)
. $ Then, the grade zero equations encode the conservations laws associated
to the global symmetries of the reduced gauge fixed action.

To compute the supercharges associated to (\ref{N=2 sinh-Gordon prime}), i.e
the grade $\pm 1/2$ DS equations, we use the general formula (\ref{AKNS
supercharges}) in the gauge (\ref{gauge field components}) and neglect the
Wilson line conjugation. Writing $Q(\delta _{\pm 1/2})=\dint_{-\infty
}^{+\infty }dx^{1}G(\delta _{\pm 1/2})$ and using (\ref{conju}), (\ref%
{conjugations}) we can write the fermionic current densities $G(\delta _{\pm
1/2})=q_{i}^{\pm }F_{i}^{(\mp 1/2)}$ in terms of the following $2+2$ real
components%
\begin{eqnarray}
^{SF}q_{1}^{+} &=&\frac{1}{X_{3}}\left( -\psi _{3}\partial _{+}X_{1}+\psi
_{4}\partial _{+}X_{2}\right) -\mu \overline{\psi }_{4}X_{3},
\label{non-compact supercharges} \\
^{SF}q_{2}^{+} &=&\frac{1}{X_{3}}\left( -\psi _{3}\partial _{+}X_{2}-\psi
_{4}\partial _{+}X_{1}\right) +\mu \overline{\psi }_{3}X_{3},  \notag \\
^{SF}q_{1}^{-} &=&\frac{1}{X_{3}}\left( -\overline{\psi }_{3}\partial
_{-}X_{1}-\overline{\psi }_{4}\partial _{-}X_{2}\right) -\mu \psi _{4}X_{3},
\notag \\
^{SF}q_{2}^{-} &=&\frac{1}{X_{3}}\left( \overline{\psi }_{3}\partial
_{-}X_{2}-\overline{\psi }_{4}\partial _{-}X_{1}\right) +\mu \psi _{3}X_{3}.
\notag
\end{eqnarray}%
which have to be compared with the superspace expressions (\ref{SS
sinh-gordon}).

Similarly, in the limit $II$ we get 
\begin{equation*}
I_{II}^{V}=-\overline{N}_{S}M_{S}^{-1}N_{S}=\frac{Y_{+}Y_{-}}{\sin
^{2}\varphi }
\end{equation*}%
and the Lagrangian density 
\begin{eqnarray}
-\frac{2\pi }{k}\mathcal{L}_{II} &=&\partial _{+}\varphi \partial
_{-}\varphi +\cot ^{2}\varphi \partial _{+}\theta \partial _{-}\theta +\psi
_{1}\partial _{-}\psi _{1}+\psi _{2}\partial _{-}\psi _{2}+\overline{\psi }%
_{1}\partial _{+}\overline{\psi }_{1}+\overline{\psi }_{2}\partial _{+}%
\overline{\psi }_{2}-  \notag \\
&&-\cot ^{2}\varphi \left[ \mathcal{\partial }_{-}\theta \psi _{1}\psi _{2}-%
\mathcal{\partial }_{+}\theta \overline{\psi }_{1}\overline{\psi }_{2}\right]
-\frac{1}{\sin ^{2}\varphi }\psi _{1}\psi _{2}\overline{\psi }_{1}\overline{%
\psi }_{2}-  \notag \\
&&-\mu ^{2}\sin ^{2}\varphi -2\mu \cos \varphi \left[ \cos \theta (\psi _{1}%
\overline{\psi }_{2}-\psi _{2}\overline{\psi }_{1})+\sin \theta (\psi _{1}%
\overline{\psi }_{1}+\psi _{2}\overline{\psi }_{2})\right] ,
\label{N=2 sine-Gordon prime}
\end{eqnarray}%
which should be compared with (\ref{N=2 complex sine-gordon}). \ Similarly,
the $U(1)$ global symmetry of (\ref{N=2 sine-Gordon prime}) is%
\begin{equation}
\widetilde{\theta }=\theta +\alpha _{2},\text{ \ \ \ \ \ }\widetilde{\rho }%
_{\pm }=e^{\pm i\beta _{2}}\rho _{\pm },\text{ \ }  \label{U(1) II}
\end{equation}%
where $\alpha _{2}=2\pi n,\beta _{2}\in 
\mathbb{R}
,\rho _{+}=\psi _{1}+i\psi _{2}$ and $\rho _{-}=\overline{\psi }_{1}-i%
\overline{\psi }_{2}.$ The Noether charges are%
\begin{equation*}
Q_{\theta }=\dint_{-\infty }^{+\infty }dx^{1}\cot ^{2}\varphi \left\{ \left( 
\mathcal{\partial }_{+}\theta +\mathcal{\partial }_{-}\theta \right) -\left(
\psi _{1}\psi _{2}-\overline{\psi }_{1}\overline{\psi }_{2}\right) \right\} ,%
\text{ \ \ \ }Q_{\rho }=\dint_{-\infty }^{+\infty }dx^{1}\left( \psi
_{1}\psi _{2}-\overline{\psi }_{1}\overline{\psi }_{2}\right)
\end{equation*}%
and as above, we have from (\ref{abelian charges}) an Abelian charge given
by 
\begin{equation*}
Q_{U(1)}=\frac{1}{2}\dint_{-\infty }^{+\infty }dx^{1}\left(
b_{+}-b_{-}\right) .
\end{equation*}%
Taking $b_{\pm }$ defined by (\ref{gauge field components}) in this
particular limit and $Q_{\theta },$ $Q_{\rho },$ we obtain the relation $%
Q_{U(1)}=+\left( Q_{\theta }-Q_{\rho }\right) .$

The superdensities can be written again as $G(\delta _{\pm 1/2})=q_{i}^{\pm
}F_{i}^{(\mp 1/2)}$ in terms of $2+2$ real components%
\begin{eqnarray}
^{SF}q_{1}^{+} &=&\frac{1}{Y_{3}}\left( -\psi _{1}\partial _{+}Y_{2}+\psi
_{2}\partial _{+}Y_{1}\right) -\mu \overline{\psi }_{1}Y_{3},
\label{compact supercharges} \\
^{SF}q_{2}^{+} &=&\frac{1}{Y_{3}}\left( -\psi _{1}\partial _{+}Y_{1}-\psi
_{2}\partial _{+}Y_{2}\right) -\mu \overline{\psi }_{2}Y_{3},  \notag \\
^{SF}q_{1}^{-} &=&\frac{1}{Y_{3}}\left( \overline{\psi }_{1}\partial
_{-}Y_{2}+\overline{\psi }_{2}\partial _{-}Y_{1}\right) -\mu \psi _{1}Y_{3},
\notag \\
^{SF}q_{2}^{-} &=&\frac{1}{Y_{3}}\left( -\overline{\psi }_{1}\partial
_{-}Y_{1}+\overline{\psi }_{2}\partial _{-}Y_{2}\right) -\mu \psi _{2}Y_{3},
\notag
\end{eqnarray}%
which should be compared with the superspace expressions (\ref{SS
sine-gordon}).

Now we perform the total $a_{\pm },b_{\pm }$ integration. This gives for (%
\ref{gauge dependent part}) that%
\begin{equation*}
I^{V}=\frac{X_{+}X_{-}}{\sinh ^{2}\phi }+\frac{Y_{+}Y_{-}}{\sin ^{2}\varphi }%
.
\end{equation*}%
Putting all together we get the full Pohlmeyer reduced $AdS_{3}\times S^{3}$
superstring sigma model action (see \cite{tseytlin II})%
\begin{eqnarray}
-\frac{2\pi }{k}\mathcal{L} &=&\partial _{+}\phi \partial _{-}\phi +\coth
^{2}\phi \mathcal{\partial }_{+}\chi \mathcal{\partial }_{-}\chi +\mathcal{%
\partial }_{+}\varphi \mathcal{\partial }_{-}\varphi +\cot ^{2}\varphi 
\mathcal{\partial }_{+}\theta \mathcal{\partial }_{-}\theta +\left( \psi
_{i}\partial _{-}\psi _{i}+\overline{\psi }_{i}\partial _{+}\overline{\psi }%
_{i}\right) -  \label{final form ads3xs3} \\
&&-\coth ^{2}\phi \left[ \partial _{+}\chi F_{-}-\partial _{-}\chi F_{+}%
\right] +\cot ^{2}\varphi \left[ \partial _{+}\theta F_{-}-\partial
_{-}\theta F_{+}\right] -F_{+}F_{-}\left( \frac{1}{\sinh ^{2}\phi }+\frac{1}{%
\sin ^{2}\varphi }\right) -V_{vec}  \notag \\
V_{vec} &=&\mu ^{2}\left\langle E_{+}^{(+1)}BE_{-}^{(-1)}B^{-1}\right\rangle
+\mu \left\langle \psi _{+}^{\left( +1/2\right) }B\psi _{-}^{\left(
-1/2\right) }B^{-1}\right\rangle ,  \notag
\end{eqnarray}%
where $i=1,...,4$ and the quantities in $V_{vec}$ are given in (\ref%
{quantities in the total action}). Note that in the final form (\ref{final
form ads3xs3}), there is no way to take any field limit leading to (\ref{N=2
sinh-Gordon prime}) and (\ref{N=2 sine-Gordon prime}) and we see how the two
sub-models get coupled in a non-trivial way.

The Lagrangian (\ref{final form ads3xs3}) is separately invariant under (\ref%
{U(1)}), (\ref{U(1) II}) and it is also invariant under 
\begin{equation*}
\widetilde{\chi }=\chi +\alpha _{1},\text{ \ \ \ \ }\widetilde{\theta }%
=\theta +\alpha _{2},\text{ \ \ \ \ }\widetilde{\lambda }_{\pm }=e^{\pm
i\beta }\lambda _{\pm },\text{ \ \ \ \ }\widetilde{\rho }_{\pm }=e^{\mp
i\beta }\rho _{\pm },
\end{equation*}%
where we have set $\beta _{2}=-\beta _{1}$, $\beta _{1}=\beta $ $.$ The
Noether charges are%
\begin{eqnarray*}
Q_{\chi } &=&\dint_{-\infty }^{+\infty }dx^{1}\coth ^{2}\phi \left\{ \left( 
\mathcal{\partial }_{+}\chi +\mathcal{\partial }_{-}\chi \right) +\left(
F_{+}-F_{-}\right) \right\} , \\
Q_{\theta } &=&\dint_{-\infty }^{+\infty }dx^{1}\cot ^{2}\varphi \left\{
\left( \mathcal{\partial }_{+}\theta +\mathcal{\partial }_{-}\theta \right)
-\left( F_{+}-F_{-}\right) \right\} , \\
Q_{\rho ,\lambda } &=&\dint_{-\infty }^{+\infty }dx^{1}\left(
F_{+}-F_{-}\right) .
\end{eqnarray*}%
From (\ref{abelian charges}) we have the residual $U(1)\times U(1)$ global
symmetry with conserved charges%
\begin{equation*}
Q_{a}=\frac{1}{2}\dint_{-\infty }^{+\infty }(a_{+}-a_{-}),\text{ \ \ \ \ \ }%
Q_{b}=\frac{1}{2}\dint_{-\infty }^{+\infty }(b_{+}-b_{-})
\end{equation*}%
and with $a_{\pm },b_{\pm }$ defined by (\ref{gauge field components}), we
find the relations 
\begin{equation*}
Q_{a}=-\left( Q_{\chi }-Q_{\rho ,\lambda }\right) ,\ \text{ \ \ \ \ }%
Q_{b}=+\left( Q_{\theta }-Q_{\rho ,\lambda }\right) .
\end{equation*}

The eight fermionic densities, because we have $\dim \mathcal{K}_{F}^{(\pm
1/2)}=4$ as can be seen from (\ref{table}), (\ref{psuxpsu susy flows I}), (%
\ref{psuxpsu susy flows II}), associated to the reduced model are computed
from\footnote{%
Recall that we are not writing the Wilson lines.} (\ref{AKNS supercharges}).
They can be written as $G(\delta _{\pm 1/2})=q_{i}^{\pm }F_{i}^{(\mp 1/2)},$
where 
\begin{eqnarray*}
^{SF}q_{1}^{+} &=&\frac{1}{X_{3}}\left[ -\psi _{3}\left( \partial
_{+}X_{1}-X_{2}F_{+}\right) +\psi _{4}\left( \partial
_{+}X_{2}+X_{1}F_{+}\right) \right] +\frac{1}{Y_{3}}\left[ -\psi _{1}\left(
\partial _{+}Y_{2}-Y_{1}F_{+}\right) +\psi _{2}\left( \partial
_{+}Y_{1}+Y_{2}F_{+}\right) \right] + \\
&&+\mu \left( -\overline{\psi }_{1}X_{1}Y_{3}+\overline{\psi }_{2}X_{2}Y_{3}+%
\overline{\psi }_{3}X_{3}Y_{2}-\overline{\psi }_{4}X_{3}Y_{1}\right) , \\
^{SF}q_{2}^{+} &=&\frac{1}{X_{3}}\left[ -\psi _{3}\left( \partial
_{+}X_{2}+X_{1}F_{+}\right) -\psi _{4}\left( \partial
_{+}X_{1}-X_{2}F_{+}\right) \right] +\frac{1}{Y_{3}}\left[ -\psi _{1}\left(
\partial _{+}Y_{1}+Y_{2}F_{+}\right) -\psi _{2}\left( \partial
_{+}Y_{2}-Y_{1}F_{+}\right) \right] + \\
&&+\mu \left( -\overline{\psi }_{1}X_{2}Y_{3}-\overline{\psi }_{2}X_{1}Y_{3}+%
\overline{\psi }_{3}X_{3}Y_{1}+\overline{\psi }_{4}X_{3}Y_{2}\right) , \\
^{SF}q_{3}^{+} &=&\frac{1}{X_{3}}\left[ -\psi _{1}\left( \partial
_{+}X_{1}-X_{2}F_{+}\right) -\psi _{2}\left( \partial
_{+}X_{2}+X_{1}F_{+}\right) \right] +\frac{1}{Y_{3}}\left[ \psi _{3}\left(
\partial _{+}Y_{2}-Y_{1}F_{+}\right) +\psi _{4}\left( \partial
_{+}Y_{1}+Y_{2}F_{+}\right) \right] + \\
&&+\mu \left( \overline{\psi }_{1}X_{3}Y_{2}+\overline{\psi }_{2}X_{3}Y_{1}+%
\overline{\psi }_{3}X_{1}Y_{3}+\overline{\psi }_{4}X_{2}Y_{3}\right) , \\
^{SF}q_{4}^{+} &=&\frac{1}{X_{3}}\left[ \psi _{1}\left( \partial
_{+}X_{2}+X_{1}F_{+}\right) -\psi _{2}\left( \partial
_{+}X_{1}-X_{2}F_{+}\right) \right] +\frac{1}{Y_{3}}\left[ -\psi _{3}\left(
\partial _{+}Y_{1}+Y_{2}F_{+}\right) +\psi _{4}\left( \partial
_{+}Y_{2}-Y_{1}F_{+}\right) \right] + \\
&&+\mu \left( -\overline{\psi }_{1}X_{3}Y_{1}+\overline{\psi }_{2}X_{3}Y_{2}-%
\overline{\psi }_{3}X_{2}Y_{3}+\overline{\psi }_{4}X_{1}Y_{3}\right)
\end{eqnarray*}%
and%
\begin{eqnarray*}
^{SF}q_{1}^{-} &=&\frac{1}{X_{3}}\left[ -\overline{\psi }_{3}\left( \partial
_{-}X_{1}+X_{2}F_{-}\right) -\overline{\psi }_{4}\left( \partial
_{-}X_{2}-X_{1}F_{-}\right) \right] +\frac{1}{Y_{3}}\left[ \overline{\psi }%
_{1}\left( \partial _{-}Y_{2}+Y_{1}F_{-}\right) +\overline{\psi }_{2}\left(
\partial _{-}Y_{1}-Y_{2}F_{-}\right) \right] + \\
&&+\mu \left( -\psi _{1}X_{1}Y_{3}-\psi _{2}X_{2}Y_{3}-\psi
_{3}X_{3}Y_{2}-\psi _{4}X_{3}Y_{1}\right) , \\
^{SF}q_{2}^{-} &=&\frac{1}{X_{3}}\left[ \overline{\psi }_{3}\left( \partial
_{-}X_{2}-X_{1}F_{-}\right) -\overline{\psi }_{4}\left( \partial
_{-}X_{1}+X_{2}F_{-}\right) \right] +\frac{1}{Y_{3}}\left[ -\overline{\psi }%
_{1}\left( \partial _{-}Y_{1}-Y_{2}F_{-}\right) +\overline{\psi }_{2}\left(
\partial _{-}Y_{2}+Y_{1}F_{-}\right) \right] + \\
&&+\mu \left( \psi _{1}X_{2}Y_{3}-\psi _{2}X_{1}Y_{3}+\psi
_{3}X_{3}Y_{1}-\psi _{4}X_{3}Y_{2}\right) , \\
^{SF}q_{3}^{-} &=&\frac{1}{X_{3}}\left[ -\overline{\psi }_{1}\left( \partial
_{-}X_{1}+X_{2}F_{-}\right) +\overline{\psi }_{2}\left( \partial
_{-}X_{2}-X_{1}F_{-}\right) \right] +\frac{1}{Y_{3}}\left[ -\overline{\psi }%
_{3}\left( \partial _{-}Y_{2}+Y_{1}F_{-}\right) +\overline{\psi }_{4}\left(
\partial _{-}Y_{1}-Y_{2}F_{-}\right) \right] + \\
&&+\mu \left( -\psi _{1}X_{3}Y_{2}+\psi _{2}X_{3}Y_{1}+\psi
_{3}X_{1}Y_{3}-\psi _{4}X_{2}Y_{3}\right) , \\
^{SF}q_{4}^{-} &=&\frac{1}{X_{3}}\left[ -\overline{\psi }_{1}\left( \partial
_{-}X_{2}-X_{1}F_{-}\right) -\overline{\psi }_{2}\left( \partial
_{-}X_{1}+X_{2}F_{-}\right) \right] +\frac{1}{Y_{3}}\left[ -\overline{\psi }%
_{3}\left( \partial _{-}Y_{1}-Y_{2}F_{-}\right) -\overline{\psi }_{4}\left(
\partial _{-}Y_{2}+Y_{1}F_{-}\right) \right] + \\
&&+\mu \left( -\psi _{1}X_{3}Y_{1}-\psi _{2}X_{3}Y_{2}+\psi
_{3}X_{2}Y_{3}+\psi _{4}X_{1}Y_{3}\right) ,
\end{eqnarray*}%
are mixings of the two sub-sets of charges (\ref{non-compact supercharges})
and (\ref{compact supercharges}). From the kernel superalgebra (\ref{psuxpsu
anticom}) we expect the reduced model to have an extended $(4,4)$ 2D
supersymmetry algebra with supercharges transforming under the $U(1)\times
U(1)$ gauge group. Of course, this is an on-shell symmetry as shown above.
We leave for the future the study of the field variations associated to the
non-local charges (\ref{AKNS supercharges}) and the possible invariance of
the gauge fixed action functional (\ref{NA Toda action}) under them.

As we have shown there are two corners in field space in which we have two
submodels (loosely) related to the $(2,2)$ supersymmetric models (\ref{N=2
sinh-Gordon prime}) and (\ref{N=2 sine-Gordon prime}). Possibly, these
models are not related to the $AdS_{3}\times S^{3}$ superstring because the
reduction procedure requires the whole gauge symmetry group $H$ to be local
and not only a part of it. \ However, it would be interesting to see if it
is possible to obtain these submodels as particular limits of the radius of
the coset $F/G,$ i.e the constant $\kappa $ in the GS action (\ref{GS
lagrangian}). In any case, it seems to be that in the reduction process
these two submodels are entangled in the form (\ref{final form ads3xs3})
with a (possibly non-local) extended $(4,4)$ supersymmetry. This is argument
is based solely in the number of non-local supercharges constructed above
and, of course, a deeper study have to be done\footnote{%
I want to thank Arkady Tseytlin for pointing me several subtleties related
to the hyperkahlerianity and reducibility of the target manifold in relation
to (4,4) SUSY.}.

\section{Concluding remarks.}

We have provided substantial evidence that the conjectured existence of the
world-sheet supersymmetry of the Pohlmeyer reduced models is of the extended
type and generated by the kernel loop superalgebra $\mathcal{K}$ constructed
out of the subalgebra $\mathfrak{f}^{\perp }\subset \mathfrak{f}$. Of
course, on-shell $\mathcal{K}$ is a true algebra of symmetries which leave
invariant the equations of motion and one of the most difficult issues yet
to be solved, is to show the invariance of the gauge fixed SSSSG model
action functional under an appropriate set of residual symmetry variations $%
\delta _{\mathcal{K}}.$ For the moment, the supersymmetries just constructed
are non manifest at the lagrangian level and the problem of making them
manifest as well as the full Poisson form of the supersymmetry algebra in
terms of the monodromy matrix will be addressed in the near future \cite{us}%
. If the action is SUSY invariant, perhaps we can try to use localization
techniques to handle the partition function and to study its properties, at
least in the conformal $AdS_{5}\times S^{5}$ superstring in which the
Pohlmeyer reduction have a chance to survive the quantization \cite{tseytlin
UV}. Another important problem to be studied is the construction and
quantization of soliton solutions involving the fermionic fields for the set
of semi-symmetric superspaces $F/G$ associated to several Lie superalgebras $%
\mathfrak{f}$ admiting a $%
\mathbb{Z}
_{4}$ decomposition. Fortunately, this can be done by extending to the
supersymmetric case, the results recently presented in \cite{class-quant
solitons} for the bosonic symmetric-space sine-Gordon models.

\paragraph{Acknowledgements.}

I would like to thank the IGFAE (Santiago de Compostela, Spain) for the kind
hospitality extended while this work was in final preparation. I want to
thank Alexis R. Aguirre, Jos\'{e} F. Gomes, Timothy J. Hollowood, J. Luis
Miramontes, Arkady Tseytlin and Abraham H. Zimerman for many useful
comments, advices, suggestions, discussions and correspondence. Special
thanks to J. Luis Miramontes for making available a preliminary version of
the forthcoming paper \cite{SSSSG AdS(5)xS(5)}. This research is supported
by a CNPq junior postdoctoral grant.

\section{Appendix A: The superalgebra $psu(1,1\mid 2)$}

Consider the following distinguished Dynkin diagram of the superalgebra $%
sl(2\mid 2)_{%
\mathbb{C}
}$%
\begin{equation*}
\underset{\alpha _{1}}{\bigcirc }\underset{\alpha _{2}}{-\otimes -}\underset{%
\alpha _{3}}{\bigcirc },
\end{equation*}%
where $\alpha _{1},\alpha _{3}$ are the bosonic simple roots and $\alpha
_{2} $ is the fermionic simple root. Introduce the step operators $E_{\pm
\alpha _{1}},$ $E_{\pm \alpha _{2}},$ $E_{\pm \alpha _{3}},$ $E_{\pm \alpha
_{1}\pm \alpha _{2}},$ $E_{\pm \alpha _{2}\pm \alpha _{3}},$ $E_{\pm \alpha
_{1}\pm \alpha _{2}\pm \alpha _{3}},$ the Cartan elements $h_{1}$, $h_{2}$, $%
h_{3},$ $I$ and work in a $4\times 4$ supermatrix representation$.$ Then, by
introducing the matrices $\left( E_{ab}\right) _{ij}=\delta _{ai}\delta
_{bj},$ $a,b,i,j=1,2,3,4,$ we have%
\begin{eqnarray*}
E_{\alpha _{1}} &=&E_{12},\text{ \ \ }E_{\alpha _{2}}\text{ }=\text{ }E_{23},%
\text{ \ \ }E_{\alpha _{3}}\text{ }=\text{ }E_{34},\text{ \ \ }E_{\alpha
_{1}+\alpha _{2}}\text{ }=\text{ }E_{13},\text{ \ \ }E_{\alpha _{2}+\alpha
_{3}}\text{ }=\text{ }E_{24},\text{ \ \ }E_{\alpha _{1}+\alpha _{2}+\alpha
_{3}}\text{ }=\text{ }E_{14}, \\
h_{1} &=&E_{11}-E_{22},\text{ \ \ }h_{2}\text{ }=\text{ }E_{22}+E_{33},\text{
\ \ \ }h_{3}\text{ }=\text{ }E_{33}-E_{44},
\end{eqnarray*}%
where the matrices corresponding to the negative roots are represented by
the transpose of the matrix corresponding to the positive root, e.g $%
E_{-\alpha _{1}}=\left( E_{\alpha _{1}}\right) ^{t}$. We introduce these
matrices in order to write the base of the $su(1,1\mid 2)$ superalgebra in
terms of them.

\subsection{$%
\mathbb{Z}
_{4}$ grading and $\mathfrak{f}^{\perp },\mathfrak{f}^{\parallel }$
decomposition of $psu(1,1\mid 2).$}

An element $M\subset sl(2\mid 2)_{%
\mathbb{C}
}$ can be represented by a $4\times 4$ supermatrix%
\begin{equation*}
M=%
\begin{pmatrix}
A & X \\ 
Y & B%
\end{pmatrix}%
,
\end{equation*}%
where $A,B$ are (even) complex $2\times 2$ matrices and $X,Y$ are (odd)
complex $2\times 2$ matrices. Introduce the following $4\times 4$ matrices%
\begin{equation*}
\Sigma =%
\begin{pmatrix}
\sigma & 0 \\ 
0 & I%
\end{pmatrix}%
\text{, \ \ \ \ \ }K=%
\begin{pmatrix}
\sigma & 0 \\ 
0 & \sigma%
\end{pmatrix}%
,\text{ \ \ \ \ \ }\sigma =%
\begin{pmatrix}
1 & 0 \\ 
0 & -1%
\end{pmatrix}%
\end{equation*}%
\ and recall the definition of super-transposition and super-Hermitian
conjugation 
\begin{equation*}
M^{st}=%
\begin{pmatrix}
A^{t} & -Y^{t} \\ 
X^{t} & B^{t}%
\end{pmatrix}%
,\text{ \ \ \ \ \ }M^{\dagger }=%
\begin{pmatrix}
A^{\dagger } & -iY^{\dagger } \\ 
-iX^{\dagger } & B^{\dagger }%
\end{pmatrix}%
.
\end{equation*}

The $\mathfrak{f=}$ $psu(1,1\mid 2)$ superalgebra is a real form of $%
psl(2\mid 2)_{%
\mathbb{C}
}$ and can be represented by $4\times 4$ supermatrices modulo the identity
matrix. It is defined by $M^{\ast }=-M$ where $M^{\ast }=\Sigma M\Sigma .$
The $%
\mathbb{Z}
_{4}$ decomposition of $\mathfrak{f}$ is implemented by the action of the
following automorphism $M^{\Omega }=-KM^{st}K,$ which allows the splitting
of $\mathfrak{f}$ as a direct sum of subspaces $\mathfrak{f=f}_{0}\oplus 
\mathfrak{f}_{1}\oplus \mathfrak{f}_{2}\oplus \mathfrak{f}_{3}$ where each $%
\mathfrak{f}_{j}$ is the eigen-space of $\Omega $ with eigenvalue $j$ i.e
for $M\in $ $\mathfrak{f}_{j}$ we have $M^{\Omega }=i^{j}M$. We also have
that $[\mathfrak{f}_{i},\mathfrak{f}_{j}]\subset \mathfrak{f}_{(i+j)\func{mod%
}4}.$

Imposing $M^{\ast }=-M$ we get $A^{\dagger }=-\sigma A\sigma ,$ $B^{\dagger
}=-B$ and $Y^{\dagger }=-i\sigma X$ defining respectively the $su(1,1)$ and $%
su(2)$ Lie algebras and reducing the number of odd elements. Using $%
M^{\Omega }=i^{j}M$ we find that $\mathfrak{f}_{j}$ is formed by the
supermatrices obeying 
\begin{eqnarray*}
\mathfrak{f}_{0} &:&\text{ }A_{0}=-\sigma A_{0}^{t}\sigma \text{, \ \ }%
B_{0}=-\sigma B_{0}^{t}\sigma ,\text{ \ \ \ \ \ \ \ }\mathfrak{f}_{2}\text{ }%
:\text{ }A_{2}=\sigma A_{2}^{t}\sigma \text{, \ \ }B_{2}=\sigma
B_{2}^{t}\sigma , \\
\mathfrak{f}_{1} &:&\text{ }Y=i\sigma X^{t}\sigma \text{, \ \ \ \ \ \ \ \ \
\ \ \ \ \ \ \ \ \ \ \ \ \ \ \ \ \ \ \ \ \ \ }\mathfrak{f}_{3}\text{ }:\text{
\ }Y=-i\sigma X^{t}\sigma \text{\ .}
\end{eqnarray*}%
Written in terms of step operators we get%
\begin{eqnarray*}
\mathfrak{f}_{0} &:&\text{ }\left\{ 
\begin{array}{c}
f_{01}=(E_{\alpha _{1}}+E_{-\alpha _{1}}) \\ 
f_{02}=i(E_{\alpha _{3}}+E_{-\alpha _{3}})%
\end{array}%
\right\} ,\text{ \ \ \ \ \ \ \ \ \ \ \ \ \ \ \ \ \ \ }\mathfrak{f}_{2}:\text{
}\left\{ 
\begin{array}{c}
f_{21}=ih_{1} \\ 
f_{22}=i(E_{\alpha _{1}}-E_{-\alpha _{1}}) \\ 
f_{23}=ih_{3} \\ 
f_{24}=\left( E_{\alpha _{3}}-E_{-\alpha _{3}}\right)%
\end{array}%
\right\} , \\
\mathfrak{f}_{1} &:&\left\{ 
\begin{array}{c}
f_{11}=(E_{\alpha _{1}+\alpha _{2}}+iE_{-\alpha _{1}-\alpha _{2}}) \\ 
f_{12}=(iE_{\alpha _{1}+\alpha _{2}+\alpha _{3}}+E_{-\alpha _{1}-\alpha
_{2}-\alpha _{3}}) \\ 
f_{13}=(E_{\alpha _{2}}-iE_{-\alpha _{2}}) \\ 
f_{14}=(iE_{\alpha _{2}+\alpha _{3}}-E_{-\alpha _{2}-\alpha _{3}})%
\end{array}%
\right\} ,\text{ \ }\mathfrak{f}_{3}:\text{ }\left\{ 
\begin{array}{c}
f_{31}=(iE_{\alpha _{1}+\alpha _{2}}+E_{-\alpha _{1}-\alpha _{2}}) \\ 
f_{32}=(E_{\alpha _{1}+\alpha _{2}+\alpha _{3}}+iE_{-\alpha _{1}-\alpha
_{2}-\alpha _{3}}) \\ 
f_{33}=(iE_{\alpha _{2}}-E_{-\alpha _{2}}) \\ 
f_{34}=(E_{\alpha _{2}+\alpha _{3}}-iE_{-\alpha _{2}-\alpha _{3}})%
\end{array}%
\right\} ,
\end{eqnarray*}%
where in the notation $f_{ia}$ the index $i=0,...,4$ stands for the $%
\mathbb{Z}
_{4}$ eigenvalue $i$ of the subspace $\mathfrak{f}_{i}$ and the index $%
a=1,...,\dim \mathfrak{f}_{i}$ is just a basis label.

By defining the semisimple element $\Lambda =\frac{1}{2}\left(
f_{21}+f_{23}\right) \in \mathfrak{a}_{2},$ we obtain the decomposition of
the superalgebra in the $%
\mathbb{Z}
_{2}$ form $\mathfrak{f=f}^{\perp }\oplus \mathfrak{f}^{\parallel }$, where $%
\mathfrak{f}^{\perp }\mathcal{=}\ker \left( ad\left( \Lambda \right) \right)
,$ $\mathfrak{f}^{\parallel }\mathcal{=}\func{Im}\left( ad\left( \Lambda
\right) \right) \ $and where every subspace $\mathfrak{f}_{i}$ is also
decomposed as $\mathfrak{f}_{i}=\mathfrak{f}_{i}^{\mathcal{\perp }}+%
\mathfrak{f}_{i}^{\mathcal{\parallel }}.$ We find%
\begin{equation}
\mathfrak{f}_{0}^{\perp }:\left\{ \varnothing \right\} \rightarrow \text{ 
\textit{No gauge flows},}  \label{psu gauge group}
\end{equation}%
\begin{equation*}
\mathfrak{f}_{0}^{\parallel }:\left\{ 
\begin{array}{c}
M_{01}=f_{01} \\ 
M_{02}=f_{02}%
\end{array}%
\right\} \rightarrow \text{ \textit{2 bosonic fields},}
\end{equation*}%
\begin{equation}
\mathfrak{f}_{1}^{\perp }:\left\{ 
\begin{array}{c}
F_{11}=f_{11} \\ 
F_{12}=f_{14}%
\end{array}%
\right\} \rightarrow \text{ \textit{2 SUSY flows},}  \label{psu susy flows I}
\end{equation}%
\begin{equation*}
\mathfrak{f}_{1}^{\parallel }:\left\{ 
\begin{array}{c}
G_{11}=f_{12} \\ 
G_{12}=f_{13}%
\end{array}%
\right\} \rightarrow \text{ \textit{2 fermionic fields},}
\end{equation*}%
\begin{equation}
\mathfrak{f}_{2}^{\perp }:\left\{ 
\begin{array}{c}
K_{21}=f_{21} \\ 
K_{22}=f_{23}%
\end{array}%
\right\} \rightarrow \Lambda =\frac{1}{2}\left( K_{21}+K_{22}\right) ,
\label{T for psu}
\end{equation}%
\begin{equation*}
\mathfrak{f}_{2}^{\parallel }:\left\{ 
\begin{array}{c}
M_{21}=f_{22} \\ 
M_{22}=f_{24}%
\end{array}%
\right\} \rightarrow \text{ \textit{No dynamics,}}
\end{equation*}%
\begin{equation}
\mathfrak{f}_{3}^{\perp }:\left\{ 
\begin{array}{c}
F_{31}=f_{31} \\ 
F_{32}=f_{34}%
\end{array}%
\right\} \rightarrow \text{ \textit{2 SUSY flows},}
\label{psu susy flows II}
\end{equation}%
\begin{equation*}
\mathfrak{f}_{3}^{\parallel }:\left\{ 
\begin{array}{c}
G_{31}=f_{32} \\ 
G_{32}=f_{33}%
\end{array}%
\right\} \rightarrow \text{ \textit{2 fermionic fields.}}
\end{equation*}

From the loop superalgebra setting (\ref{loop algebra}), (\ref{half-integer
expansion}), we get the basis elements in $\mathcal{K=}$ $\widehat{\mathfrak{%
f}}^{\perp }$ and $\mathcal{M=}$ $\widehat{\mathfrak{f}}^{\parallel }$%
\begin{eqnarray}
F_{i}^{(-1/2)} &=&z^{-1/2}F_{3i}\in \mathcal{K}_{F}^{(-1/2)},\text{ \ \ \ \
\ }F_{i}^{(+1/2)}\text{ }=\text{ }z^{+1/2}F_{1i}\in \mathcal{K}_{F}^{(+1/2)},
\label{psu loop basis} \\
G_{i}^{(-1/2)} &=&z^{-1/2}G_{3i}\in \mathcal{M}_{F}^{(-1/2)},\text{ \ \ }%
M_{a}^{(0)}\text{ }=\text{ }z^{0}M_{0a}\in \mathcal{M}_{B}^{(0)},\text{ \ \ }%
G_{i}^{(+1/2)}\text{ }=\text{ }z^{+1/2}G_{1i}\in \mathcal{M}_{F}^{(+1/2)}. 
\notag
\end{eqnarray}

The kernel algebra $\mathcal{K}$ is%
\begin{equation}
\left\{ F_{i}^{(\pm 1/2)},F_{j}^{(\pm 1/2)}\right\} =\pm 2\delta
_{ij}\Lambda _{\pm }^{(\pm 1)},\text{ \ \ \ \ \ }\left\{
F_{i}^{(+1/2)},F_{j}^{(-1/2)}\right\} =0.  \label{psu anticom}
\end{equation}%
In this case the algebra $\mathcal{K}$ is related through $\delta _{\mathcal{%
K}}$ to the usual $(2,2)$ supersymmetric extension of the 2D Poincar\'{e}
algebra.

\subsection{$%
\mathbb{Z}
_{4}$ grading and $\mathfrak{f}^{\perp },\mathfrak{f}^{\parallel }$
decomposition of $psu(1,1\mid 2)^{\times 2}.$}

An element $X\subset sl(2\mid 2)_{%
\mathbb{C}
}^{\times 2}$ can be represented by a $8\times 8$ supermatrix%
\begin{equation*}
X=%
\begin{pmatrix}
M & 0 \\ 
0 & N%
\end{pmatrix}%
,\text{ \ \ \ \ \ }M=%
\begin{pmatrix}
A & X \\ 
Y & B%
\end{pmatrix}%
,\text{ \ \ \ \ \ }N=%
\begin{pmatrix}
C & Z \\ 
W & D%
\end{pmatrix}%
,\text{\ }
\end{equation*}%
where $M$ and $N$ are $4\times 4$ elements of $sl(2\mid 2)_{%
\mathbb{C}
}.$ The $\mathfrak{f=}psu(1,1\mid 2)^{\times 2}$ superalgebra is a real form
of $psl(2\mid 2)_{%
\mathbb{C}
}^{\times 2}$ and can be represented by $8\times 8$ supermatrices modulo the
identity matrix. It is defined by $M^{\ast }=-M$ and $N=-N^{\ast },$ with *
as above, then the only subtle point relies in the $%
\mathbb{Z}
_{4}$ decomposition \cite{tseytlin II}. Introduce the following $8\times 8$
matrices%
\begin{equation*}
K=%
\begin{pmatrix}
0 & k \\ 
k & 0%
\end{pmatrix}%
,\text{ \ \ \ \ \ }k=%
\begin{pmatrix}
I & 0 \\ 
0 & I%
\end{pmatrix}%
.
\end{equation*}

The $%
\mathbb{Z}
_{4}$ decomposition of $\mathfrak{f}$ is implemented by the action of the
following automorphism $M^{\Omega }=-KM^{st}K.$ Using $M^{\Omega }=i^{j}M$
we find that $\mathfrak{f}_{j}$ is formed by the supermatrices $X\in 
\mathfrak{f}$ obeying%
\begin{eqnarray*}
\mathfrak{f}_{0} &:&\text{ }%
\begin{pmatrix}
A & 0 & 0 & 0 \\ 
0 & B & 0 & 0 \\ 
0 & 0 & -A^{t} & 0 \\ 
0 & 0 & 0 & -B^{t}%
\end{pmatrix}%
,\text{ \ \ }\mathfrak{f}_{1}\text{ }:\text{ }%
\begin{pmatrix}
0 & X & 0 & 0 \\ 
Y & 0 & 0 & 0 \\ 
0 & 0 & 0 & -iY^{t} \\ 
0 & 0 & iX^{t} & 0%
\end{pmatrix}%
, \\
\mathfrak{f}_{2} &:&\text{ }%
\begin{pmatrix}
A & 0 & 0 & 0 \\ 
0 & B & 0 & 0 \\ 
0 & 0 & A^{t} & 0 \\ 
0 & 0 & 0 & B^{t}%
\end{pmatrix}%
,\text{ \ \ \ \ \ \ \ }\mathfrak{f}_{3}\text{ }:\text{ }%
\begin{pmatrix}
0 & X & 0 & 0 \\ 
Y & 0 & 0 & 0 \\ 
0 & 0 & 0 & iY^{t} \\ 
0 & 0 & -iX^{t} & 0%
\end{pmatrix}%
,
\end{eqnarray*}%
where $Y=iX^{\dagger }\Sigma $ and $\Sigma $ is defined above. It is enough
to consider only the first copy of $psu(1,1\mid 2)$ inside $\mathfrak{f}$
and we can return to the $4\times 4$ supermatrix representation of $%
psu(1,1\mid 2)$ used above. Note that now we have $\mathfrak{f}_{0}=%
\mathfrak{f}_{2}=su(1,1)\times su(2).$

By taking the semisimple element $\Lambda =\frac{1}{2}\left(
f_{21}+f_{23}\right) \in \mathfrak{a}_{2},$ we obtain the $%
\mathbb{Z}
_{2}$ decomposition $\mathfrak{f=f}^{\perp }\oplus \mathfrak{f}^{\parallel }$
of the superalgebra and every subspace $\mathfrak{f}_{i}$ is decomposed as $%
\mathfrak{f}_{i}=\mathfrak{f}_{i}^{\mathcal{\perp }}+\mathfrak{f}_{i}^{%
\mathcal{\parallel }},$ with

\begin{equation}
\mathfrak{f}_{0}^{\mathcal{\perp }}:\text{ }\left\{ 
\begin{array}{c}
K_{01}=ih_{1} \\ 
K_{02}=ih_{3}%
\end{array}%
\right\} \rightarrow \text{ }U(1)\times U(1)\text{ \textit{gauge flows,}}
\label{gauge group for psuxpsu}
\end{equation}%
\begin{equation}
\mathfrak{f}_{0}^{\mathcal{\parallel }}:\text{ }\left\{ 
\begin{array}{c}
M_{01}=(E_{\alpha _{1}}+E_{-\alpha _{1}}) \\ 
M_{02}=i(E_{\alpha _{1}}-E_{-\alpha _{1}}) \\ 
M_{03}=(E_{\alpha _{3}}-E_{-\alpha _{3}}) \\ 
M_{04}=i(E_{\alpha _{3}}+E_{-\alpha _{3}})%
\end{array}%
\right\} \rightarrow \text{ \textit{4 bosonic fields,}}
\label{psuxpsu susy flows I}
\end{equation}%
\begin{equation*}
\mathfrak{f}_{1}^{\mathcal{\perp }}:\left\{ 
\begin{array}{c}
F_{11}=(E_{\alpha _{1}+\alpha _{2}}+iE_{-\alpha _{1}-\alpha _{2}}) \\ 
F_{12}=(iE_{\alpha _{1}+\alpha _{2}}+E_{-\alpha _{1}-\alpha _{2}}) \\ 
F_{13}=(E_{\alpha _{2}+\alpha _{3}}-iE_{-\alpha _{2}-\alpha _{3}}) \\ 
F_{14}=(iE_{\alpha _{2}+\alpha _{3}}-E_{-\alpha _{2}-\alpha _{3}})%
\end{array}%
\right\} \rightarrow \text{ \textit{\ 4 SUSY flows,}}
\end{equation*}%
\begin{equation*}
\mathfrak{f}_{1}^{\mathcal{\parallel }}:\left\{ 
\begin{array}{c}
G_{11}=(E_{\alpha _{1}+\alpha _{2}+\alpha _{3}}+iE_{-\alpha _{1}-\alpha
_{2}-\alpha _{3}}) \\ 
G_{12}=(iE_{\alpha _{1}+\alpha _{2}+\alpha _{3}}+E_{-\alpha _{1}-\alpha
_{2}-\alpha _{3}}) \\ 
G_{13}=(E_{\alpha _{2}}-iE_{-\alpha _{2}}) \\ 
G_{14}=(iE_{\alpha _{2}}-E_{-\alpha _{2}})%
\end{array}%
\right\} \rightarrow \text{ \textit{4 fermionic fields,}}
\end{equation*}%
\begin{equation}
\mathfrak{f}_{2}^{\mathcal{\perp }}:\text{ }\left\{ 
\begin{array}{c}
K_{21}=K_{01} \\ 
K_{22}=K_{02}%
\end{array}%
\right\} \rightarrow \text{ }\Lambda =\frac{1}{2}\left( K_{21}+K_{22}\right)
,  \label{T for psuxpsu}
\end{equation}%
\begin{equation*}
\mathfrak{f}_{2}^{\mathcal{\parallel }}:\text{ }\left\{ 
\begin{array}{c}
M_{21}=M_{01} \\ 
M_{22}=M_{02} \\ 
M_{23}=M_{03} \\ 
M_{24}=M_{04}%
\end{array}%
\right\} \rightarrow \text{ \textit{No dynamics,}}
\end{equation*}%
\begin{equation}
\mathfrak{f}_{3}^{\mathcal{\perp }}:\left\{ 
\begin{array}{c}
F_{31}=F_{11} \\ 
F_{32}=F_{12} \\ 
F_{33}=F_{13} \\ 
F_{34}=F_{14}%
\end{array}%
\right\} \rightarrow \text{ \ \textit{4 SUSY flows,}}
\label{psuxpsu susy flows II}
\end{equation}%
\begin{equation*}
\mathfrak{f}_{3}^{\mathcal{\parallel }}:\left\{ 
\begin{array}{c}
G_{31}=G_{11} \\ 
G_{32}=G_{12} \\ 
G_{33}=G_{13} \\ 
G_{34}=G_{14}%
\end{array}%
\right\} \rightarrow \text{\textit{\ 4 fermionic fields.}}
\end{equation*}

In the loop superalgebra, we have the following basis elements of $\mathcal{K%
}$ and $\mathcal{M}$%
\begin{eqnarray}
F_{i}^{(-1/2)} &=&z^{-1/2}F_{3i}\in \mathcal{K}_{F}^{(-1/2)},\text{ \ \ \ \
\ }K_{b}^{(0)}\text{ }=\text{ }z^{0}K_{0b}\in \mathcal{K}_{B}^{(0)}\text{, \
\ \ \ \ }F_{i}^{(+1/2)}\text{ }=\text{ }z^{+1/2}F_{1i}\in \mathcal{K}%
_{F}^{(+1/2)},  \label{psuxpsu loop basis} \\
G_{i}^{(-1/2)} &=&z^{-1/2}G_{3i}\in \mathcal{M}_{F}^{(-1/2)},\text{ \ \ \ }%
M_{a}^{(0)}\text{ }=\text{ }z^{0}M_{0a}\in \mathcal{M}_{B}^{(0)},\text{ \ \
\ }G_{i}^{(+1/2)}\text{ }=\text{ }z^{+1/2}G_{1i}\in \mathcal{M}_{F}^{(+1/2)}.
\notag
\end{eqnarray}

The kernel algebra $\mathcal{K}$ is 
\begin{eqnarray}
\left\{ F_{i}^{(\pm 1/2)},F_{j}^{(\pm 1/2)}\right\} &=&\pm 2\delta
_{ij}\Lambda _{\pm }^{(\pm 1)},\ \text{\ \ \ \ \ }\left\{
F_{i}^{(+1/2)},F_{j}^{(-1/2)}\right\} \text{ }=\text{ }2\delta _{ij}Z^{(0)},
\label{psuxpsu anticom} \\
\left[ K_{a}^{(0)},F_{i}^{(\pm 1/2)}\right] &=&-(-1)^{a}X_{ij}F_{j}^{(\pm
1/2)},\text{ \ \ \ \ \ }a=1,2,  \notag
\end{eqnarray}%
where $Z^{(0)}=K_{1}^{(0)}+K_{2}^{(0)}$ commutes with everything in $%
\mathcal{K}$ and $\left[ X_{ij}\right] =\left( 
\begin{array}{cc}
\epsilon & 0 \\ 
0 & -\epsilon%
\end{array}%
\right) $ with $\epsilon =\left( 
\begin{array}{cc}
0 & 1 \\ 
-1 & 0%
\end{array}%
\right) .$ In this case the algebra $\mathcal{K}$ is related through $\delta
_{\mathcal{K}}$ to the $N=(4,4)$ supersymmetric extension of the 2D Poincar%
\'{e} algebra with a gauge group $U(1)\times U(1).$

\section{Appendix B: Relevant quantities used in the computations}

Using the definitions of Appendix A, section 6.2 we obtain, respectively,
the following currents, traces and conjugations 
\begin{eqnarray}
\partial _{+}B_{A}B_{A}^{-1} &=&\left( \partial _{+}\chi \cosh ^{2}\phi
\right) K_{1}^{(0)}+\left( \cos \chi \partial _{+}\phi +\partial _{+}\chi
\sin \chi \cosh \phi \sinh \phi \right) M_{1}^{(2)}+  \notag \\
&&+\left( \sin \chi \partial _{+}\phi -\partial _{+}\chi \cos \chi \cosh
\phi \sinh \phi \right) M_{2}^{(0)},  \notag \\
\partial _{+}B_{S}B_{S}^{-1} &=&\left( \partial _{+}\theta \cos ^{2}\varphi
\right) K_{2}^{(0)}+\left( -\sin \theta \partial _{+}\varphi +\partial
_{+}\theta \cos \theta \cos \varphi \sin \varphi \right) M_{3}^{(2)}+  \notag
\\
&&+\left( \cos \theta \partial _{+}\varphi +\partial _{+}\theta \sin \theta
\cos \varphi \sin \varphi \right) M_{4}^{(0)},  \notag \\
B_{S}^{-1}\partial _{-}B_{S} &=&\left( \partial _{-}\chi \cosh ^{2}\phi
\right) K_{1}^{(0)}+\left( \cos \chi \partial _{-}\phi +\partial _{-}\chi
\sin \chi \cosh \phi \sinh \phi \right) M_{1}^{(0)}+  \notag \\
&&+\left( -\sin \chi \partial _{-}\phi +\partial _{-}\chi \cos \chi \cosh
\phi \sinh \phi \right) M_{2}^{(0)},  \notag \\
B_{S}^{-1}\partial _{-}B_{S} &=&\left( \partial _{-}\theta \cos ^{2}\varphi
\right) K_{2}^{(0)}+\left( \sin \theta \partial _{-}\varphi -\partial
_{-}\theta \cos \theta \cos \varphi \sin \varphi \right) M_{3}^{(0)}+
\label{currents} \\
&&+\left( \cos \theta \partial _{-}\varphi +\partial _{-}\theta \sin \theta
\cos \varphi \sin \varphi \right) M_{4}^{(0)}.  \notag
\end{eqnarray}%
\begin{eqnarray}
\left\langle G_{i}^{(+1/2)},BG_{i}^{(-1/2)}B^{-1}\right\rangle &=&2\cos
\varphi \cosh \phi \sin (\theta +\chi ),\text{ \ \ \ \ \ }i=1,2,3,4  \notag
\\
\left\langle G_{1}^{(+1/2)}\left( BG_{2}^{(-1/2)}B^{-1}\right) \right\rangle
&=&-\left\langle G_{2}^{(+1/2)}\left( BG_{1}^{(-1/2)}B^{-1}\right)
\right\rangle =  \notag \\
\left\langle G_{4}^{(+1/2)}\left( BG_{3}^{(-1/2)}B^{-1}\right) \right\rangle
&=&-\left\langle G_{3}^{(+1/2)}\left( BG_{4}^{(-1/2)}B^{-1}\right)
\right\rangle =2\cos \varphi \cosh \phi \cos (\theta +\chi ),  \notag \\
-\left\langle G_{1}^{(+1/2)}\left( BG_{3}^{(-1/2)}B^{-1}\right)
\right\rangle &=&\left\langle G_{3}^{(+1/2)}\left(
BG_{1}^{(-1/2)}B^{-1}\right) \right\rangle =  \notag \\
-\left\langle G_{2}^{(+1/2)}\left( BG_{4}^{(-1/2)}B^{-1}\right)
\right\rangle &=&\left\langle G_{4}^{(+1/2)}\left(
BG_{2}^{(-1/2)}B^{-1}\right) \right\rangle =2\sin \varphi \sinh \phi , 
\notag \\
\left\langle G_{1}^{(+1/2)}\left( BG_{4}^{(-1/2)}B^{-1}\right) \right\rangle
&=&\left\langle G_{2}^{(+1/2)}\left( BG_{3}^{(-1/2)}B^{-1}\right)
\right\rangle =  \notag \\
\left\langle G_{4}^{(+1/2)}\left( BG_{1}^{(-1/2)}B^{-1}\right) \right\rangle
&=&\left\langle G_{3}^{(+1/2)}\left( BG_{2}^{(-1/2)}B^{-1}\right)
\right\rangle =0.\text{\ \ \ \ \ \ \ \ \ \ \ \ \ \ \ \ \ \ \ \ \ \ \ \ \ \ \
\ \ \ \ }  \label{traces}
\end{eqnarray}%
\begin{eqnarray}
\left( BK_{1}^{(0)}B^{-1}\right) ^{\parallel } &=&\sin \chi \sinh 2\phi
M_{1}^{(0)}-\cos \chi \sinh 2\phi M_{2}^{(0)},  \label{conju} \\
\left( BK_{2}^{(0)}B^{-1}\right) ^{\parallel } &=&\cos \theta \sin 2\varphi
M_{3}^{(0)}+\sin \theta \sin 2\varphi M_{4}^{(0)},  \notag \\
\left( B^{-1}K_{1}^{(0)}B\right) ^{\parallel } &=&\sin \chi \sinh 2\phi
M_{1}^{(0)}+\cos \chi \sinh 2\phi M_{2}^{(0)},  \notag \\
\left( B^{-1}K_{1}^{(0)}B\right) ^{\parallel } &=&-\cos \theta \sin 2\varphi
M_{3}^{(0)}+\sin \theta \sin 2\varphi M_{4}^{(0)},  \notag
\end{eqnarray}

\begin{eqnarray}
\left( BG_{1}^{(-1/2)}B^{-1}\right) ^{\perp }
&=&X_{2}Y_{3}F_{1}^{(-1/2)}-X_{1}Y_{3}F_{2}^{(-1/2)}+X_{3}Y_{1}F_{3}^{(-1/2)}+X_{3}Y_{2}F_{4}^{(-1/2)},
\notag \\
\left( BG_{2}^{(-1/2)}B^{-1}\right) ^{\perp }
&=&X_{1}Y_{3}F_{1}^{(-1/2)}+X_{2}Y_{3}F_{2}^{(-1/2)}-X_{3}Y_{2}F_{3}^{(-1/2)}+X_{3}Y_{1}F_{4}^{(-1/2)},
\notag \\
\left( BG_{3}^{(-1/2)}B^{-1}\right) ^{\perp }
&=&X_{3}Y_{1}F_{1}^{(-1/2)}-X_{3}Y_{2}F_{2}^{(-1/2)}-X_{2}Y_{3}F_{3}^{(-1/2)}-X_{1}Y_{3}F_{4}^{(-1/2)},
\notag \\
\left( BG_{4}^{(-1/2)}B^{-1}\right) ^{\perp }
&=&X_{3}Y_{2}F_{1}^{(-1/2)}+X_{3}Y_{1}F_{2}^{(-1/2)}+X_{1}Y_{3}F_{3}^{(-1/2)}-X_{2}Y_{3}F_{4}^{(-1/2)},
\notag \\
\left( B^{-1}G_{1}^{(+1/2)}B\right) ^{\perp }
&=&X_{2}Y_{3}F_{1}^{(+1/2)}+X_{1}Y_{3}F_{2}^{(+1/2)}-X_{3}Y_{1}F_{3}^{(+1/2)}+X_{3}Y_{2}F_{4}^{(+1/2)},
\notag \\
\left( B^{-1}G_{2}^{(+1/2)}B\right) ^{\perp }
&=&-X_{1}Y_{3}F_{1}^{(+1/2)}+X_{2}Y_{3}F_{2}^{(+1/2)}-X_{3}Y_{2}F_{3}^{(+1/2)}-X_{3}Y_{1}F_{4}^{(+1/2)},
\notag \\
\left( B^{-1}G_{3}^{(+1/2)}B\right) ^{\perp }
&=&-X_{3}Y_{1}F_{1}^{(+1/2)}-X_{3}Y_{2}F_{2}^{(+1/2)}-X_{2}Y_{3}F_{3}^{(+1/2)}+X_{1}Y_{3}F_{4}^{(+1/2)},
\notag \\
\left( B^{-1}G_{4}^{(+1/2)}B\right) ^{\perp }
&=&X_{3}Y_{2}F_{1}^{(+1/2)}-X_{3}Y_{1}F_{2}^{(+1/2)}-X_{1}Y_{3}F_{3}^{(+1/2)}-X_{2}Y_{3}F_{4}^{(+1/2)},
\label{conjugations}
\end{eqnarray}%
where $X_{i},Y_{i},$ $i=1,2,3$ were defined in (\ref{X y Y}).

\end{document}